\documentclass[amsthm]{elsart}
\usepackage{yjsco}
\usepackage{natbib}

\usepackage{enumerate}
\usepackage{amsmath,amsfonts,amssymb,amsthm}

\usepackage[linkcolor=black,colorlinks=true,citecolor=blue,urlcolor=blue]{hyperref} 

\defcitealias{BecLab00}{ibid.}
\defcitealias{vzGathen13}{ibid.}

\makeatletter \let\cl@part\relax \makeatother 
\usepackage{cleveref}
\Crefname{subsec}{Subsection}{Subsections}
\crefname{subsec}{subsection}{subsections}

\newcommand{\gathen}[1]{#1}
\renewcommand\leq\leqslant
\newcommand{\storeArg}{} 
\newcounter{notationCounter}
\newcommand{\bigO}[1]{\mathcal{O}(#1)} 
\newcommand{\bigOPar}[1]{\mathcal{O}\left(#1\right)} 
\newcommand{\softO}[1]{\mathcal{O}\tilde{~}(#1)} 
\newcommand{\polmultime}[1]{\mathsf{M}(#1)}
\newcommand{\polmultimePar}[1]{\mathsf{M}{#1}}
\newcommand{\polmatmultime}[1]{\mathsf{MM}(#1)}
\newcommand{\polmatmultimePrime}[1]{\mathsf{MM'}(#1)}
\newcommand{\polmatmultimeBis}[1]{\mathsf{MM''}(#1)}
\newcommand{\polmatmultimePrimeDnc}[1]{\mathsf{\overline{MM'}}(#1)}
\newcommand{\polmatmultimeBisDnc}[1]{\mathsf{\overline{MM''}}(#1)}
\newcommand{\expmatmul}{\omega} 
\newcommand{\algoname}[1]{{\normalfont\textsc{#1}}}
\newcommand{\algoword}[1]{\textsf{#1}}
\renewcommand{\ge}{\geqslant} 
\renewcommand{\le}{\leqslant} 
\newcommand{\ZZ}{\mathbb{Z}} 
\newcommand{\NN}{\mathbb{N}} 
\newcommand{\NNp}{\mathbb{N}_{> 0}} 
\newcommand{\var}{X} 
\newcommand{\field}{\mathbb{K}} 
\newcommand{\vecSpace}{\mathfrak{E}} 
\newcommand{\polSpace}{\mathfrak{F}} 
\newcommand{\polRing}{\field[\var]} 
\newcommand{\matSpace}[1][\rdim]{\renewcommand\storeArg{#1}\matSpaceAux} 
\newcommand{\polMatSpace}[1][\rdim]{\renewcommand\storeArg{#1}\polMatSpaceAux} 
\newcommand{\matSpaceAux}[1][\storeArg]{\field^{\storeArg \times #1}} 
\newcommand{\polMatSpaceAux}[1][\storeArg]{\polRing^{\storeArg \times #1}} 
\newcommand{\row}[1]{\mathbf{\MakeLowercase{#1}}} 
\newcommand{\mat}[1]{\mathbf{\MakeUppercase{#1}}} 
\newcommand{\trsp}[1]{#1^\mathsf{T}} 
\newcommand{\rdim}{m} 
\newcommand{\cdim}{n} 
\newcommand{\smat}[1][1]{\setcounter{notationCounter}{#1}\mat{\Alph{notationCounter}}} 
\newcommand{\spolmat}[1][1]{\setcounter{notationCounter}{#1}\mathbf{\Alph{notationCounter}}} 
\newcommand{\matrow}[2]{{#1}_{#2,*}} 
\newcommand{\matcol}[2]{{#1}_{*,#2}} 
\newcommand{\idMat}[1][\rdim]{\mat{I}_{#1}} 
\newcommand{\shiftMat}[1]{\mat{\var}^{#1\,}} 
\newcommand{\shift}[2][s]{#1_{#2}} 
\newcommand{\shifts}[1][s]{\mathbf{#1}} 
\newcommand{\sshifts}[1][\shifts]{|#1|} 
\newcommand{\shiftSpace}[1][\rdim]{\NN^{#1}} 
\newcommand{\unishift}{\mathbf{0}} 
\newcommand{\rdeg}[2][]{\mathrm{rdeg}_{{#1}}(#2)} 
\newcommand{\order}{\sigma} 
\newcommand{\linFunc}{\ell} 
\renewcommand{\int}{\row{p}} 
\newcommand{\intBasis}{\mat{P}} 
\newcommand{\intSpace}[1][\rdim]{\polMatSpace[#1][#1]} 
\newcommand{\evMat}{\mat{E}} 
\newcommand{\evSpace}[2][\rdim]{\matSpace[#1][#2]}
\newcommand{\mul}{\cdot} 
\newcommand{\mulmat}{\jordan} 
\newcommand{\jordan}{\mat{J}} 
\newcommand{\mulshift}{\mat{Z}} 
\newcommand{\nbbl}{n} 
\newcommand{\szbl}{\order} 
\newcommand{\evpt}{x} 
\newcommand{\module}{\mathfrak{M}}
\newcommand{\kernel}{\mathfrak{K}}
\newcommand{\intMod}[1][\evMat,\mulmat]{\mathfrak{I}(#1)}
\newcommand{\rowvec}[1][1]{\setcounter{notationCounter}{#1}\addtocounter{notationCounter}{21} \row{\alph{notationCounter}}} 
\newcommand{\krylov}[2][]{\mathcal{K}_{#1}(#2)} 
\newcommand{\linPolMat}[2][]{\mathcal{E}_{#1}(#2)} 
\newcommand{\polFromLin}[2][]{\mathcal{C}_{#1}(#2)} 
\newcommand{\pivmat}[2][]{\mathcal{P}_{#1}(#2)} 
\newcommand{\tgtmat}[2][]{\mathcal{T}_{#1}(#2)} 
\newcommand{\relmat}[2][]{\mathcal{R}_{#1}(#2)} 
\newcommand{\prioPerm}[1][\shifts]{\mat{\pi}_{#1}}
\newcommand{\prioEval}{\psi_{\shifts}}
\newcommand{\prioIndex}[1][\shifts]{\phi_{#1}}
\newcommand{\minDeg}{\delta}
\newcommand{\maxDeg}{\delta}
\newcommand{\rep}[2]{\mathrm{rep}(#1,#2)}
\newcommand{\evMatF}{\mat{F}} 
\newcommand{\supp}{\mu}
\newcommand{\nbpt}{p}

\newcommand{\expSet}{\Gamma}
\newcommand{\nvars}{r}

\newcommand{\sumVec}[1]{|#1|} 
\theoremstyle{definition}
\newtheorem{algorithm}{Algorithm}
\theoremstyle{plain}
\newtheorem{definition}{Definition}[section]
\newtheorem{theorem}[definition]{Theorem}
\newtheorem{corollary}[definition]{Corollary}
\newtheorem{proposition}[definition]{Proposition}
\newtheorem{lemma}[definition]{Lemma}
\newtheorem{problem}{Problem}
\theoremstyle{remark}
\newtheorem{remark}[definition]{Remark}
\newtheorem{example}[definition]{Example}

\begin{document}

\begin{frontmatter}

\title{Computing minimal interpolation bases}

\author{Claude-Pierre Jeannerod}
\address{Inria, Universit\'e de Lyon \\ Laboratoire LIP (CNRS, Inria, ENS de Lyon, UCBL), Lyon, France}

\author{Vincent Neiger}
\address{ENS de Lyon, Universit\'e de Lyon \\ Laboratoire LIP (CNRS, Inria, ENS de Lyon, UCBL), Lyon, France}

\author{\'Eric Schost}
\address{University of Waterloo \\ David Cheriton School of Computer Science, Waterloo ON, Canada}

\author{Gilles Villard}
\address{ENS de Lyon, Universit\'e de Lyon \\ Laboratoire LIP (CNRS, Inria, ENS de Lyon, UCBL), Lyon, France}

\begin{abstract}
We consider the problem of computing univariate polynomial matrices over a
field that represent minimal solution bases for a general interpolation
problem, some forms of which are the \emph{vector M-Pad\'e approximation
problem} in [Van Barel and Bultheel, Numerical Algorithms 3, 1992] and the
\emph{rational interpolation problem} in [Beckermann and Labahn, SIAM J. Matrix
Anal. Appl. 22, 2000]. Particular instances of this problem include the
bivariate interpolation steps of Guruswami-Sudan hard-decision and
K\"otter-Vardy soft-decision decodings of Reed-Solomon codes, the multivariate
interpolation step of list-decoding of folded Reed-Solomon codes, and
Hermite-Pad\'e approximation.

In the mentioned references, the problem is solved using iterative algorithms
based on recurrence relations. Here, we discuss a fast, divide-and-conquer
version of this recurrence, taking advantage of fast matrix computations over
the scalars and over the polynomials. This new algorithm is deterministic, and
for computing shifted minimal bases of relations between $\rdim$ vectors of
size $\order$ it uses $\softO{\rdim^{\expmatmul-1} (\order+\sshifts)}$ field
operations, where $\expmatmul$ is the exponent of matrix multiplication, and
$\sshifts$ is the sum of the entries of the input shift $\shifts$, with
$\min(\shifts)=0$.  This complexity bound improves in particular on earlier
algorithms in the case of bivariate interpolation for soft decoding, while
matching fastest existing algorithms for simultaneous Hermite-Pad\'e
approximation.
\end{abstract}


\end{frontmatter}
\endNoHyper 

\section{Introduction}
\label{sec:intro}

\subsection{Context} In this paper, we study fast algorithms for
generalizations of classical \emph{Pad\'e approximation} and \emph{polynomial
interpolation} problems. Two typical examples of such problems are the
following.

\medskip
\begin{description}
  \item[\rm\emph{Constrained bivariate interpolation.}] In coding
    theory, some decoding algorithms rely on solving a bivariate interpolation
    problem which may be formulated as follows. Given a set of $\order$ points
    $\{(\evpt_1, y_1), \ldots , (\evpt_{\order}, y_{\order}) \}$ with
    coordinates in a field $\field$, find a non-zero polynomial $Q \in
    \field[X,Y]$ of $Y$-degree less than $\rdim$ satisfying
    \begin{equation*}
      Q(x_1,y_1) = \cdots = Q(x_\order,y_\order) = 0,
    \end{equation*}
    as well as a weighted degree constraint. In terms of linear algebra, we
    interpret this using the $\field$-linear functionals
    $\linFunc_1,\ldots,\linFunc_{\order}$ defined by $\linFunc_j(Q) =
    Q(\evpt_j,y_j)$ for polynomials $Q$ in $\field[X,Y]$. Then, given the
    points, the problem is to find a polynomial $Q$ satisfying the degree
    constraints and such that $\linFunc_j(Q)=0$ for each $j$. Writing $Q =
    \sum_{j<\rdim} p_{j+1}(X) Y^j$, in this context, one may actually want to
    compute a whole basis $\mat{P}$ of such interpolants $\row{p} =
    (p_1,\ldots,p_\rdim)$, and the weighted degree constraint is satisfied
    through the minimization of some suitably defined degree of $\mat{P}$.
    \medskip
  \item[\rm \emph{Hermite-Pad\'e approximation.}] Given a vector of $\rdim$
    polynomials $\row{f} = (f_1,\ldots,f_{\rdim}) \in \polRing^\rdim$, with
    coefficients in a field $\field$, and given a target order $\order$, find another
    vector of polynomials $\row{p} = (p_1,\ldots,p_\rdim)$ such that
    \begin{equation} \label{eq:dfnpade0} p_1 f_1 + \cdots + p_\rdim f_\rdim
      \,=\, 0 \bmod X^{\sigma}, \end{equation}
    with some prescribed degree constraints on $p_1,\ldots,p_\rdim$.

    Here as well, one may actually wish to compute a set of such vectors
    $\row{p}$, forming the rows of a matrix $\mat{P}$ over $\polRing$, which
    describe a whole basis of solutions. Then, these vectors may not all
    satisfy the degree constraints, but by requiring that the basis matrix
    $\mat{P}$ minimizes some suitably defined degree, we will ensure that at
    least one of its rows does (unless the problem has no solution).
\end{description}

\medskip
Concerning Hermite-Pad\'e approximation, a minimal basis of solutions can be
computed in $\softO{\rdim^{\expmatmul-1} \order}$ operations in
$\field$~\citep{ZhoLab12}. Here and hereafter, the soft-O notation
$\softO{\cdot}$ indicates that we omit polylogarithmic terms, and the exponent
$\expmatmul$ is so that we can multiply $\rdim \times \rdim$ matrices in
$\bigO{\rdim^\expmatmul}$ ring operations on any ring, the best known bound
being $\expmatmul < 2.38$~\citep{CopWin90, LeGall14}. For constrained bivariate
interpolation, assuming that $\evpt_1,\ldots,\evpt_\order$ are pairwise
distinct, the best known cost bound for computing a minimal basis is
$\softO{\rdim^\expmatmul \order}$~\citep{Bernstein11,CohHen15,Nielsen14}; the
cost bound $\softO{\rdim^{\expmatmul-1} \order}$ was achieved in~\citep{CJNSV15}
with a probabilistic algorithm which outputs only one interpolant satisfying
the degree constraints.

Following the work of \citet{BarBul92}, \citet{BecLab00},
and McEliece's presentation of K\"otter's algorithm~\citep[Section
7]{McEliece03}, we adopt a framework that encompasses both examples above, and
many other applications detailed in \Cref{sec:applications}; we propose
a deterministic algorithm for computing a minimal basis of solutions to this
general problem.

\subsection{Minimal interpolation bases} Consider a field $\field$ and the
vector space $\vecSpace = \field^\order$, for some positive integer $\order$;
we see its elements as row vectors. Choosing a $\order \times \order$ matrix
${\mulmat}$ with entries in $\field$ allows us to make $\vecSpace$ a
$\polRing$-module in the usual manner, by setting $p \mul \row{e} = \row{e} \,
p({\mulmat})$, for $p$ in $\polRing$ and $\row{e}$ in $\vecSpace$. We will call
$\mulmat$ the {\em multiplication matrix} of $(\vecSpace, \mul)$. 

\begin{definition}[Interpolant]
  \label{dfn:interp}
  Given a vector $\evMat = (\row{e}_1,\ldots,\row{e}_\rdim)$ in
  $\vecSpace^\rdim$ and a vector $\row{p} = (p_1,\ldots,p_\rdim)$ in
  $\polRing^\rdim$, we write $\row{p} \mul \evMat = p_1 \mul \row{e}_1 + \cdots +
  p_\rdim \mul \row{e}_\rdim \in \vecSpace$. We say that $\row{p}$ is an
  \emph{interpolant for $(\evMat,\mulmat)$} if 
\begin{equation}\label{eq:definterpol}
\row{p} \mul \evMat =0.
\end{equation} 
\end{definition}
\noindent Here, $\int$ is seen as a row vector, and $\evMat$ is seen as a
column vector of $\rdim$ elements of $\vecSpace$: as a matter of notation,
$\evMat$ will often equivalently be seen as an $\rdim \times \order$ matrix
over $\field$.

Interpolants $\row{p}$ are often called {\em relations} or {\em syzygies} of
$\row{e}_1,\dots,\row{e}_\rdim$. This notion of interpolants was introduced by
\citet{BecLab00}, with the requirement that $\mulmat$ be upper triangular. One
of the main results of this paper holds with no assumption on $\mulmat$; for
our second main result, we will work under the stronger assumption that
$\mulmat$ is a Jordan matrix: it has $\nbbl$ Jordan blocks of respective sizes
$\szbl_1, \ldots, \szbl_\nbbl$ and with respective eigenvalues $\evpt_1,
\ldots, \evpt_\nbbl$. 

In the latter context, the notion of interpolant directly relates to the one
introduced by \citet{BarBul92} in terms of $\field[X]$-modules. Indeed, one may
identify $\vecSpace$ with the product of residue class rings
\[\polSpace=\polRing/(X^{\szbl_1}) \times \cdots \times
\polRing/(X^{\szbl_\nbbl}),\] by mapping a vector $\row{f}=(f_1,\dots,f_\nbbl)$
in $\polSpace$ to the vector $\row{e} \in \vecSpace$ made from the
concatenation of the coefficient vectors of $f_1,\dots,f_\nbbl$. Then, over
$\polSpace$, the $\polRing$-module structure on $\vecSpace$ given by $p \mul
\row{e} = \row{e}\, p(\mulmat)$ simply becomes \[p \cdot \row{f} =
(p(X+\evpt_1) f_1 \bmod X^{\szbl_1}, \dots,p(X+\evpt_\nbbl) f_\nbbl \bmod
X^{\szbl_n}).\] Now, if $(\row{e}_1,\dots,\row{e}_\rdim)$ in $\vecSpace^\rdim$
is associated to $(\row{f}_1,\dots,\row{f}_\rdim)$ in $\polSpace^\rdim$, with
$\row{f}_i = ({f}_{i,1},\dots,{f}_{i,\nbbl})$ and ${f}_{i,j}$ in
$\polRing/(X^{\szbl_j})$ for all $i,j$, the relation $p_1 \mul \row{e}_1 +
\cdots + p_\rdim \mul \row{e}_\rdim = 0$ means that for all $j$ in
$\{1,\dots,\nbbl\}$, we have \[p_1(X+\evpt_j) f_{1,j} + \cdots +
p_\rdim(X+\evpt_j) f_{\rdim,j} = 0 \bmod X^{\sigma_j};\] applying a translation
by $-\evpt_j$, this is equivalent to 
\begin{equation}\label{eq:vectorMPade}
p_1 f_{1,j}(X-x_j) + \cdots + p_\rdim
f_{\rdim,j}(X-x_j) = 0 \bmod (X-x_j)^{\sigma_j}.  
\end{equation}
Thus, in terms of vector M-Pad\'e approximation as in~\citep{BarBul92},
$(p_1,\ldots,p_\rdim)$ is an interpolant for
$(\row{f}_1,\ldots,\row{f}_\rdim)$, $\evpt_1, \ldots, \evpt_\nbbl$, and
$\szbl_1, \ldots, \szbl_\nbbl$.

Both examples above, and many more along the same lines, can be cast
into this setting. In the second example above, this is
straightforward: we have $\nbbl=1$, $\evpt_1=0$, so that the multiplication
matrix is the upper shift matrix
\begin{equation}
  \label{eqn:mulshift}
\mulshift =
\begin{bmatrix}
\;\;0\;\; & \;\;1\;\; \\
          & \ddots    & \ddots                \\
          &           & \;\;0\;\; & \;\;1\;\; \\
          &           &           & \;\;0\;\;
\end{bmatrix};  
\end{equation}
it is a nilpotent Jordan block.  In the first example, we 
have $\nbbl=\order$, and the multiplication matrix is the diagonal matrix
\[\mat{D}= 
  \begin{bmatrix}
    \;\;x_1\;\; \\
        & \;\;x_2\;\; \\
        &     & \ddots \\
        &     &        & \;\;x_\order\;\;
  \end{bmatrix}.
\]
Then, for $p$ in $\polRing$ and $\row{e} = [e_1,\dots,e_\order]$ in
$\vecSpace$, $p \mul \row{e}$ is the row vector $[p(x_1) e_1, \dots,
p(x_\order) e_\order]$. In this case, to solve the interpolation problem, we
start from the tuple of bivariate polynomials $(1,Y,\dots,Y^{\rdim-1})$; their
evaluations $\evMat = (\row{e}_1,\dots,\row{e}_\rdim)$ in $\vecSpace^\rdim$ are
the vectors $\row{e}_i = [y_1^i,\dots,y_\order^i]$ and the relation $\row{p}
\mul \evMat = 0$ precisely means that $Q(X,Y) = p_1 + p_2 Y + \cdots + p_\rdim
Y^{\rdim-1}$ vanishes at all points $\{(x_j,y_j), 1\le j\le \order\}$, where
$\row{p} = (p_1,\ldots,p_\rdim)$.
 
Let us come back to our general context. The set of all interpolants for
$(\evMat,\mulmat)$ is a submodule of $\polRing^\rdim$, which we will denote by
$\intMod$. Since it contains $\Pi_{\mulmat}(X) \polRing^\rdim$, where
$\Pi_{\mulmat} \in \polRing$ is the minimal polynomial of $\mulmat$, this
submodule is free of rank $\rdim$ (see for example~\citep[Chapter 12, Theorem
4]{DumFoo04}).
\begin{definition}[Interpolation basis]
  \label{dfn:ib}
  Given $\evMat$ in $\vecSpace^\rdim$ and $\mulmat$ in $\matSpace[\order]$, a
  matrix $\intBasis$ in $\intSpace$ is an \emph{interpolation basis for
  $(\evMat,\mulmat)$} if its rows form a basis of $\intMod$.
\end{definition}
In terms of notation, if a matrix $\mat{P} \in \polMatSpace[k][\rdim]$ has rows
$\row{p}_1,\dots,\row{p}_k$, we write $\mat{P} \mul \evMat$ for $(\row{p}_1
\mul \evMat, \dots, \row{p}_k \mul \evMat) \in \vecSpace^k$, seen as a column
vector. Thus, for $k = \rdim$, if $\intBasis$ is an interpolation basis for
$(\evMat,\mulmat)$ then in particular $\intBasis \mul \evMat = 0$.

In many situations, one wants to compute an interpolation basis which has
sufficiently small degrees: as we will see in \Cref{sec:applications},
most previous algorithms compute a basis which is \emph{reduced} with respect
to some degree \emph{shift}. In what follows, by shift, we mean a tuple of
nonnegative integers which will be used as degree weights on the columns of a
polynomial matrix. Before giving a precise definition of shifted minimal
interpolation bases, we recall the notions of shifted row degree and shifted
reducedness for univariate polynomial matrices; for more details we refer
to~\citep{Kailath80} and \citep[Chapter 2]{Zhou12}.

The row degree of a matrix $\mat{P} = [p_{i,j}]_{i,j}$ in
$\polMatSpace[k][\rdim]$ with no zero row is the tuple $\rdeg{\mat{P}} = (d_1,
\ldots , d_k) \in \NN^k$ with $d_i = \max_j \deg(p_{i,j})$ for all~$i$.  For a
shift $\shifts = (\shift{1}, \ldots \shift{\rdim}) \in \shiftSpace$, the
diagonal matrix with diagonal entries $\var^{\shift{1}}, \ldots,
\var^{\shift{\rdim}}$ is denoted by $\shiftMat{\shifts}$, and the
$\shifts$-row degree of $\mat{P}$ is $\rdeg[\shifts]{\mat{P}} = \rdeg{\mat{P}
\shiftMat{\shifts}}$. Then, the $\shifts$-leading matrix of $\mat{P}$ is the
matrix in $\matSpace[k][\rdim]$ whose entries are the coefficients of degree
zero of $\shiftMat{-\rdeg[{\shifts}]{\mat{P}}} \mat{P} \shiftMat{\shifts}$, and
we say that $\mat{P}$ is $\shifts$-reduced when its $\shifts$-leading matrix has
full rank. In particular, if $\mat{P}$ is square (still with no zero row), it
is $\shifts$-reduced if and only if its $\shifts$-leading matrix is invertible.
We note that $\mat{P}$ is $\shifts$-reduced if and only if
$\mat{P}\shiftMat{\shifts}$ is $\unishift$-reduced, where $\unishift =
(0,\ldots,0)$ is called the \emph{uniform} shift.

\begin{definition}[Shifted minimal interpolation basis]
  \label{dfn:mib}
  Consider $\vecSpace= \field^\order$ and a multiplication matrix $\mulmat$ in
  $\matSpace[\order]$. Given $\evMat = (\row{e}_1,\ldots,\row{e}_\rdim)$ in
  $\vecSpace^\rdim$ and a shift $\shifts \in \shiftSpace$, a matrix $\intBasis
  \in \intSpace$ is said to be an \emph{$\shifts$-minimal interpolation basis
    for $(\evMat,\mulmat)$} if
  \begin{itemize}
  \item $\intBasis$ is an interpolation basis for $(\evMat,\mulmat)$, and
  \item $\intBasis$ is $\shifts$-reduced.
\end{itemize}
\end{definition}

We recall that all bases of a free $\polRing$-module of rank $\rdim$ are
unimodularly equivalent: given two bases $\mat{A}$ and $\mat{B}$, there exists
$\mat{U}\in\intSpace$ such that $\mat{A} = \mat{U}\mat{B}$ and $\mat{U}$ is
unimodular (that is, $\mat{U}$ is invertible in $\intSpace$). Among all the
interpolation bases for $(\evMat,\mulmat)$, an $\shifts$-minimal basis
$\intBasis$ has a type of minimal degree property. Indeed, $\mat{P}$ is
$\shifts$-reduced if and only if $\rdeg[\shifts]{\mat{P}} \le
\rdeg[\shifts]{\mat{U}\mat{P}}$ for any unimodular $\mat{U}$; in this
inequality, the tuples are first sorted in non-decreasing order and then
compared lexicographically. In particular, a row of $\intBasis$ which has
minimal $\shifts$-row degree among the rows of $\intBasis$ also has minimal
$\shifts$-row degree among \emph{all} interpolants for $(\evMat,\mulmat)$.

\subsection{Main results} In this article, we propose fast deterministic
algorithms that solve \Cref{pbm:mib}.

\medskip
\begin{center}
\fbox{ \begin{minipage}{10.2cm}
\begin{problem}[Minimal interpolation basis]
  \label{pbm:mib}
~ \smallskip \\
\emph{Input:}
  \begin{itemize}
  \item the base field $\field$,
  \item the dimensions $\rdim$ and $\order$,
  \item a matrix $\evMat \in \evSpace{\order}$,
  \item a matrix $\mulmat \in \matSpace[\order]$,
  \item a shift $\shifts\in\shiftSpace$.
  \end{itemize}

\smallskip
\emph{Output:} an $\shifts$-minimal interpolation basis
$\intBasis\in\intSpace$ for $(\evMat,\mulmat)$.
\end{problem}
\end{minipage}
}
\end{center}
\medskip

Our first main result deals with an arbitrary matrix $\mulmat$, and 
uses techniques from fast linear algebra.

Taking $\mulmat$ \emph{upper triangular} as in~\citep{BecLab00} would allow us
to design a divide-and-conquer algorithm, using the leading and trailing
principal submatrices of $\mulmat$ for the recursive calls. However, this
assumption alone is not enough to obtain an algorithm with cost quasi-linear in
$\order$, as simply representing $\mulmat$ would require a number of
coefficients in $\field$ quadratic in $\order$. Taking $\mulmat$ a Jordan
matrix solves this issue, and is not a strong restriction for applications: it
is satisfied in all those we have in mind, which are detailed in
\Cref{sec:applications}.

If $\jordan\in\matSpace[\order]$ is a Jordan matrix with $\nbbl$ diagonal
blocks of respective sizes $\szbl_1,\ldots,\szbl_\nbbl$ and with respective
eigenvalues $\evpt_1,\ldots,\evpt_\nbbl$, we will write it in a compact manner
by specifying only those sizes and eigenvalues. Precisely, we will 
assume that
$\mulmat$ is given to us as the form
\begin{equation}\label{eq:mulmat_normal}
\mulmat=((x_1,\sigma_{1,1}),\dots,(x_1,\sigma_{1,r_1}),\dots,(x_t,\sigma_{t,1}),\dots,(x_t,\sigma_{t,r_t})),  
\end{equation}
for some pairwise distinct $x_1,\dots,x_t$, with $r_1 \ge \cdots \ge r_t$ and
$\sigma_{i,1} \ge \cdots \ge \sigma_{i,r_i}$ for all $i$; we will say that this
representation is {\em standard}. If $\mulmat$ is given as an arbitrary list $(
(\evpt_1,\szbl_1),\ldots,(\evpt_\nbbl,\szbl_\nbbl))$, we can reorder it (and
from that, permute the columns of $\evMat$ accordingly) to bring it to the
above form in time $\bigO{\polmultime{\sigma} \log(\sigma)^3}$ using the
algorithm of \citet[Proposition~12]{BoJeSc08}; if $\field$ is equipped with an
order, and if we assume that comparisons take unit time, it is of course enough
to sort the $x_i$'s. Here, $\polmultime{\cdot}$ is a multiplication time
function for $\polRing$: polynomials of degree at most $d$ in $\polRing$ can be
multiplied using $\polmultime{d}$ operations in $\field$, and
$\polmultime{\cdot}$ satisfies the super-linearity properties of~\citep[Chapter
8]{vzGathen13}. It follows from the algorithm of \citet{CanKal91} that
$\polmultime{d}$ can be taken in $\bigO{d \log(d) \log(\log(d))}$.

Adding a constant to every entry of $\shifts$ does not change the notion of
$\shifts$-reducedness, and thus does not change the output matrix $\intBasis$
of \Cref{pbm:mib}; in particular, one may ensure that $\min(\shifts)=0$
without loss of generality. The shift $\shifts$, as a set of degree weights on
the columns of $\intBasis$, naturally affects how the degrees of the entries of
$\intBasis$ are distributed. Although no precise degree profile of $\intBasis$
can be stated in general, we do have a global control over the degrees in
$\intBasis$, as showed in the following results; in these statements, for a
shift $\shifts$, we write $\sshifts$ to denote the quantity $\sshifts =
\shift{1} + \cdots + \shift{\rdim}$.

We start with the most general result in this paper, where we make no
assumption on~$\mulmat$. In this case, we obtain an algorithm whose cost is
essentially that of fast linear algebra over $\field$. The output of this
algorithm has an extra uniqueness property: it is in {\em Popov form}; we refer
the reader to \Cref{sec:lin-mib} or~\citep{BeLaVi06} for a definition.

\begin{theorem}
  \label{thm:mib-linear}
  There is a deterministic algorithm which solves \Cref{pbm:mib} using
  \begin{align*}
    \bigO{ & \order^\expmatmul ( \lceil \rdim / \order \rceil + \log(\order) )} & \text{if } \expmatmul>2 \\
    \bigO{ & \order^2 ( \lceil \rdim / \order \rceil + \log(\order) ) \log(\order) } & \text{if } \expmatmul=2
  \end{align*}
  operations in $\field$ and returns the unique $\shifts$-minimal interpolation
  basis for $(\evMat,\mulmat)$ which is in $\shifts$-Popov form. Besides, the
  sum of the column degrees of this basis is at most $\order$.
\end{theorem}

In the usual case where $\rdim = \bigO{\order}$, the cost is thus
$\softO{\order^\expmatmul}$; this is to be compared with the algorithm
of~\citet{BecLab00}, which we discuss in \Cref{sec:applications}. Next, we deal
with the case of $\mulmat$ in Jordan canonical form, for which we obtain a cost
bound that is quasi-linear with respect to $\order$.

\begin{theorem}
	\label{thm:mib}
  Assuming that $\mulmat\in\matSpace[\order]$ is a Jordan matrix, given by a
  standard representation, there is a deterministic algorithm which solves
  \Cref{pbm:mib} using
  \begin{align*}
   \bigO{ & \rdim^{\expmatmul-1} \polmultime{\order} \log(\order)
   \log(\order/\rdim) + \rdim^{\expmatmul-1} \polmultime{\xi} \log(\xi/\rdim)}
   & \text{if } \expmatmul>2 \\
   \bigO{ & \rdim \polmultime{\order} \log(\order) \log(\order/\rdim)
   \log(\rdim)^3 + \rdim \polmultime{\xi} \log(\xi/\rdim) \log(\rdim)^2} &
   \text{if } \expmatmul=2
  \end{align*}
   operations in $\field$, where $\xi = \sshifts[\shifts-\min(\shifts)]$.
   Besides, the sum of the row degrees of the computed $\shifts$-minimal
   interpolation basis is at most $\order+\xi$.
\end{theorem}

The reader interested in the logarithmic factors should refer to the more
precise cost bound in \Cref{prop:mib}. Masking logarithmic factors,
this cost bound is $\softO{ \rdim^{\expmatmul-1} (\order+\xi) }$. We remark
that the bound on the output row degree implies that the size of the output
matrix $\mat{P}$ is $\bigO{ \rdim (\order + \xi)}$, where by size we mean the
number of coefficients of $\field$ needed to represent this matrix.

We are not aware of a previous cost bound for the general question stated in
\Cref{pbm:mib} that would be similar to our result; we give a detailed
comparison with several previous algorithms and discuss useful particular cases
in \Cref{sec:applications}.

\subsection{Overview of the algorithms} To deal with an arbitrary matrix
$\mulmat$, we rely on a linear algebra approach presented in
\Cref{sec:lin-mib}, using a linearization framework that is classical
for this kind of problems~\citep{Kailath80}. Our algorithm computes the rank
profile of a block Krylov matrix using techniques that are reminiscent of the
algorithm of \citet{KelGeh85}; this framework also allows us to derive a bound
on the sum of the row degrees of shifted minimal interpolation bases.
\Cref{sec:lin-mib} is the last section of this paper; it is the only
section where we make no assumption on $\mulmat$, and it does not use results
from other parts of the paper.

We give in \Cref{sec:algo} a divide-and-conquer algorithm for the case of a
matrix $\mulmat$ in Jordan canonical form. The idea is to use a
Knuth-Sch\"onhage-like half-gcd approach~\citep{Knu70,Sch71,BrGuYu80},
previously carried over to the specific case of simultaneous Hermite-Pad\'e
approximation in~\citep{BecLab94,GiJeVi03}. This approach consists in reducing
a problem in size $\order$ to a first sub-problem in size $\order/2$, the
computation of the so-called \emph{residual}, a second sub-problem in size
$\order/2$, and finally a recombination of the results of both sub-problems via
polynomial matrix multiplication. The shift to be used in the second recursive
call is essentially the $\shifts$-row degree of the outcome of the first
recursive call.

The main difficulty is to control the sizes of the interpolation bases that are
obtained recursively. The bound we rely on, as stated in our main theorems,
depends on the input shift. In our algorithm, we cannot make any assumption on
the shifts that will appear in recursive calls, since they depend on the
degrees of the previously computed bases. Hence, even in the case of a uniform
input shift for which the output basis is of size $\bigO{\rdim \order}$, there
may be recursive calls with an unbalanced shift, which may output bases that
have large size.

Our workaround is to perform all recursive calls with the uniform shift
$\shifts = \unishift$, and resort to a change of shift that will be studied in
\Cref{sec:change-shift}; this strategy is an alternative to the
linearization approach that was used by~\citet{ZhoLab12} in the specific case
of simultaneous Hermite-Pad\'e approximation. We note that our change of shift
uses an algorithm in~\citep{ZhLaSt12}, which itself relies on simultaneous
Hermite-Pad\'e approximation; with the dimensions needed here, this
approximation problem is solved efficiently using the algorithm
of~\citet{GiJeVi03}, without resorting to~\citep{ZhoLab12}.

Another difficulty is to deal with instances where $\order$ is small.  Our
bound on the size of the output states that when $\order \le \rdim$ and the
shift is uniform, the average degree of the entries of a minimal interpolation
basis is at most $1$. Thus, in this case our focus is not anymore on using fast
polynomial arithmetic but rather on exploiting efficient linear algebra over
$\field$: the divide-and-conquer process stops when reaching $\order \le
\rdim$, and invokes instead the algorithm based on linear algebra proposed in
\Cref{sec:lin-mib}.

The last ingredient is the fast computation of the residual, that is,
a matrix in $\matSpace[\rdim][\order/2]$ for restarting the process
after having found a basis $\intBasis^{(1)}$ for the first sub-problem
of size $\order/2$. This boils down to computing $\intBasis^{(1)} \mul
\evMat$ and discarding the first $\order/2$ columns, which are known to
be zero. In \Cref{sec:multiplication}, we design a general
procedure for computing this kind of product, using Hermite
interpolation and evaluation to reduce it to multiplying polynomial
matrices. 

Concerning the multiplication of the bases obtained recursively, to handle the
fact that they may have unbalanced row degrees, we use the approach
in~\citep[Section 3.6]{ZhLaSt12}; we give a detailed algorithm and cost analysis
in \Cref{sec:cost-mult}.

%
%

\section{Applications and comparisons with previous work}
\label{sec:applications}

In this section, we review and expand the scope of the examples described in
the introduction, and we compare our results for these examples to previous
work. \emph{For ease of comparison, in all this section we consider the case
$\expmatmul>2$.} 

\subsection{General case}

In this paragraph, we consider the general \Cref{pbm:mib}, without assuming
that $\mulmat$ is a Jordan matrix. The only previous work that we are aware of
is~\citep{BecLab00}, where it is still assumed that $\mulmat$ is upper
triangular. This assumption allows one to use an iteration on the columns of
$\evMat$, combining Gaussian elimination with multiplication by monic
polynomials of degree $1$ to build the basis $\intBasis$: after $i$ iterations,
$\intBasis$ is an $\shifts$-minimal interpolation basis for the first $i$
columns of $\evMat$ and the $i \times i$ leading principal submatrix of
$\mulmat$. This algorithm uses $\bigO{ \rdim \order^4 }$ operations and returns
a basis in $\shifts$-Popov form. We note that in this context the mere
representation of $\mulmat$ uses $\Theta(\order^2)$ elements.

As a comparison, the algorithm of \Cref{thm:mib-linear} computes an
interpolation basis for $(\evMat,\mulmat)$ in $\shifts$-Popov form for
{\em any} matrix $\mulmat$ in $\matSpace[\order]$. For any shift
$\shifts$, this algorithm uses $\bigO{ \order^\expmatmul \log(\order)
}$ operations when $\rdim \in \bigO{ \order }$ and $\bigO{
  \order^{\expmatmul-1} \rdim + \order^{\expmatmul} \log(\order) }$
operations when $\order \in \bigO{\rdim}$.  

\subsection{M-Pad\'e approximation}

We continue with the case of a matrix $\mulmat$ in Jordan canonical form. As
pointed out in the introduction, \Cref{pbm:mib} can be formulated in this case
in terms of polynomial equations and corresponds to an M-Pad\'e approximation
problem; this problem was studied in \citep{Lubbe83,Beckermann90} and named
after the work of Mahler, including in particular \citep{Mahler68}. Indeed, up
to applying a translation in the input polynomials, the problem stated
in~\eqref{eq:vectorMPade} can be rephrased as follows.

\medskip
\begin{center}
\fbox{ \begin{minipage}{10.2cm}
\begin{problem}[M-Pad\'e approximation]
\label{pbm:m-pade}
~ \smallskip \\
\emph{Input:}
  \begin{itemize}
  \item pairwise distinct points $\evpt_1,\ldots,\evpt_\cdim$ in $\field$,
  \item integers $\order_1 \ge  \cdots \ge \order_\cdim > 0$,
  \item matrix $\mat{F}$ in $\polMatSpace[\rdim][\cdim]$ with its
    $j$-th column $\matcol{\mat{F}}{j}$ of degree $< \order_j$,
  \item shift $\shifts \in \NN^m$.
  \end{itemize}

\smallskip
  \emph{Output:} a matrix $\intBasis$ in $\intSpace$ such that
  \begin{itemize}
    \item the rows of $\intBasis$ form a basis of the $\polRing$-module
    \[ \{ \row{p} \in \polMatSpace[1][\rdim] \;\mid\; \row{p}
    \matcol{\mat{F}}{j} = 0 \bmod (\var-\evpt_j)^{\order_j} \text{ for each } j
      \}, \]
    \item $\intBasis$ is $\shifts$-reduced.
  \end{itemize}
\end{problem}
\end{minipage}
}
\end{center}
\medskip

Here, $\order = \order_1 + \cdots + \order_\cdim$. Previous work on this
particular problem includes~\citep{Beckermann92,BarBul92}; note that
in~\citep{Beckermann92}, the input consists of a single column $\mat{F}$ in
$\polMatSpace[\rdim][1]$ of degree less than $\order$: to form the input of our
problem, we compute $\hat{\mat{F}} = [ \mat{F} \bmod (\var-\evpt_1)^{\order_1}
| \cdots | \mat{F} \bmod (\var-\evpt_\cdim)^{\order_\cdim} ]$. The algorithms
in these references have a cost of $\bigO{\rdim^2 \order^2}$ operations, which
can be lowered to $\bigO{\rdim \order^2}$ if one computes only the small degree
rows of an $\shifts$-minimal basis (gradually discarding the rows of degree
more than, say, $2\order/\rdim$ during the computation).  To the best of our
knowledge, no algorithm with a cost quasi-linear in $\order$ has been given in
the literature.

\subsection{Hermite-Pad\'e approximation}

Specializing the discussion of the previous paragraph to the case where all
$\evpt_i$ are zero, we obtain the following important particular case.

\medskip
\begin{center}
\fbox{ \begin{minipage}{10.2cm}
\begin{problem}[Simultaneous Hermite-Pad\'e approximation]
  \label{pbm:h-pade}
  ~\smallskip \\
  \emph{Input:}
  \begin{itemize}
  \item integers $\order_1 \ge \cdots \ge \order_\cdim > 0$,
  \item matrix $\mat{F}$ in $\polMatSpace[\rdim][\cdim]$ with its
    $j$-th column $\matcol{\mat{F}}{j}$ of degree $< \order_j$,
  \item shift $\shifts \in \NN^m$.
  \end{itemize}

  \smallskip
  \emph{Output:} a matrix $\intBasis$ in $\intSpace$ such that
  \begin{itemize}
    \item the rows of $\intBasis$ form a basis of the $\polRing$-module
    \[ \{ \row{p} \in \polMatSpace[1][\rdim] \;\mid\; \row{p}
      \matcol{\mat{F}}{j} = 0 \bmod \var^{\order_j} \text{ for each } j \}, \]
    \item $\intBasis$ is $\shifts$-reduced.
  \end{itemize}
\end{problem}
\end{minipage}
}
\end{center}
\medskip

Here, $\order = \order_1 + \cdots + \order_\rdim$. Our main result says that
there is an algorithm which solves \Cref{pbm:h-pade} using $\bigO{
\rdim^{\expmatmul-1} \polmultime{\order} \log(\order)\log(\order/\rdim) +
\rdim^{\expmatmul-1} \polmultime{\xi} \log(\xi/\rdim)}$ operations in
$\field$, where $\xi = \sshifts[\shifts-\min(\shifts)]$. As we will see,
slightly faster algorithms for this problem exist in the literature, but they
do not cover the same range of cases as we do; to the best of our knowledge,
for instance, most previous work on this problem dealt with the case
$\order_1 = \cdots = \order_\cdim = \order/\cdim$ and $\cdim \le \rdim$.

First algorithms with a cost quadratic in $\order$ were given
in~\citep{Sergeyev87,Paszkowski87} in the case of Hermite-Pad\'e approximation
($\cdim=1$), assuming a type of genericity of $\mat{F}$ and outputting a single
approximant $\row{p} \in \polMatSpace[1][\rdim]$ which satisfies some
prescribed degree constraints. For $\cdim=1$, \citet{BarBul91} propose an
algorithm which uses $\bigO{\rdim^2 \order^2}$ operations to compute a
$\shifts$-minimal basis of approximants for $\mat{F}$ at order $\order$, for
any $\mat{F}$ and $\shifts$. This result was extended by \citet{BecLab94} to
the case of any $\cdim \le \rdim$, with the additional remark that the cost
bound is $\bigO{\rdim \order^2}$ if one restricts to computing only the rows of
small degree of an $\shifts$-minimal basis.

In~\citep{BecLab94}, the authors also propose a divide-and-conquer algorithm
using $\softO{\rdim^\expmatmul \order}$ operations in $\field$; the base case
of the recursion deals with the constant coefficient of a single column of the
input $\mat{F}$. Then, \citet{GiJeVi03} follow a similar divide-and-conquer
approach, introducing a base case which has matrix dimension $\rdim \times
\cdim$ and is solved efficiently by means of linear algebra over $\field$; this
yields an algorithm with the cost bound $\bigO{\rdim^\expmatmul
  \polmultime{\order/\cdim} \log(\order/\cdim)}$, which is particularly
  efficient when $\cdim \in \Theta(\rdim)$. On the other hand, in the case of
  Hermite-Pad\'e approximation ($\cdim=1$) it was noticed by \citet{Lecerf01}
  that the cost bound $\softO{\rdim^\expmatmul \order}$ is pessimistic, at
  least when some type of genericity is assumed concerning the input $\mat{F}$,
  and that in this case there is hope to achieve $\softO{\rdim^{\expmatmul-1}
\order}$.

This cost bound was then obtained by~\citet{Storjohann06}, for computing only
the small degree rows of an $\shifts$-minimal basis, via a reduction of the
case of small $\cdim$ and order $\order$ to a case with larger column dimension
$\cdim' \approx \rdim$ and smaller order $\order' \approx \cdim \order /
\rdim$.  This was exploited to compute a full $\shifts$-minimal
basis using $\bigO{ \rdim^\expmatmul \polmultime{\order/\rdim}
\log(\order/\cdim) }$ operations~\citep{ZhoLab12}, under the assumption that either
$\sshifts[{\shifts-\min(\shifts)}] \in \bigO{\order}$ or
$\sshifts[{\max(\shifts)-\shifts}] \in \bigO{\order}$ (which both imply that an
$\shifts$-minimal basis has size $\bigO{\rdim \order}$). We note that this
algorithm is faster than ours by a logarithmic factor, and that our result does
not cover the second assumption on the shift with a similar cost bound; on the
other hand, we handle situations that are not covered by that result, when for
instance all $\order_i$'s are not equal.

\citet[Section 3.6]{Zhou12} gives a fast algorithm for the case $\nbbl \le
\order \le \rdim$, for the uniform shift and $\szbl_1 = \cdots =
\szbl_\nbbl=\order/\cdim$. For an input of dimensions $\rdim \times \nbbl$, the
announced cost bound is $\softO{\rdim^{\expmatmul-1} \order}$; a more precise
analysis shows that the cost is $\bigO{\order^{\expmatmul-1} \rdim +
\order^\expmatmul \log(\order/\cdim)}$. For this particular case, there is no
point in using our divide-and-conquer algorithm, since the recursion stops at
$\order \le \rdim$; our general algorithm in \Cref{sec:lin-mib} handles
these situations, and its cost (as given in \Cref{prop:lin-mib} with
$\maxDeg = \order/\cdim$) is the same as that of Zhou's algorithm in this
particular case. In addition, this cost bound is valid for any shift $\shifts$
and the algorithm returns a basis in $\shifts$-Popov form.

\subsection{Multivariate interpolation}
\label[subsec]{subsec:multi-int}
In this subsection, we discuss how our algorithm can be used to solve bivariate
interpolation problems, such as those appearing in the
list-decoding~\citep{Sudan97,GurSud98} and the
soft-decoding~\citep[Sec.\,III]{KoeVar03a} of Reed-Solomon codes; as well as
multivariate interpolation problems, such as those appearing in the
list-decoding of folded Reed-Solomon codes~\citep{GurRud08} and in robust
Private Information Retrieval~\citep{DeGoHe12}. Our contribution leads to the
best known cost bound we are aware of for the interpolation steps of the
list-decoding and soft-decoding of Reed-Solomon codes: this is detailed in
\Cref{subsec:list-dec,subsec:soft-dec}.

In those problems, we have $\nvars$ new variables $Y = Y_1,\ldots,Y_\nvars$ and
we want to find a multivariate polynomial $Q(X,Y)$ which vanishes at some given
points $\{(x_k,y_k)\}_{1\le k\le \nbpt}$ with prescribed \emph{supports}
$\{\supp_k\}_{1\le k\le \nbpt}$. In this context, given a point $(x,y) \in
\field^{\nvars+1}$ and an exponent set $\supp \subset \NN^{\nvars+1}$, we say
that \emph{$Q$ vanishes at $(x,y)$ with support $\supp$} if the shifted
polynomial $Q(X+x,Y+y)$ has no monomial with exponent in $\supp$. Besides, the
solution $Q(X,Y)$ should also have sufficiently small \emph{weighted degree}
$\deg_X(Q(X,X^{\shifts[w]} Y))$ for some given weight $\shifts[w] \in
\shiftSpace[\nvars]$.

To make this fit in our framework, we first require that we are given an
exponent set $\expSet \subseteq \NN^\nvars$ such that any monomial $X^i Y^j$
appearing in a solution $Q(X,Y)$ should satisfy $j \in \expSet$. Then, we can
identify $Q(X,Y) = \sum_{j \in \expSet} p_j(X) Y^j$ with $\row{p} =
[p_j]_{j\in\expSet} \in \polMatSpace[1][\rdim]$, where $\rdim$ is the
cardinality of $\expSet$. We will show below how to construct matrices $\evMat$
and $\mulmat$ such that solutions $Q(X,Y)$ to those multivariate interpolation
problems correspond to interpolants $\row{p} = [p_j]_{j\in\expSet}$ for
$(\evMat,\mulmat)$ that have sufficiently small shifted degree.

We also require that each considered support $\supp$ satisfies
\begin{equation}
  \label{eqn:multiplicity}
  \text{if } (i,j) \in \supp \text{ and } i>0, \text{ then } (i-1,j) \in\supp.
\end{equation}
Then, the set
\begin{align*}
  \mathfrak{I}_{\mathrm{int}} & = \Bigg\{ \row{p} = (p_j)_{j\in\expSet} \in
  \polMatSpace[1][\rdim] \;\Big| \\
  & \sum_{j\in\expSet} p_j(X) Y^j \text{ vanishes at } (\evpt_k,y_k) \text{
    with support } \supp_k \text{ for } 1 \le k \le \nbpt \Bigg\}
\end{align*}
is a $\polRing$-module, as can be seen from the equality 
\begin{equation}\label{eq:XQ}
(XQ)(X+\evpt, Y+y) = (X+\evpt) Q(X+\evpt, Y+y).
\end{equation}

Finally, as before, we assume that we are given a shift $\shifts$ such that we
are looking for polynomials of small $\shifts$-row degree in
$\mathfrak{I}_{\mathrm{int}}$. In this context, the shift is often derived from
a weighted degree condition which states that, given some input weights
$\shifts[w] = (\shift[w]{1}, \ldots, \shift[w]{\nvars}) \in
\shiftSpace[\nvars]$, the degree $\deg_X ( Q(X,X^{\shift[w]{1}} Y_1, \cdots,
X^{\shift[w]{\nvars}} Y_\nvars))$ should be sufficiently small. Since
$Q(X,X^{\shift[w]{1}} Y_1, \cdots, X^{\shift[w]{\nvars}} Y_\nvars) = \sum_{j
\in \expSet} X^{\shift[w]{1} j_1 + \cdots + \shift[w]{\nvars} j_{\nvars}}
p_j(X) Y^j$, this $\shifts[w]$-weighted degree of $Q$ is exactly the
$\shifts$-row degree of $\row{p}$ for the shift $\shifts = (\shift[w]{1} j_1 +
\cdots + \shift[w]{\nvars} j_{\nvars})_{ (j_1,\ldots,j_\nvars)\in\expSet }$. An
$\shifts$-reduced basis of $\mathfrak{I}_{\mathrm{int}}$ contains a row of
minimal $\shifts$-row degree among all $\row{p} \in
\mathfrak{I}_{\mathrm{int}}$; we note that in some applications, for example in
robust Private Information Retrieval~\citep{DeGoHe12}, it is important to
return a whole basis of solutions and not only a small degree one.

\medskip
\begin{center}
\fbox{ \begin{minipage}{10.2cm}
\begin{problem} [Constrained multivariate interpolation]
  \label{pbm:multi-int}
 ~\smallskip \\
 \emph{Input:}
 \begin{itemize}
   \item number of variables $\nvars>0$,
   \item exponent set $\expSet  \subset \NN^\nvars$ of cardinality $\rdim$,
   \item pairwise distinct points $\{(\evpt_k,y_k) \in \field^{\nvars+1}, 1\le
   k\le \nbpt \}$ with $\evpt_1,\ldots,\evpt_\nbpt \in \field$ and
   $y_1,\dots,y_\nbpt$ in $\field^\nvars$,
   \item supports $\{\supp_k \subset \NN^{\nvars+1}, 1 \le k \le \nbpt\}$
     where each $\supp_k$ satisfies~\eqref{eqn:multiplicity},
 \item shift $\shifts \in \shiftSpace$.
 \end{itemize}

 \smallskip
 \emph{Output:} a matrix $\intBasis \in \intSpace$ such that
  \begin{itemize}
    \item the rows of $\intBasis$ form a basis of the $\polRing$-module $\mathfrak{I}_{\mathrm{int}}$,
  \item $\intBasis$ is $\shifts$-reduced.
  \end{itemize}
\end{problem}
\end{minipage}
}
\end{center}
\medskip

This problem can be embedded in our framework as follows. We consider $\module
= \field[\var,Y_1,\ldots,Y_\nvars]$ and the $\field$-linear functionals $\{
  \linFunc_{i,j,k} : \module \rightarrow \field, (i,j) \in \supp_k, 1 \le k \le
\nbpt \}$, where $\linFunc_{i,j,k}(Q)$ is the coefficient of $X^{i}Y^j$ in
$Q(X+\evpt_k, Y+y_k)$. These functionals are linearly independent, and the
intersection $\kernel$ of their kernels is the $\polRing$-module of polynomials
in $\module$ vanishing at $(x_k,y_k)$ with support $\supp_k$ for all $k$. The
quotient $\module/\kernel$ is a $\field$-vector space of dimension $\order =
\szbl_1 + \cdots + \szbl_\nbpt$ where $\szbl_k = \#\supp_k$; it is thus
isomorphic to $\vecSpace = \field^\order$, with a basis of the dual space given
by the functionals $\linFunc_{i,j,k}$.

Our assumption on the supports $\supp_k$ implies that $\module/\kernel$ is a
$\field[X]$-module; we now describe the corresponding multiplication matrix
$\mulmat$.  For a given $k$, let us order the functionals $\linFunc_{i,j,k}$ in
such a way that, for any $(i,j)$ such that $(i+1,j)$ is in $\supp_k$, the
successor of $\linFunc_{i,j,k}$ is $\linFunc_{i+1,j,k}$. Equation~\eqref{eq:XQ}
implies that 
\begin{equation}\label{eq:ellXQ}
\linFunc_{i,j,k}(XQ) = \linFunc_{i-1,j,k}(Q) + \evpt_k
\linFunc_{i,j,k}(Q)  
\end{equation}
holds for all $Q$, all $k \in \{1,\ldots,\nbpt\}$, and all $(i,j) \in \supp_k$
with $i>0$.

Hence, $\mulmat$ is block diagonal with diagonal blocks
$\mulmat_1,\ldots,\mulmat_\nbpt$, where $\mulmat_k$ is a $\szbl_k \times
\szbl_k$ Jordan matrix with only eigenvalue $\evpt_k$ and block sizes given by
the support $\supp_k$. More precisely, denoting $\Lambda_k = \{ j \in
\NN^\nvars \mid (i,j) \in \supp_k \text{ for some } i \}$ and $\szbl_{k,j} =
\max\{i\in\NN \mid (i,j) \in \supp_k \}$ for each $j \in \Lambda_k$, we have
the disjoint union
\[ \supp_k = \bigcup_{j\in\Lambda_k} \big\{ (i,j), 0\le i\le \szbl_{k,j}\big\}. \]
Then, $\mulmat_k$ is block diagonal with $\#\Lambda_k$ blocks: to each $j \in
\Lambda_k$ corresponds a $\szbl_{k,j} \times \szbl_{k,j}$ Jordan block with
eigenvalue $\evpt_k$. It is reasonable to consider $\evpt_1,\ldots,\evpt_\nbpt$
ordered as we would like for a standard representation of $\mulmat$. For
example, in problems coming from coding theory, these points are part of the
code itself, so the reordering can be done as a pre-computation as soon as the
code is fixed.

To complete the reduction to \Cref{pbm:mib}, it remains to construct
$\evMat$. For each exponent $\gamma \in \expSet$ we consider the monomial
$Y^\gamma$ and take its image in $\module/\kernel$: this is the vector
$\row{e}_\gamma \in \vecSpace$ having for entries the evaluations of the
functionals $\linFunc_{i,j,k}$ at $Y^\gamma$. Let then $\evMat$ be the matrix
in $\matSpace[\rdim][\order]$ with rows $(\row{e}_\gamma)_{\gamma \in
\expSet}$: our construction shows that a row $\row{p} = (p_\gamma)_{\gamma \in
\expSet} \in \polMatSpace[1][\rdim]$ is in $\mathfrak{I}_{\mathrm{int}}$ if and
only if it is an interpolant for $(\evMat,\mulmat)$.

\medskip To make this reduction to \Cref{pbm:mib} efficient, we make the
assumption that the exponent sets $\expSet$ and $\supp_k$ are \emph{stable
under division}: this means that if $j \in \expSet$ then all $j'$ such that $j'
\le j$ (for the product order on $\NN^\nvars$) belong to $\expSet$; and if
$(i,j)$ is in $\supp$, then all $(i',j')$ such that $(i',j') \le (i,j)$ (for
the product order on $\NN^{\nvars+1}$) belong to $\supp_k$. This assumption is
satisfied in the applications detailed below; besides, using the
straightforward extension of~\eqref{eq:ellXQ} to multiplication by
$Y_1,\dots,Y_r$, it allows us to compute all entries
$\linFunc_{i,j,k}(Y^\gamma)$ of the matrix $\evMat$ inductively in $\bigO{\rdim
\order}$, which is negligible compared to the cost of solving the resulting
instance of \Cref{pbm:mib}. (As a side note, we remark that this
assumption also implies that $\kernel$ is a zero-dimensional ideal of
$\module$.)

\begin{proposition}
  \label{prop:multi-int}
  Assuming that $\expSet$ and $\supp_1,\ldots,\supp_\nbpt$ are stable under
  division, there is an algorithm which solves \Cref{pbm:multi-int}
  using
  \begin{align*}
   \bigO{ & \rdim^{\expmatmul-1} \polmultime{\order} \log(\order)\log(\order/\rdim)
   + \rdim^{\expmatmul-1} \polmultime{\xi} \log(\xi/\rdim)}  & \text{if } \expmatmul>2 \\
   \bigO{ & \rdim \polmultime{\order} \log(\order)\log(\order/\rdim) \log(\rdim)^3
   + \rdim \polmultime{\xi} \log(\xi/\rdim) \log(\rdim)^2 } & \text{if } \expmatmul=2
  \end{align*}
  operations in $\field$, where $\order = \#\supp_1 + \cdots + \#\supp_\nbpt$
  and $\xi = \sshifts[\shifts-\min(\shifts)]$.
\end{proposition}

\subsection{List-decoding of Reed-Solomon codes} 
\label[subsec]{subsec:list-dec}

In the algorithms of \citet{Sudan97,GurSud99}, the bivariate interpolation step
deals with \Cref{pbm:multi-int} with $r=1$, $\expSet =
\{0,\ldots,\rdim-1\}$, pairwise distinct points $x_1,\ldots,x_\nbpt$, and
$\supp_k = \{(i,j) \mid i+j < b\}$ for all $k$; $b$ is the \emph{multiplicity
parameter} and $\rdim-1$ is the \emph{list-size parameter}. As explained above,
the shift $\shifts$ takes the form $\shifts = (0,w,2w,\ldots,(\rdim-1)w)$; here
$w+1$ is the message length of the considered Reed-Solomon code and $\nbpt$ is
its block length.

We will see below, in the more general interpolation step of the soft-decoding,
that $\sshifts \in \bigO{\order}$. As a consequence,
\Cref{prop:multi-int} states that this bivariate interpolation step
can be performed using $\bigO{\rdim^{\expmatmul-1} \polmultime{\order}
\log(\order) \log(\order/\rdim)}$ operations. To our knowledge, the best
previously known cost bound for this problem is $\bigO{\rdim^{\expmatmul-1}
\polmultime{\order} \log(\order)^2}$ and was obtained using a randomized
algorithm~\citep[Corollary~14]{CJNSV15}. In contrast, \Cref{algo:mib}
in this paper makes no random choices. The algorithm in~\citep{CJNSV15} uses
fast structured linear algebra, following an approach studied
in~\citep{OlsSho99,RotRuc00,ZeGeAu11}. Restricting to deterministic algorithms,
the best previously known cost bound is $\bigO{ \rdim^\expmatmul
  \polmultime{\order/b} (\log(\rdim)^2 + \log(\order/b))
}$~\citep{Bernstein11,CohHen15}, where $b < \rdim$ is the multiplicity
parameter mentioned above. This is obtained by first building a known $\rdim
\times \rdim$ interpolation basis with entries of degree at most $\order/b$,
and then using fast deterministic reduction of polynomial
matrices~\citep{GuSaStVa12}; other references on this approach
include~\citep{Alekhnovich02,LeeOSul08,BeeBra10}.

A previous divide-and-conquer algorithm can be found in~\citep{Nielsen14}. The
recursion is on the number of points $\nbpt$, and using fast multiplication of
the bases obtained recursively, this algorithm has cost bound $\bigO{\rdim^2
\order b} + \softO{\rdim^\expmatmul \order / b}$
\citep[Proposition~3]{Nielsen14}. In this reference, the bases computed
recursively are allowed to have size as large as $\Theta(\rdim^2 \order/b)$,
with $b < \rdim$.

For a more detailed perspective of the previous work on this specific problem,
the reader may for instance refer to the cost comparisons in the introductive
sections of~\citep{BeeBra10,CJNSV15}.

\subsection{Soft-decoding of Reed-Solomon codes}
\label[subsec]{subsec:soft-dec}

In K\"otter and Vardy's algorithm, the so-called soft-interpolation
step~\citep[Section~III]{KoeVar03a} is \Cref{pbm:multi-int} with $\nvars=1$,
$\expSet = \{0,\ldots,\rdim-1\}$, and $\shifts = (0,w,\ldots,(\rdim-1)w)$. The
points $\evpt_1,\ldots,\evpt_\nbpt$ are not necessarily pairwise distinct, and
to each $\evpt_k$ for $k \in \{1,\ldots,\nbpt\}$ is associated a multiplicity
parameter $b_k$ and a corresponding support $\supp_k = \{(i,j) \mid i+j < b_k
\}$. Then, $\order = \sum_{1\le k \le \nbpt} \binom{b_k+1}{2}$ is called the
\emph{cost} \citep[Section~III]{KoeVar03a}, since it corresponds to the number
of linear equations that one obtains in the straightforward linearization of
the problem.

In this context, one chooses for $\rdim$ the smallest integer such that the
number of linear unknowns in the linearized problem is more than $\order$. This
number of unknowns being directly linked to $\sshifts$, this leads to $\sshifts
\in \bigO{\order}$, which can be proven for example using \citep[Lemma~1 and
Equations~(10) and~(11)]{KoeVar03a}. Thus, our algorithm solves the
soft-interpolation step using $\bigO{ \rdim^{\expmatmul-1} \polmultime{\order}
\log(\order) \log(\order/\rdim) }$ operations in $\field$, which is to our
knowledge the best known cost bound to solve this problem.

Iterative algorithms, now often referred to as \emph{K\"otter's algorithm} in
coding theory, were given in~\citep{Koetter96,HohNie00}; one may also refer to
the general presentation of this solution in~\citep[Section 7]{McEliece03}.
These algorithms use $\bigO{\rdim \order^2}$ operations in $\field$ (for the
considered input shifts which satisfy $\sshifts \in \bigO{\order}$). We showed
above how the soft-interpolation step can be reduced to a specific instance of
M-Pad\'e approximation (\Cref{pbm:m-pade}): likewise, one may remark
that the algorithms in~\citep{HohNie00,McEliece03} are closely linked to those
in~\citep{Beckermann92,BarBul92,BecLab00}. In particular, up to the
interpretation of row vectors in $\polMatSpace[\rdim]$ as bivariate polynomials
of degree less than $\rdim$ in $Y$, they use the same recurrence. Our recursive
\Cref{algo:mib-rec} is a fast divide-and-conquer version of these
algorithms.

An approach based on polynomial matrix reduction was developed
in~\citep{Alekhnovich02,Alekhnovich05} and in~\citep{LeeOSul06}. It consists
first in building an interpolation basis, that is, a basis $\mat{B} \in
\intSpace$ of the $\polRing$-module $\mathfrak{I}_{\mathrm{int}}$, and then in
reducing the basis for the given shift $\shifts$ to obtain an $\shifts$-minimal
interpolation basis. The maximal degree in $\mat{B}$ is
\[\maxDeg = \sum_{1\le k\le \nbpt} \frac{ \max \{b_i \;\mid\; 1\le i \le \nbpt
  \text{ and } \evpt_i = \evpt_k\} }{ \# \{1\le i \le \nbpt \;\mid\; \evpt_i =
  \evpt_k\} }; \] one can check using $b_k < \rdim$ for all $k$ that \[ \order
  = \sum_{1\le k\le \nbpt} \frac{b_k(b_k+1)}{2} \le \frac{\rdim}{2} \sum_{1 \le
  k\le \nbpt} b_k \le \frac{\rdim\maxDeg}{2}.\] Using the fast deterministic
  reduction algorithm in~\citep{GuSaStVa12}, this approach has cost bound
  $\bigO{\rdim^\expmatmul \polmultime{ \maxDeg } (\log(\rdim)^2 +
\log(\maxDeg))}$; the cost bound of our algorithm is thus smaller by a factor
$\softO{\rdim \delta / \order}$.

In~\citep[Section~5.1]{Zeh13} the so-called key equations commonly used in the
decoding of Reed-Solomon codes were generalized to this soft-interpolation
step. It was then showed in~\citep{CJNSV15} how one can efficiently compute a
solution to these equations using fast structured linear algebra. In this
approach, the set of points $\{(x_k,y_k), 1\le k \le \nbpt\}$ is partitioned as
$\mathcal{P}_1 \cup \cdots \cup \mathcal{P}_q$, where in each $\mathcal{P}_h$
the points have pairwise distinct $\evpt$-coordinates. We further write
$b^{(h)} = \max \{ b_k \;\mid\; 1\le i\le \nbpt \text{ and } (x_k,y_k) \in
\mathcal{P}_h \}$ for each $h$, and $\beta = \sum_{1\le h\le q} b^{(h)}$. Then,
the cost bound is $\bigO{ (\rdim+\beta)^{\expmatmul-1} \polmultime{\order}
\log(\order)^2 }$, with a probabilistic
algorithm~\citep[Section~IV.C]{CJNSV15}.  We note that $\beta$ depends on the
chosen partition of the points. The algorithm in this paper is deterministic
and has a better cost.

All the mentioned algorithms for the interpolation step of list- or
soft-decoding of Reed-Solomon codes, including the one presented in this paper,
can be used in conjunction with the \emph{re-encoding
technique}~\citep{WelBer86,KoeVar03b} for the decoding problem.

\subsection{Applications of the multivariate case} The case $\nvars > 1$ is
used for example in the interpolation steps of the list-decoding of
Parvaresh-Vardy codes~\citep{ParVar05} and of folded Reed-Solomon
codes~\citep{GurRud08}, as well as in Private Information
Retrieval~\citep{DeGoHe12}. In these contexts, one deals with
\Cref{pbm:multi-int} for some $\nvars>1$ with, in most cases, $\expSet =
\{ j \in \NN^\nvars \;\mid\; \sumVec{j} < a\}$ where $a$ is the list-size
parameter, and $\supp_k = \{ (i,j) \in \NN\times \NN^\nvars \;\mid\; i +
  \sumVec{j} < b \}$ where $b$ is the multiplicity parameter. Then, $\rdim =
  \binom{\nvars+a}{\nvars}$ and $\order = \binom{\nvars+b}{\nvars+1} \nbpt$.
  Besides, the weight (as mentioned in \Cref{subsec:multi-int}) is
  $\shifts[w] = (w,\ldots,w) \in \shiftSpace[\nvars]$ for a fixed positive
  integer $w$; then, the corresponding input shift is $\shifts = (\sumVec{j}
  w)_{j \in \expSet}$.

To our knowledge, the best known cost has been obtained by a probabilistic
algorithm which uses $\bigO{ \rdim^{\expmatmul-1} \polmultime{\order}
\log(\order)^2 }$ operations~\citep[Theorem~1]{CJNSV15} to compute one solution
with some degree constraints related to $\shifts$. However, this is not
satisfactory for the application to Private Information Retrieval, where one
wants an $\shifts$-minimal basis of solutions: using the polynomial matrix
reduction approach mentioned in \Cref{subsec:list-dec,subsec:soft-dec}, such a
basis can be computed in $\bigO{ \rdim^\expmatmul \polmultime{ b \nbpt } (
\log(\rdim)^2 + \log( b \nbpt ) ) }$ operations~\citep{PB08,KB10,CohHen12} (we
note that $\rdim b \nbpt > \order$). It is not clear to us what bound we have
on $\sshifts$ in this context, and thus how our algorithm compares to these two
results.

%
%

\section{A divide-and-conquer algorithm}\label{sec:algo}

In this section, we assume that $\mulmat \in \matSpace[\order]$ is a Jordan
matrix given by means of a standard representation as
in~\eqref{eq:mulmat_normal}, and we provide a description of our
divide-and-conquer algorithm together with a proof of the following refinement
of our main result, \Cref{thm:mib}.

\begin{proposition}
  \label{prop:mib}
  Without loss of generality, assume that $\min(\shifts)=0$; then, let $\xi =
  \sshifts[\shifts]$. \Cref{algo:mib} solves \Cref{pbm:mib}
  deterministically, and the sum of the $\shifts$-row degrees of the computed
  $\shifts$-minimal interpolation basis is at most $\order+\xi$. If
  $\order>\rdim$, a cost bound for this algorithm is given by
  \begin{align*}
  \bigO{ & \rdim^{\expmatmul-1} \polmultime{\order} + \rdim^\expmatmul \polmultime{\order/\rdim}
    \log(\order/\rdim)^2 + \rdim \polmultime{\order} \log(\order) \log(\order/\rdim) \\ 
    & + \rdim^{\expmatmul-1} \polmultime{\xi}  + \rdim^\expmatmul \polmultime{\xi/\rdim}
  \log(\xi/\rdim)}
    & \text{if } \expmatmul>2, \\[0.2cm]
    \bigO{ & \rdim \polmultime{\order} (\log(\rdim)^3 + \log(\order) \log(\order/\rdim)) + \rdim^2 \polmultime{\order/\rdim}
    \log(\order) \log(\order/\rdim) \log(\rdim) \\ 
    & + \rdim \polmultime{\xi} \log(\rdim)^2 + \rdim^2 \polmultime{\xi/\rdim}
  \log(\xi/\rdim) \log(\rdim) }
    & \text{if } \expmatmul=2;
  \end{align*}
  if $\order\le\rdim$, it is given in \Cref{prop:lin-mib}.
\end{proposition}

This algorithm relies on some subroutines, for which cost estimates are given
in the following sections and are taken for granted here:
\begin{itemize}
\item in \Cref{sec:cost-mult}, the fast multiplication of two
  polynomial matrices with respect to the average row degree of the operands
  and of the result;
\item in \Cref{sec:change-shift}, the \emph{change of shift}: given an
  $\shifts$-minimal interpolation basis $\mat{P}$ and some shift $\shifts[t]$,
  compute an $(\shifts+\shifts[t])$-minimal interpolation basis;
\item in \Cref{sec:multiplication}, the fast computation of a product of
  the form $\mat{P} \mul \evMat$;
\item in \Cref{sec:lin-mib}, the computation of an $\shifts$-minimal
  interpolation basis using linear algebra, which is used here for the base
  case of the recursion.
\end{itemize}

	\begin{figure}[h!]
	\centering
	\fbox{\begin{minipage}{12.5cm}
	\begin{algorithm}[\algoname{InterpolationBasis}]
    \label{algo:mib}

  ~\smallskip \\
	\algoword{Input:} 
  \begin{itemize}
    \item a matrix $\evMat\in\evSpace{\order}$,
    \item a Jordan matrix $\jordan \in \matSpace[\order]$ in standard representation,
    \item a shift $\shifts\in\shiftSpace$.
  \end{itemize}

  \smallskip
  \algoword{Output:} an $\shifts$-minimal interpolation basis $\intBasis \in
  \intSpace$ for $(\evMat,\mulmat)$.

  \smallskip
	\begin{enumerate}[{\bf 1.}] 
    \item \algoword{If} $\order \le \rdim$, \algoword{Return} $\algoname{LinearizationInterpolationBasis}(\evMat,\mulmat,\shifts,2^{\lceil\log(\order)\rceil})$
    \item \algoword{Else}
      \begin{enumerate}[{\bf a.}]
        \item $\intBasis \leftarrow \algoname{InterpolationBasisRec}(\jordan,\evMat)$
        \item \algoword{Return} $\algoname{Shift}(\intBasis, \unishift, \shifts)$
      \end{enumerate}
	\end{enumerate}
	\end{algorithm}

	\begin{algorithm}[\algoname{InterpolationBasisRec}]
  \label{algo:mib-rec}
  ~\smallskip\\
	\algoword{Input:}
    \begin{itemize}
    \item a matrix $\evMat\in\evSpace{\order}$,
    \item a Jordan matrix $\jordan \in \matSpace[\order]$ in standard representation.
    \end{itemize}

  \smallskip
  \algoword{Output:} a $\unishift$-minimal interpolation basis $\intBasis \in
  \intSpace$ for $(\evMat,\mulmat)$.

  \smallskip
	\begin{enumerate}[{\bf 1.}]
    \item \algoword{If} $\order \le \rdim$, \algoword{Return} $\algoname{LinearizationInterpolationBasis}(\evMat,\mulmat,\unishift,2^{\lceil\log(\order)\rceil})$
    \item \algoword{Else}
			\begin{enumerate}[{\bf a.}]
				\item $\evMat^{(1)} \leftarrow$ first $\lfloor\sigma/2\rfloor$ columns of $\evMat$
				\item $\intBasis^{(1)} \leftarrow \algoname{InterpolationBasisRec}(\evMat^{(1)},\mulmat^{(1)})$
        \item $\evMat^{(2)} \leftarrow$ last $\lceil\sigma/2\rceil$ columns of $\intBasis^{(1)} \mul \evMat = \algoname{ComputeResiduals}(\mulmat,\intBasis^{(1)},\evMat)$
				\item $\intBasis^{(2)} \leftarrow \algoname{InterpolationBasisRec}(\evMat^{(2)},\mulmat^{(2)})$ 
        \item $\mat{R}^{(2)} \leftarrow \algoname{Shift}(\intBasis^{(2)},\shifts[0],\rdeg{\intBasis^{(1)}})$
        \item \algoword{Return} $\algoname{UnbalancedMultiplication}(\intBasis^{(1)},\mat{R}^{(2)},\order)$
			\end{enumerate}
	\end{enumerate}
	\end{algorithm}
	\end{minipage}}
\end{figure}

First, we focus on the divide-and-conquer subroutine given in
\Cref{algo:mib-rec}. In what follows, $\jordan^{(1)}$ and
$\jordan^{(2)}$ always denote the leading and trailing principal $\order/2
\times \order/2$ submatrices of $\jordan$.  These two submatrices are still in
Jordan canonical form, albeit not necessarily in standard representation; this
can however be restored by a single pass through the array.

\begin{lemma}
	\label{lem:mib-rec}
  \Cref{algo:mib-rec} solves \Cref{pbm:mib} deterministically
  for the uniform shift $\shifts=\unishift$, and the sum of the row degrees of
  the computed $\unishift$-minimal interpolation basis is at most $\order$. 
  If $\order>\rdim$, a cost bound for this algorithm is given by
  \begin{align*}
  \bigO{ & \rdim^{\expmatmul-1} \polmultime{\order} + \rdim^\expmatmul \polmultime{\order/\rdim}
  \log(\order/\rdim)^2 + \rdim \polmultime{\order} \log(\order) \log(\order/\rdim) }
    & \text{if } \expmatmul>2, \\
    \bigO{ & \rdim \polmultime{\order} (\log(\rdim)^3 + \log(\order) \log(\order/\rdim) ) + \rdim^2 \polmultime{\order/\rdim}
    \log(\order) \log(\order/\rdim) \log(\rdim) }
    & \text{if } \expmatmul=2;
  \end{align*}
  if $\order\le\rdim$, it is given in \Cref{prop:lin-mib}.
\end{lemma}

In the rest of this paper, we use convenient notation for the cost of
polynomial matrix multiplication and for related quantities that arise when
working with submatrices of a given degree range as well as in
divide-and-conquer computations. Hereafter, $\log(\cdot)$ always stands for the
logarithm in base $2$.
\begin{definition}
  \label{dfn:polmatmul}
  Let $\rdim$ and $d$ be two positive integers. Then, $\polmatmultime{\rdim,d}$
  is such that two matrices in $\polMatSpace[\rdim]$ of degree at most $d$ can
  be multiplied using $\polmatmultime{\rdim,d}$ operations in $\field$. Then,
  writing $\bar{\rdim}$ and $\bar{d}$ for the smallest powers of $2$ at least
  $\rdim$ and $d$, we also define
  \begin{itemize}
    \item $\polmatmultimePrime{\rdim,d} = \sum_{0\le i \le \log(\bar{\rdim})} 2^i
      \polmatmultime{2^{-i}\bar{\rdim},2^i \bar{d}}$,
    \item $\polmatmultimePrimeDnc{\rdim,d} = \sum_{0\le i \le
        \log(\bar{\rdim})} 2^i \polmatmultimePrime{2^{-i}\bar{\rdim},2^i \bar{d}}$,
    \item $\polmatmultimeBis{\rdim,d} = \sum_{0\le i\le
        \log(\bar{d})} 2^i \polmatmultime{\bar{\rdim},2^{-i}\bar{d}}$,
    \item $\polmatmultimeBisDnc{\rdim,d} = \sum_{0\le i \le \log(\bar{\rdim})} 2^i
      \polmatmultimeBis{2^{-i}\bar{\rdim},2^i \bar{d}}$.
  \end{itemize}
\end{definition}
\noindent We note that one can always take $\polmatmultime{\rdim,d} \in \bigO{
  \rdim^\expmatmul \polmultime{d} }$. Upper bounds for the other quantities are
  detailed in \Cref{app:cost-mult}.

\begin{proof}[Proof of \Cref{lem:mib-rec}]
  Concerning the base case $\order \le \rdim$, the correctness and the cost
  bound both follow from results that will be given in \Cref{sec:lin-mib}, in
  particular \Cref{prop:lin-mib}.
  
  Let us now consider the case of $\order > \rdim$, and assume that
  $\intBasis^{(1)}$ and  $\intBasis^{(2)}$, as computed by the recursive calls,
  are $\unishift$-minimal interpolation bases for
  $(\evMat^{(1)},\jordan^{(1)})$ and $(\evMat^{(2)},\jordan^{(2)})$,
  respectively. The input $\evMat$ takes the form $\evMat = [\evMat^{(1)} |
  \;\boldsymbol{*}\; ]$ and we have $\intBasis^{(1)} \mul \evMat = [\;\mat{0}\;
  | \evMat^{(2)}]$ as well as $\intBasis^{(2)} \mul \evMat^{(2)} = 0$.
  Besides, $\mat{R}^{(2)}$ is by construction unimodularly equivalent to
  $\intBasis^{(2)}$, so there exists $\mat{U}$ unimodular such that
  $\mat{R}^{(2)} = \mat{U}\intBasis^{(2)}$. Defining $\intBasis =
  \mat{R}^{(2)}\intBasis^{(1)}$, our goal is to show that $\intBasis$ is a
  minimal interpolation basis for $(\evMat,\jordan)$.

  First, since $\jordan$ is upper triangular we have $\intBasis \mul \evMat =
  \mat{R}^{(2)} \mul [\;\mat{0}\: | \evMat^{(2)}] = [\;\mat{0}\; |
  \mat{U}\intBasis^{(2)}\mul\evMat^{(2)}] = 0$, so that every row of
  $\intBasis$ is an interpolant for $(\evMat,\jordan)$. Let us now consider an
  arbitrary interpolant $\int\in\polMatSpace[1][\rdim]$ for $(\evMat,\jordan)$.
  Then, $\int$ is in particular an interpolant for
  $(\evMat^{(1)},\jordan^{(1)})$: there exists some row vector $\row{V}$ such
  that $\int = \row{V}\intBasis^{(1)}$.  Furthermore, the equalities $0 =
  \row{p}\mul\evMat = \row{V}\intBasis^{(1)}\mul\evMat = [ \;\mat{0}\; |
  \row{v}\mul\evMat^{(2)}]$ show that $\row{v}\mul\evMat^{(2)}=0$.  Thus, there
  exists some row vector $\row{w}$ such that $\row{v}=\row{w}\intBasis^{(2)}$,
  which gives $\int = \row{w} \intBasis^{(2)} \intBasis^{(1)} = \row{w}
  \mat{U}^{-1} \intBasis$.  This means that every interpolant for
  $(\evMat,\jordan)$ is a $\polRing$-linear combination of the rows of
  $\intBasis$.

  Then, it remains to check that $\intBasis$ is $\unishift$-reduced.  As a
  $\unishift$-minimal interpolation basis, $\intBasis^{(2)}$ is
  $\unishift$-reduced and has full rank. Then, the construction of
  $\mat{R}^{(2)}$ using Algorithm \algoname{Shift} ensures that it is
  $\shifts[t]$-reduced, where $\shifts[t] = \rdeg{\intBasis^{(1)}}$. Define
  $\shifts[u] = \rdeg{\mat{P}}=\rdeg{ \mat{R}^{(2)} \intBasis^{(1)} }$; the
  predictable-degree property~\citep[Theorem~6.3-13]{Kailath80} implies that
  $\shifts[u] = \rdeg[{\shifts[t]}]{\mat{R}^{(2)}}$. Using the identity
  $\shiftMat{-\shifts[u]}\mat{P} = \shiftMat{-\shifts[u]} \mat{R}^{(2)}
  \shiftMat{\shifts[t]} \shiftMat{-\shifts[t]} \mat{P}^{(1)}$, we obtain that
  the $\unishift$-leading matrix of $\mat{P}$ is the product of the
  $\shifts[t]$-leading matrix of $\mat{R}^{(2)}$ and the $\unishift$-leading
  matrix of $\mat{P}^{(1)}$, which are both invertible. Thus, the
  $\unishift$-leading matrix of $\mat{P}$ is invertible as well, and therefore
  $\mat{P}$ is $\unishift$-reduced.

  Thus, for any $\order$, the algorithm correctly computes an $\shifts$-minimal
  interpolation basis for $(\evMat,\jordan)$.  As shown in \Cref{sec:intro},
  there is a direct link between M-Pad\'e approximation and shifted minimal
  interpolation bases with a multiplication matrix in Jordan canonical form.
  Then, the result in \citep[Theorem 4.1]{BarBul92} proves that for a given
  $\order$, the determinant of a $\unishift$-minimal interpolation basis
  $\intBasis$ has degree at most $\order$. Hence $\sumVec{\rdeg{\intBasis}} =
  \deg(\det(\intBasis)) \le \order$ follows the fact that the sum of the row
  degrees of a $\unishift$-reduced matrix equals the degree of its determinant.

  Let us finally prove the announced cost bound for $\order > \rdim$. Without
  loss of generality, we assume that $\order / \rdim$ is a power of $2$. Each
  of Steps~\textbf{2.b} and~\textbf{2.d} calls the algorithm recursively on an
  instance of \Cref{pbm:mib} with dimensions $\rdim$ and $\order/2$.
  \begin{itemize}
    \item The leaves of the recursion are for $\rdim/2 \le \order \le \rdim$
      and thus according to \Cref{prop:lin-mib} each of them uses
      $\bigO{\rdim^\expmatmul \log(\rdim)}$ operations if $\expmatmul>2$, and
      $\bigO{\rdim^2 \log(\rdim)^2}$ operations if $\expmatmul=2$. For
      $\order>\rdim$ (with $\order/\rdim$ a power of $2$), the recursion leads
      to $\order/\rdim$ leaves, which thus yield a total cost of $\bigO{
        \rdim^{\expmatmul-1} \order \log(\rdim)}$ operations if $\expmatmul>2$,
        and $\bigO{ \rdim \order \log(\rdim)^2}$ operations if $\expmatmul=2$.
    \item According to \Cref{prop:update-evaluations},
      Step~\textbf{2.c} uses $\bigO{ \polmatmultime{\rdim,\order/\rdim}
      \log(\order/\rdim) + \rdim \polmultime{\order} \log(\order) }$
      operations. Using the super-linearity property of $d \mapsto
      \polmatmultime{\rdim,d}$, we see that this contributes to the total cost
      as $ \bigO{ \polmatmultime{\rdim,\order/\rdim} \log(\order/\rdim)^2 +
      \rdim \polmultime{\order} \log(\order)\log(\order/\rdim) }$ operations.
    \item For Step~\textbf{2.e}, we use \Cref{prop:change_shift} with $\xi = \order$,
      remarking that both the sum of the entries of $\shifts[t] =
      \rdeg{\intBasis^{(1)}}$ and that of $\rdeg{\intBasis^{(2)}}$ are at most
      $\order/2$. Then, the change of shift is performed using $\bigO{
        \polmatmultimePrimeDnc{ \rdim, \order/\rdim } + \polmatmultimeBisDnc{
        \rdim, \order/\rdim } }$ operations. Thus, altogether the time spent in
        this step is $\bigO{ \sum_{0\le i\le \log(\order/\rdim)} 2^i
        (\polmatmultimePrimeDnc{ \rdim, 2^{-i}\order/\rdim } +
        \polmatmultimeBisDnc{ \rdim, 2^{-i}\order/\rdim } ) }$ operations; we
        give an upper bound for this quantity in
        \Cref{lem:polmatmulDoubleDnc-bound}.
    \item From \Cref{prop:change_shift} we obtain that
      $\rdeg[{\shifts[t]}]{\mat{R}^{(2)}} \le \sshifts[\rdeg{\intBasis^{(2)}}]
      + \sshifts[{\shifts[t]}] \le \order$. Then, using
      \Cref{prop:unbalanced-polmatmul} with $\xi = \order$, the polynomial
      matrix multiplication in Step~\textbf{2.f} can be done in time $\bigO{
        \polmatmultimePrime{ \rdim, \order/\rdim } }$. Besides, it is easily
        verified that by definition $\polmatmultimePrime{ \rdim, \order/\rdim }
        \le \polmatmultimePrimeDnc{ \rdim, \order/\rdim }$, so that the cost
        for this step is dominated by the cost for the change of shift.
  \end{itemize}
  Adding these costs and using the bounds in \Cref{app:cost-mult} leads
  to the conclusion.
\end{proof}

We now prove our main result.
\medskip
\begin{proof}[Proof of \Cref{prop:mib}]
  The correctness of \Cref{algo:mib} follows from the correctness of
  \Cref{algo:mib-rec,algo:shift,algo:lin-mib}. Concerning the cost bound when
  $\order > \rdim$, \Cref{lem:mib-rec} gives the number of operations used by
  Step~\textbf{2.a} to produce $\intBasis$, which satisfies
  $\sshifts[{\rdeg{\intBasis}}] \le \order$. Then, considering
  $\min(\shifts)=0$ without loss of generality, we have
  $\sshifts[\rdeg{\intBasis}] + \sshifts \le \order + \xi$:
  \Cref{prop:change_shift} states that Step~\textbf{2.b} can be performed using
  $\bigO{ \polmatmultimePrimeDnc{\rdim,(\order+\xi)/\rdim} +
  \polmatmultimeBisDnc{\rdim,(\order+\xi)/\rdim} }$ operations. The cost bound
  then follows from the bounds in \Cref{lem:polmatmulDnc-bound}. Furthermore,
  from \Cref{prop:change_shift} we also know that the sum of the $\shifts$-row
  degrees of the output matrix is exactly $\sshifts[{\rdeg{\intBasis}}] +
  \sshifts$, which is itself at most $\order + \xi$.
\end{proof}

%
%
\section{Multiplying matrices with unbalanced row degrees}
\label{sec:cost-mult}

In this section, we give a detailed complexity analysis concerning the fast
algorithm from~\citep[Section 3.6]{ZhLaSt12} for the multiplication of matrices
with controlled, yet possibly unbalanced, row degrees. For completeness, we
recall this algorithm below in \Cref{algo:unbalanced-polmatmul}. It is a
central building block in our algorithm: it is used in the multiplication of
interpolation bases (Step~\textbf{2.f} of \Cref{algo:mib-rec}) and also for the
multiplication of nullspace bases that occur in the algorithm
of~\citep{ZhLaSt12}, which we use to perform the change of shift in
Step~\textbf{2.e} of \Cref{algo:mib-rec}.

In the rest of the paper, most polynomial matrices are interpolation bases and
thus are square and nonsingular. In contrast, in this section polynomial
matrices may be rectangular and have zero rows, since this may occur in
nullspace basis computations (see \Cref{app:cost-mnb}). Thus, we extend the
definitions of shift and row degree as follows. For a matrix $\mat{A} =
[a_{i,j}]_{i,j}$ in $\polMatSpace[k][\rdim]$ for some $k$, and a shift $\shifts
= (\shift{1}, \ldots, \shift{\rdim}) \in (\NN\cup\{-\infty\})^\rdim$, the
$\shifts$-row degree of $\mat{A}$ is the tuple $\rdeg[\shifts]{\mat{A}} = (d_1,
\ldots , d_k) \in (\NN\cup\{-\infty\})^k$, with $d_i = \max_j (\deg(a_{i,j}) +
\shift{j})$ for all~$i$, and using the usual convention that $\deg(0) =
-\infty$. Besides, $\sshifts$ denotes the sum of the non-negative entries of
$\shifts$, that is, $\sshifts = \sum_{1\le i\le k, \shift{i} \neq -\infty}
\shift{i}$.

\medskip In order to multiply matrices with unbalanced row degrees, we use in
particular a technique based on \emph{partial linearization}, which can be seen
as a simplified version of the one in~\citep[Section~6]{GuSaStVa12} for the
purpose of multiplication. For a matrix $\mat{B}$ with sum of row degrees
$\xi$, meant to be the left operand in a product $\mat{B} \mat{A}$ for some
$\mat{A} \in \polMatSpace$, this technique consists in expanding the
high-degree rows of $\mat{B}$ so as to obtain a matrix $\widetilde{\mat{B}}$
with $\bigO{\rdim}$ rows and degree at most $\xi/\rdim$, then computing the
product $\widetilde{\mat{B}} \mat{A}$, and finally retrieving the actual
product $\mat{B}\mat{A}$ by grouping together the rows that have been expanded
(called \emph{partial compression} in what follows).

More precisely, let $\mat{B} \in \polMatSpace[k][\rdim]$ for some $k$ and
$\rdim$, with $\rdeg{\mat{B}} = (d_1,\ldots,d_k)$ and write $\xi = d_1 + \cdots
+ d_k$. We are given a target degree bound $d$. For each $i$, the row
$\matrow{\mat{B}}{i}$ of degree $d_i$ is expanded into $\alpha_i = 1 + \lfloor d_i /
(d+1) \rfloor$ rows $\matrow{\widetilde{\mat{B}}}{(i,0)},\ldots,\matrow{\widetilde{\mat{B}}}{(i,\alpha_i-1)}$ of
degree at most $d$, related by the identity
\begin{equation}
  \label{eqn:partial-lin}
  \matrow{\mat{B}}{i} = \matrow{\widetilde{\mat{B}}}{(i,0)} + X^{d+1}
  \matrow{\widetilde{\mat{B}}}{(i,1)} + \cdots + X^{(\alpha_i-1) (d+1)}
  \matrow{\widetilde{\mat{B}}}{(i,\alpha_i-1)}.
\end{equation}
Then, the expanded matrix $\widetilde{\mat{B}}$ has $\sum_{1\le i\le k}
\alpha_i \le k + \xi / (d+1)$ rows $\matrow{\widetilde{\mat{B}}}{(i,j)}$. We
will mainly use this technique for $k \le \rdim$ and $d = \lfloor \xi / \rdim
\rfloor$ or $d = \lceil \xi / \rdim \rceil$, in which case
$\widetilde{\mat{B}}$ has fewer than $2\rdim$ rows. The partial compression is
the computation of the row $i$ of the product $\mat{B} \mat{A}$ from the rows
$(i,0), \ldots, (i,\alpha_i-1)$ of $\widetilde{\mat{B}} \mat{A}$ using the
formula in~\eqref{eqn:partial-lin}. 

\begin{proposition}
  \label{prop:unbalanced-polmatmul}
  Let $\mat{A}$ and $\mat{B}$ in $\polMatSpace[\rdim]$, and $\shifts[d] =
  \rdeg{\mat{A}}$. Let $\xi \ge \rdim$ be an integer such that
  $\sshifts[{\shifts[d]}] \le \xi$ and
  $\sshifts[{\rdeg[{\shifts[d]}]{\mat{B}}}] \le \xi$. Then, the product
  $\mat{B}\mat{A}$ can be computed using 
  \begin{align*}
    \bigO{ & \polmatmultimePrime{\rdim,\xi/\rdim} }  \\
    & \subseteq \bigO{ \rdim^{\expmatmul-1} \polmultime{\xi} } & \text{if } \expmatmul>2 \\
    & \subseteq \bigO{ \rdim \polmultime{\xi} \log(\rdim) } & \text{if } \expmatmul=2
  \end{align*}
  operations in $\field$. 
\end{proposition}
\begin{proof}
  In this proof, we use notation from \Cref{algo:unbalanced-polmatmul}. The
  correctness of this algorithm follows from the identity $\mat{B}\mat{A} =
  \mat{\hat{B}}\mat{\hat{A}} = \mat{B}_0\mat{A}_0 + \mat{B}_1\mat{A}_1 + \cdots
  + \mat{B}_\ell \mat{A}_\ell$. In what follows we focus on proving the cost
  bound $\bigO{\polmatmultimePrime{\rdim,\xi/\rdim} }$; the announced upper
  bounds on this quantity follow from \Cref{lem:polmatmul-bound}.

  We start with Step~\textbf{6}, which adds the $\rdim \times \rdim$ matrices
  $\mat{P}_i = \mat{B}_i \mat{A}_i$ obtained after the first five steps. For
  each $i$ in $\{0,\ldots,\ell\}$, we have $\rdeg{\mat{P}_i} \le
  \rdeg{\mat{B}\mat{A}} \le \rdeg[{\shifts[d]}]{\mat{B}}$ componentwise, hence
  $\sshifts[{\rdeg{\mat{P}_i}}] \le \xi$. Recalling that $\ell = \lceil
  \log(\rdim) \rceil$, the sum at Step~\textbf{6} thus uses $\bigO{\rdim \xi
  \ell} = \bigO{\rdim \xi \log(\rdim)}$ additions in $\field$. On the other
  hand, by definition of $\polmatmultimePrime{\cdot,\cdot}$, the trivial lower
  bound $\polmatmultime{\cdim,d} \ge \cdim^2 d$ for any $\cdim,d$ implies that
  $\rdim \xi \log(\rdim) \in \bigO{\polmatmultimePrime{\rdim,\xi/\rdim}}$.

  Now we study the \algoword{for} loop. We remark that only Step~\textbf{5.c}
  involves arithmetic operations in $\field$. Therefore the main task is to
  give bounds on the dimensions and degrees of the matrices we multiply at
  Step~\textbf{5.c}. For $i$ in $\{0,\ldots,\ell\}$, the column dimension of
  $\mat{A}_i$ is $\rdim$ and the row dimension of $\mat{B}^\emptyset_i$ is
  $k_i$. We further denote by $\rdim_i$ the row dimension of $\mat{A}_i$ (and
  column dimension of $\widetilde{\mat{B}}^\emptyset_i$), and we write
  $[\shift[d]{\pi(1)},\ldots,\shift[d]{\pi(\rdim)}] = [ \shifts[d]_0 | \cdots |
  \shifts[d]_\ell ]$ where the sizes of $\shifts[d]_0,\ldots,\shifts[d]_\ell$
  correspond to those of the blocks of $\mat{\hat{B}}$ as in Step~\textbf{4} of
  the algorithm.

  First, let $i=0$. Then, $\mat{A}_0$ is $\rdim_0 \times \rdim$ of degree at
  most $\xi/\rdim$ and $\mat{B}_0$ is $k_0 \times \rdim_0$ with $\rdim_0 \le
  \rdim$ and $k_0 \le \rdim$ (we note that for $i=0$ these may be equalities
  and thus one does not need to discard the zero rows of $\mat{B}_0$ to obtain
  efficiency). Besides, we have the componentwise inequality $\rdeg{\mat{B}_0}
  \le \rdeg[{\shifts[d]_0}]{\mat{B}_0} \le \rdeg[{\shifts[d]}]{\mat{B}}$, so
  that $\sshifts[{\rdeg{\mat{B}_0}}] \le
  \sshifts[{\rdeg[{\shifts[d]}]{\mat{B}}}] \le \xi$. Then,
  $\mat{B}^\emptyset_0$ can be partially linearized into a matrix
  $\widetilde{\mat{B}}^\emptyset_0$ which has at most $2\rdim$ rows and degree
  at most $\xi/\rdim$, and the computation at Step~\textbf{5.c} for $i=0$ uses
  $\bigO{\polmatmultime{\rdim,\xi/\rdim}}$ operations.

  Now, let $i\in \{1,\ldots,\ell\}$. By assumption, the sum of the row degrees
  of $\mat{A}$ does not exceed $\xi$: since all rows in $\mat{A}_i$ have degree
  more than $2^{i-1}\xi/\rdim$, this implies that $\rdim_i < \rdim / 2^{i-1}$.
  Besides, since $\min(\shifts[d]_i) > 2^{i-1}\xi/\rdim$, we obtain that every
  nonzero row of $\mat{B}_i$ has $\shifts[d]_i$-row degree more than
  $2^{i-1}\xi/\rdim$. Then, $\xi \ge \sshifts[{\rdeg[{\shifts[d]}]{\mat{B}}}]
  \ge \sshifts[{\rdeg[{\shifts[d]_i}]{\mat{B}_i^\emptyset}}] > k_i
  2^{i-1}\xi/\rdim$ implies that $k_i < \rdim/2^{i-1}$. Furthermore, since we
  have $\sshifts[{\rdeg{\mat{B}^\emptyset_i}}] = \sshifts[{\rdeg{\mat{B}_i}}]
  \le \xi$, the partial linearization at Step \textbf{5.b} can be done by at
  most doubling the number of rows of $\mat{B}^\emptyset_i$, producing
  $\widetilde{\mat{B}}^\emptyset_i$ with fewer than $2 \rdim/2^{i-1}$ rows and
  of degree at most $2^{i-1}\xi/\rdim$. To summarize: $\mat{A}_i$ has $\rdim$
  columns, $\rdim_i < \rdim/2^{i-1}$ rows, and degree at most $2^i\xi/\rdim$;
  $\widetilde{\mat{B}}^\emptyset_i$ has fewer than $2 \rdim/2^{i-1}$ rows, and
  degree less than $2^i\xi/\rdim$.  Then, the computation of
  $\widetilde{\mat{P}}^\emptyset_i$ uses $\bigO{ 2^i
    \polmatmultime{\rdim/2^{i-1}, 2^i\xi/\rdim} }$ operations in $\field$.
    Thus, overall the \algoword{for} loop uses $\bigO{
      \polmatmultimePrime{\rdim, \xi/\rdim} }$ operations in $\field$.
\end{proof}

\begin{figure}
\centering
\fbox{\begin{minipage}{12.5cm}
\begin{algorithm} [\algoname{UnbalancedMultiplication}]
\label{algo:unbalanced-polmatmul}

~\smallskip\\
\algoword{Input:}
\begin{itemize}
  \item polynomial matrices $\mat{A}$ and $\mat{B}$ in $\polMatSpace[\rdim]$,
  \item an integer $\xi$ with $\xi \ge \rdim$, $\sshifts[{\shifts[d]}]
    \le \xi$, and $\sshifts[{\rdeg[{\shifts[d]}]{\mat{B}}}] \le \xi$ where
    $\shifts[d] = \rdeg{\mat{A}}$.
\end{itemize}

\smallskip
\algoword{Output:} 
the product $\mat{P} = \mat{B}\mat{A}$.

\smallskip
\begin{enumerate}[{\bf 1.}]
  \item $\pi$ $\leftarrow$ a permutation of $\{1,\ldots,\rdim\}$ such that
    $(\shift[d]{\pi(1)},\ldots,\shift[d]{\pi(\rdim)})$ is non-decreasing
  \item $\mat{\hat{A}} \leftarrow \pi \mat{A}$ and $\mat{\hat{B}} \leftarrow \mat{B} \pi^{-1}$
  \item define $\ell = \lceil \log(\rdim) \rceil$ and the row blocks
    $\mat{\hat{A}} = \trsp{[ \trsp{\mat{A}_0} | \trsp{\mat{A}_1} | \cdots |
    \trsp{\mat{A}_\ell} ]}$, \\ where the rows in $\mat{A}_0$ have row degree
    at most $\xi/\rdim$ and for $i=1,\ldots,\ell$ the rows in $\mat{A}_i$ have
    row degree in $\{2^{i-1}\xi/\rdim+1, \dots, 2^i\xi/\rdim\}$
  \item define $\mat{\hat{B}} = 
    [ \mat{B}_0 | \mat{B}_1 | \cdots | \mat{B}_\ell ]$
    the corresponding column blocks of $\mat{\hat{B}}$
  \item \algoword{For} $i$ \algoword{from} $0$ \algoword{to} $\ell$:
    \begin{enumerate}[{\bf a.}]
      \item Read $r_1,\ldots,r_{k_i}$ the indices of the nonzero rows in $\mat{B}_i$ \\
        and define $\mat{B}^{\emptyset}_i$ the submatrix of $\mat{B}_i$ obtained
        by removing the zero rows
      \item 
        $\widetilde{\mat{B}}^\emptyset_i \leftarrow$ partial
        linearization of $\mat{B}^\emptyset_i$ with
        $\deg(\widetilde{\mat{B}}^\emptyset_i) \le 2^i \xi/\rdim$
      \item 
        $\widetilde{\mat{P}}^\emptyset_i \leftarrow
        \widetilde{\mat{B}}^\emptyset_i \mat{A}_i$
      \item Perform the partial compression
        $\mat{P}^\emptyset_i \leftarrow \widetilde{\mat{P}}^\emptyset_i$
      \item Re-introduce the zero rows to obtain $\mat{P}_i$, which is $\mat{B}_i \mat{A}_i$ \\
        (its rows at indices $r_1,\ldots,r_{k_i}$ are those of
        $\mat{P}^\emptyset_i$, its other rows are zero)
    \end{enumerate}
  \item \algoword{Return} $\mat{P} = \mat{P}_0 + \mat{P}_1 + \cdots + \mat{P}_\ell$
\end{enumerate}
\end{algorithm}
\end{minipage}}
\end{figure}

\section{Fast shifted reduction of a reduced matrix}\label{sec:change-shift} 

In our interpolation \cref{algo:mib}, a key ingredient to achieve efficiency is
to control the size of the intermediate interpolation bases that are computed
in recursive calls. For this, we compute all minimal bases for the
\emph{uniform} shift and then recover the \emph{shifted} minimal basis using
what we call a change of shift, that we detail in this section. More precisely,
we are interested in the ability to transform an $\shifts$-reduced matrix
$\mat{P} \in \polMatSpace[\rdim]$ with full rank into a unimodularly equivalent
matrix that is $(\shifts+\shifts[t])$-reduced for some given shift $\shifts[t]
\in \shiftSpace$: this is the problem of polynomial lattice reduction for the
shift $\shifts+\shifts[t]$, knowing that the input matrix is already reduced
for the shift $\shifts$.

Compared to a general row reduction algorithm such as the one
in~\citep{GuSaStVa12}, our algorithm achieves efficient computation with
regards to the average row degree of the input $\mat{P}$ rather than the
maximum degree of the entries of $\mat{P}$. The main consequence of having an
\emph{$\shifts$-reduced} input $\mat{P}$ is that no high-degree cancellation
can occur when performing unimodular transformations on the rows of $\mat{P}$,
which is formalized as the predictable-degree property~\citep[Theorem
6.3-13]{Kailath80}. In particular, the unimodular transformation between
$\mat{P}$ and an $(\shifts+\shifts[t])$-reduced equivalent matrix has small row
degree, and the proposition below shows how to exploit this to solve our
problem via the computation of a shifted minimal (left) nullspace basis of some
$2\rdim \times \rdim$ polynomial matrix. We remark that similar ideas about the
use of minimal nullspace bases to compute reduced forms were already
in~\citep[Section~3]{BeHuPr88}.

\begin{lemma}
	\label{lem:change-shift}
  Let $\shifts\in\shiftSpace$ and $\shifts[t]\in\shiftSpace$, let
  $\mat{P}\in\polMatSpace[\rdim]$ be $\shifts$-reduced and nonsingular, and
  define $\shifts[d]=\rdeg[\shifts]{\intBasis}$.  Then $\mat{R} \in
  \polMatSpace[\rdim]$ is an $(\shifts+\shifts[t])$-reduced form of $\mat{P}$
  with unimodular transformation $\mat{U} = \mat{R}
  \mat{P}^{-1}\in\polMatSpace[\rdim]$ if and only if $[ \mat{U} |
  \mat{R}\shiftMat{\shifts} ]$ is a $(\shifts[d],\shifts[t])$-minimal nullspace
  basis of $\trsp{[\shiftMat{\shifts}\trsp{\mat{P}} | -\idMat[\rdim]]}$.
\end{lemma}
\begin{proof}
  We first assume that the result holds for the uniform shift $\shifts =
  \unishift\in\shiftSpace$, and we show that the general case
  $\shifts\in\shiftSpace$ follows. Indeed, considering the $\shifts[0]$-reduced
  matrix $\mat{P}\shiftMat{\shifts}$ we have
  $\shifts[d]=\rdeg[\shifts]{\mat{P}}=\rdeg{\mat{P}\shiftMat{\shifts}}$.  Hence
  $[ \mat{U} | \mat{R} ]$ is a $(\shifts[d],\shifts[t])$-minimal nullspace
  basis of $\trsp{[\shiftMat{\shifts}\trsp{\mat{P}} | -\idMat[\rdim]]}$ if and
  only if $\mat{R}$ is a $\shifts[t]$-reduced form of
  $\mat{P}\shiftMat{\shifts}$ with unimodular transformation $\mat{U}$ such
  that $\mat{U} \mat{P}\shiftMat{\shifts} = \mat{R}$; that is, if and only if
  $\mat{R}\shiftMat{-\shifts}\in\polMatSpace[\rdim]$ is a
  $(\shifts+\shifts[t])$-reduced form of $\mat{P}$ with unimodular
  transformation $\mat{U}$ such that $\mat{U}\mat{P} =
  \mat{R}\shiftMat{-\shifts}$.

  Let us now prove the proposition for the uniform shift $\shifts =
  \shifts[0]$.  First, we assume that $\mat{R}\in\polMatSpace[\rdim]$ is a
  $\shifts[t]$-reduced form of $\mat{P}$ with unimodular transformation
  $\mat{U}$.  From $\mat{U}\mat{P} = \mat{R}$ it follows that the rows of
  $[\mat{U} | \mat{R} ]$ are in the nullspace of
  $\trsp{[\trsp{\mat{P}} | -\idMat[\rdim]]}$.  Writing
  $[\mat{N}|\,*\,]$ with $\mat{N}\in\polMatSpace[\rdim]$ to denote an arbitrary
  basis of that nullspace, we have $[\mat{U} | \mat{R} ] = \mat{V}
  [\mat{N}|\,*\,]$ for some $\mat{V} \in \polMatSpace[\rdim]$ and thus $\mat{U}
  = \mat{V} \mat{N}$. Since $\mat{U}$ is unimodular, $\mat{V}$ is unimodular
  too and $[\mat{U} | \mat{R} ]$ is a basis of the nullspace of
  $\trsp{[\trsp{\mat{P}} | -\idMat[\rdim]]}$.  It remains to check
  that $[\mat{U} | \mat{R} ]$ is $(\shifts[d],\shifts[t])$-reduced.  Since
  $\mat{P}$ is reduced, we have $\rdeg[{\shifts[d]}\,]{\mat{U}}=
  \rdeg{\mat{U}\mat{P}} = \rdeg{\mat{R}}$ by the predictable-degree
  property~\citep[Theorem 6.3-13]{Kailath80} and, using $\shifts[t] \ge
  \shifts[0]$, we obtain $\rdeg[{\shifts[d]}\,]{\mat{U}} \le
  \rdeg[{\shifts[t]}\,]{\mat{R}}$.  Hence
  $\rdeg[{(\shifts[d],\shifts[t])}]{[\mat{U} | \mat{R} ]} =
  \rdeg[{\shifts[t]}\,]{\mat{R}}$ and, since $\mat{R}$ is
  $\shifts[t]$-reduced, this implies that $[\mat{U} | \mat{R} ]$ is
  $(\shifts[d],\shifts[t])$-reduced.

  Now, let $[ \mat{U} | \mat{R} ]$ be a $(\shifts[d],\shifts[t])$-minimal
  nullspace basis of $\trsp{[\trsp{\mat{P}} | -\idMat[\rdim]]}$.  First, we
  note that $\mat{U}$ satisfies $\mat{U}= \mat{R}\mat{P}^{-1} $.  It remains to
  check that $\mat{U}$ is unimodular and that $\mat{R}$ is
  $\shifts[t]$-reduced.  To do this, let $\widehat{\mat{R}}$ denote an
  arbitrary $\shifts[t]$-reduced form of $\mat{P}$ and let $\widehat{\mat{U}}=
  \widehat{\mat{R}} \mat{P}^{-1}$ be the associated unimodular transformation.
  From the previous paragraph, we know that $[ \widehat{\mat{U}} |
  \widehat{\mat{R}} ]$ is a basis of the nullspace of $\trsp{[\trsp{\mat{P}} |
  -\idMat[\rdim]]}$, and since by definition $[ \mat{U} | \mat{R} ]$ is also
  such a basis, we have $[\mat{U} | \mat{R}] = \mat{W} [\widehat{\mat{U}} |
  \widehat{\mat{R}} ]$ for some unimodular matrix
  $\mat{W}\in\polMatSpace[\rdim]$. In particular, $\mat{U} =
  \mat{W}\widehat{\mat{U}}$ is unimodular.  Furthermore, the two unimodularly
  equivalent matrices $[\mat{U} | \mat{R}]$ and $[\widehat{\mat{U}} |
  \widehat{\mat{R}}]$ are $(\shifts[d],\shifts[t])$-reduced, so that they share
  the same shifted row degree up to permutation (see for instance~\citep[Lemma
  6.3-14]{Kailath80}). Now, from the previous paragraph, we know that
  $\rdeg[{(\shifts[d],\shifts[t])}]{[\widehat{\mat{U}} | \widehat{\mat{R}}]} =
  \rdeg[{\shifts[t]}\,]{\widehat{\mat{R}}}$, and similarly, having $\mat{P}$
  reduced, $\mat{U}\mat{P} = \mat{R}$, and $\shifts[t]\ge
  \shifts[0]$ imply that $\rdeg[{(\shifts[d],\shifts[t])}]{[\mat{U} | \mat{R}]}
  = \rdeg[{\shifts[t]}\,]{\mat{R}}$.  Thus $\rdeg[{\shifts[t]}\,]{\mat{R}}$ and
  $\rdeg[{\shifts[t]}\,]{\widehat{\mat{R}}}$ are equal up to permutation, and
  combining this with the fact that $\mat{R}=\mat{W}\widehat{\mat{R}}$ where
  $\widehat{\mat{R}}$ is $\shifts[t]$-reduced and $\mat{W}$ is unimodular, we
  conclude that $\mat{R}$ is $\shifts[t]$-reduced.
\end{proof}

This leads to \Cref{algo:shift}, and in particular such a change of shift can
be computed efficiently using the minimal nullspace basis algorithm of
\citet{ZhLaSt12}.

	\begin{figure}[h!]
	\centering
	\fbox{\begin{minipage}{12.5cm}
  \begin{algorithm} [\algoname{Shift}]
	\label{algo:shift}

  ~\smallskip\\
  \algoword{Input:}
  \begin{itemize}
    \item a matrix $\intBasis\in\intSpace$ with full rank,
    \item two shifts $\shifts,\shifts[t]\in\shiftSpace$ such that $\intBasis$
      is $\shifts$-reduced.
  \end{itemize}

  \smallskip
	\algoword{Output:} 
		an $(\shifts+\shifts[t])$-reduced form of $\intBasis$.

  \smallskip
	\begin{enumerate}[{\bf 1.}]
		\item $\shifts[d] \leftarrow \rdeg[\shifts]{\intBasis}$
    \item $[\mat{U} | \mat{R}] \leftarrow$ $\algoname{MinimalNullspaceBasis}(
      \trsp{[\shiftMat{\shifts}\trsp{\mat{P}} | -\idMat[\rdim]]},
      (\shifts[d],\shifts[t]) )$
    \item \algoword{Return} $\mat{R}\shiftMat{-\shifts}$
	\end{enumerate}
	\end{algorithm}
	\end{minipage}}
	\end{figure}

\begin{proposition} \label{prop:change_shift}
  Let $\shifts\in\shiftSpace$ and $\shifts[t]\in\shiftSpace$, let
  $\mat{P}\in\polMatSpace[\rdim]$ have full rank and be $\shifts$-reduced, and
  define $\shifts[d]=\rdeg[\shifts]{\mat{P}}$. We write $\xi$ to denote a
  parameter such that $\xi \ge \rdim$ and $\sshifts[{\shifts[d]}] +
  \sshifts[{\shifts[t]}] \le \xi$.  Then, an $(\shifts+\shifts[t])$-reduced
  form $\mat{R}\in\polMatSpace[\rdim]$ of $\mat{P}$ and the corresponding
  unimodular transformation $\mat{U} = \mat{R}\mat{P}^{-1} \in
  \polMatSpace[\rdim]$ can be computed using
  \begin{align*}
    \bigO{ & \polmatmultimePrimeDnc{\rdim,\xi/\rdim} +
    \polmatmultimeBisDnc{\rdim,\xi/\rdim} }   \\
    & \subseteq\; \bigO{ \rdim^{\expmatmul-1} \polmultime{\xi} +
    \rdim^{\expmatmul} \polmultime{\xi/\rdim} \log(\xi/\rdim) } & \text{if } \expmatmul>2 \\
    & \subseteq\; \bigO{ \rdim \polmultime{\xi} \log(\rdim)^2 + \rdim^2
    \polmultime{\xi/\rdim} \log(\xi/\rdim) \log(\rdim) } & \text{if } \expmatmul=2
  \end{align*}
  operations in $\field$. Besides, we have
  $\sshifts[{\rdeg[{\shifts+\shifts[t]}]{\mat{R}}}] = \sshifts[{\shifts[d]}] +
  \sshifts[{\shifts[t]}]$.
\end{proposition}
\begin{proof}
  Write $\shifts[u] = (\shifts[d],\shifts[t])$ and
  $\mat{M}=\trsp{[\shiftMat{\shifts}\trsp{\mat{P}} | -\idMat[\rdim]]}$.
  According to \Cref{lem:change-shift}, \Cref{algo:shift} is
  correct: it computes $[\mat{U}|\mat{R}]$ a $\shifts[u]$-minimal nullspace
  basis of $\mat{M}$, and returns $\mat{R}\shiftMat{-\shifts}$ which is an
  $(\shifts+\shifts[t])$-reduced form of $\mat{P}$. For a fast solution, the
  minimal nullspace basis can be computed using~\citep[Algorithm 1]{ZhLaSt12},
  which we have rewritten in \Cref{app:cost-mnb}
  (\Cref{algo:mnb}) along with a detailed cost analysis.

  Here, we show that the requirements of this algorithm on its input are
  fulfilled in our context. Concerning the input matrix, we note that $\mat{M}$
  has more rows than columns, and $\mat{M}$ has full rank since by assumption
  $\mat{P}$ has full rank. Now, considering the requirement on the input shift,
  first, each element of the shift $\shifts[u]$ bounds the corresponding row
  degree of $\mat{M}$; and second, the rows of $\mat{M}$ can be permuted before
  the nullspace computation so as to have $\shifts[u]$ non-decreasing, and then
  the columns of the obtained nullspace basis can be permuted back to the
  original order. In details, we first compute $\shifts[v]$ being the tuple
  $\shifts[u]$ sorted in non-decreasing order together with the corresponding
  permutation matrix $\pi\in\matSpace[2\rdim]$ such that, when $\shifts[v]$ and
  $\shifts[u]$ are seen as column vectors in $\NN^{2\rdim\times 1}$, we have
  $\shifts[v] = \pi\shifts[u]$. Now that $\shifts[v]$ is non-decreasing and
  bounds the corresponding row degree of $\pi\mat{M}$, we compute $\mat{N}$ a
  $\shifts[v]$-minimal nullspace basis of $\pi\mat{M}$ using
  \Cref{algo:mnb}, then, $\mat{N}\pi$ is a $\shifts[u]$-minimal
  nullspace basis of $\mat{M}$. Since by assumption $|\shifts[v]| =
  \sshifts[{\shifts[d]}] + \sshifts[{\shifts[t]}] \le \xi$, the announced cost
  bound follows directly from \Cref{prop:cost-mnb} in
  \Cref{app:cost-mnb}.

  Finally, we prove the bound on the sum of the $(\shifts + \shifts[t])$-row
  degrees of $\mat{R}$. Since $\mat{P}$ is $\shifts$-reduced and $\mat{R}$ is
  $(\shifts + \shifts[t])$-reduced, we have $\sshifts[{\shifts[d]}] =
  \deg(\det(\mat{P} \shiftMat{\shifts}))$ as well as
  $\sshifts[{\rdeg[{\shifts+\shifts[t]}]{\mat{R}}}] = \deg(\det(\mat{R}
  \shiftMat{\shifts+\shifts[t]}))$ \citep[Section~6.3.2]{Kailath80}. Then, we
  have that $\sshifts[{\rdeg[{\shifts+\shifts[t]}]{\mat{R}}}] =
  \deg(\det(\mat{U} \mat{P} \shiftMat{\shifts+\shifts[t]})) = \deg(\det(\mat{P}
  \shiftMat{\shifts})) + \sshifts[{\shifts[t]}] = \sshifts[{\shifts[d]}] +
  \sshifts[{\shifts[t]}]$, which concludes the proof.
\end{proof}


\section{Computing residuals}
\label{sec:multiplication} 
        
Let $\vecSpace = \field^{1\times \order}$ and $\mulmat \in \matSpace[\order]$
be as in the introduction; in particular, we suppose that $\mulmat$ is a Jordan
matrix, given by a standard representation as in~\eqref{eq:mulmat_normal}.
Given $\evMat$ in $\vecSpace^\rdim = \evSpace{\order}$ and a matrix $\mat{P}$
in $\polMatSpace[\rdim]$, we show how to compute the product $\mat{P} \mul
\evMat \in \vecSpace^\rdim$. We will often call the result \emph{residual}, as
this is the role this vector plays in our main algorithm.

To give our complexity estimates, we will make two assumptions, namely that
$\rdim \le \order$ and that the sum of the row degrees of $\mat{P}$ is in
$\bigO{\order}$; they will both be satisfied when we apply the following
result.

\begin{proposition}
  \label{prop:update-evaluations}
  There exists an algorithm \algoname{ComputeResiduals} that computes the
  matrix $\mat{P}\mul \evMat \in \vecSpace^\rdim$. If $\order \ge \rdim$ and if
  the sum of the row degrees of $\mat{P}$ is $\bigO{\order}$, this algorithm
  uses $\bigO{ \polmatmultime{ \rdim, \order/\rdim } \log(\order/\rdim) + \rdim
  \polmultime{\order} \log(\order) }$ operations in $\field$.
\end{proposition}
Remark that when the sum of the row degrees of $\mat{P}$ is $\bigO{\order}$,
storing $\mat{P}$ requires $\bigO{\rdim \order}$ elements in $\field$, so that
representing the input and output of this computation involves $\bigO{\rdim
\order}$ field elements. At best, one could thus hope for an algorithm of cost
$\bigO{\rdim \order}$. Our result is close, as we get a cost of
$\softO{\rdim^{1.38}\order}$ with the best known value of $\expmatmul$.


\subsection{Preliminaries, Chinese remaindering, and related questions}

The following lemma writes the output in a more precise manner. The proof is a
straightforward consequence of the discussion in \Cref{sec:intro} about
writing the notion of interpolant in terms of M-Pad\'e approximation.

\begin{lemma}
  \label{lem:evaluations}
  Suppose that $\mulmat$ has the form $ ( (\evpt_1,\szbl_1), \ldots,
  (\evpt_\nbbl,\szbl_\nbbl))$. Let $\mat{P} \in \polMatSpace[\rdim]$ and
  $\evMat \in \matSpace[\rdim][\order]$, and write $\evMat = [\evMat_1 | \cdots
  | \evMat_\nbbl] $ with $\evMat_j$ in $\matSpace[\rdim][\szbl_j]$ for $1\le
  j\le \nbbl$. For $1\le j\le\nbbl$, define the following matrices:
  \begin{itemize}
  \item $\evMat_{j,{\rm poly}} = \evMat_j\,
    \trsp{[1,\var,\ldots,\var^{\szbl_j-1}]} \in \polMatSpace[\rdim][1]$ is the
    column vector with polynomial entries built from the columns of $\evMat_j$,
  \item $\evMatF_{j,{\rm poly}} = \mat{P}(\var+\evpt_j)\,\evMat_{j,{\rm poly}} \bmod X^{\sigma_j}
    \in \polMatSpace[\rdim][1]$,
  \item $\evMatF_j = [\evMatF_{j,0}, \dots,
    \evMatF_{j,\szbl_j-1}] \in \matSpace[\rdim][\szbl_j]$ is the matrix whose
    columns are the coefficients of $\evMatF_{j,{\rm poly}}$ of degrees
    $0,\dots,\szbl_j-1$.
  \end{itemize}
 Then, $\mat{P}\mul\evMat = \evMatF$ with $\evMatF = [\evMatF_1 | \cdots |
 \evMatF_\nbbl] \in\matSpace[\rdim][\order]$.
\end{lemma}

To give an idea of our algorithm's behaviour, let us first consider the case
where $\mulmat$ is the upper shift matrix $\mulshift$ as
in~\eqref{eqn:mulshift}, so there is only one Jordan block whose eigenvalue is
$0$. This corresponds to having $\nbbl=1$ in the previous lemma, which thus
says that we can turn the input $\evMat$ into a vector of $\rdim$ polynomials
of degree at most $\order$, and that we simply have to left-multiply this
vector by~$\mat{P}$. Suppose furthermore that all entries in $\mat{P}$ have
degree $\bigO{\order/\rdim}$ (this is the most natural situation ensuring that
the sum of its row degrees is $\bigO{\order}$, as assumed in
\Cref{prop:update-evaluations}), so that we have to multiply an $\rdim \times
\rdim$ matrix with entries of degree $\bigO{\order/\rdim}$ by an $\rdim \times
1$ vector with entries of degree $\order$. For this, we use the partial
linearization presented in \Cref{sec:cost-mult}: we expand the right-hand side
into an $\rdim \times \rdim$ polynomial matrix with entries of degree
$\bigO{\order/\rdim}$, we multiply it by $\mat{P}$, and we recombine the
entries of the result; this leads us to the cost $\bigO{ \polmatmultime{ \rdim,
\order/\rdim } }$.

On the other side of the spectrum, we encountered the case of a diagonal matrix
$\mulmat$, with diagonal entries $\evpt_1,\dots,\evpt_\order$ (so all
$\order_i$'s are equal to $1$); suppose furthermore that these entries are
pairwise distinct. In this case, if we let $\evMat_1,\dots,\evMat_\order$ be
the columns of $\evMat$, \Cref{lem:evaluations} shows that the output is
the matrix whose columns are $\mat{P}(\evpt_1)
\evMat_1,\dots,\mat{P}(\evpt_\order) \evMat_\order$. Evaluating $\mat{P}$ at
all $\evpt_i$'s would be too costly, as simply representing all the evaluations
requires $\rdim^2 \order$ field elements; instead, we interpolate a column
vector of $\rdim$ polynomials $E_1,\dots,E_\rdim$ of degree less than $\order$
from the respective rows of $\evMat$, do the same matrix-vector product as
above, and evaluate the output at the $\evpt_i$'s; the total cost is $\bigO{
  \polmatmultime{ \rdim, \order/\rdim } + \rdim \polmultime{\order}\log(\order)
}$.

Our main algorithm generalizes these two particular processes. We now state a
few basic results that will be needed for this kind of calculation, around
problems related to polynomial modular reduction and Chinese remaindering.

\begin{lemma}\label{lemma:multirem}
  The following cost estimates hold:
  \begin{itemize}
  \item Given $p$ of degree $d$ in $\polRing$, and $x$ in $\field$, one can
    compute $p(X+x)$ in $\bigO{\polmultime{d} \log(d)}$ operations in $\field$.
  \item Given moduli $q_1,\dots,q_s$ in $\polRing$, whose sum of degrees is
    $e$, and $p$ of degree $d+e$, one can compute $p \bmod q_1,\dots,p \bmod
    q_s$ using $\bigO{\polmultime{d} + \polmultime{e} \log(e)}$ operations in
    $\field$.
  \item Conversely, Chinese remaindering modulo polynomials 
    with sum of degrees $d$ can be done in $\bigO{ \polmultime{d} \log(d)}$
    operations in $\field$.
  \end{itemize}
\end{lemma}
\begin{proof}
  For the first and third point, we refer the reader to \citep[Chapters~9
  and~10]{vzGathen13}. For the second point, we first compute $q=q_1 \cdots
  q_s$ in time $\bigO{\polmultime{e} \log(e)}$, reduce $p$ modulo $q$ in time
  $\bigO{\polmultime{d+e}}$, and use the simultaneous modular reduction
  algorithm of~\citepalias[Corollary~10.17]{vzGathen13}, which takes time
  $\bigO{\polmultime{e}\log(e)}$. Besides, we have $\polmultime{d+e} +
  \polmultime{e}\log(e) \in \bigO{ \polmultime{d} + \polmultime{e}\log(e) }$,
  as can be seen by considering the cases $d \le e$ and $d > e$.
\end{proof}


\subsection{Main algorithm}

For a Jordan matrix $\mulmat \in \matSpace[\order]$ given in standard
representation, and for any $x$ in $\field$, we will denote by
$\rep{x}{\mulmat}$ the number of pairs $(x,s)$ appearing in that
representation, counting repetitions (so that $\sum_{x \in \field}
\rep{x}{\mulmat} = \order$).

For an integer $k \in \{0,\ldots,\lceil\log(\order)\rceil\}$, we select from
the representation of $\mulmat$ all those pairs $(x,s)$ with $s$ in
$\{2^k,\dots,2^{k+1}-1\}$, obtaining a set $\mulmat^{(k)}$. Since $\mulmat$ is
in standard representation, we can compute all $\mulmat^{(k)}$ by a single pass
through the array $\mulmat$, and we can ensure for free that all
$\mulmat^{(k)}$ themselves are in standard representation. We decompose
$\mulmat^{(k)}$ further into two classes $\mulmat^{(k,> \rdim)}$, where all
pairs $(x,s)$ are such that $\rep{x}{\mulmat^{(k)}}$ is greater than $\rdim$,
and $\mulmat^{(k,\le \rdim)}$, which contains all other pairs.  As above, this
decomposition can be done in linear time, and we can ensure for no extra cost
that $\mulmat^{(k,> \rdim)}$ and $\mulmat^{(k,\le \rdim)}$ are in standard
representation.  Explicitly, these sequences will be written as
\begin{align*}
\mulmat^{(k,> \rdim)} &=((\evpt^{(k)}_{1},s^{(k)}_{1,1}),\dots,(\evpt^{(k)}_{1},s^{(k)}_{1,r^{(k)}_{1}}),\dots,(\evpt^{(k)}_{t^{(k)}},s^{(k)}_{t^{(k)},1}),\dots,(\evpt^{(k)}_{t^{(k)}},s^{(k)}_{t^{(k)},r^{(k)}_{t^{(k)}}})),
\end{align*}
with $(r^{(k)}_i)_i = (\rep{\evpt^{(k)}_{i}}{\mulmat^{(k)}})_i$ non-increasing,
and where for $i$ in $\{1,\dots,t^{(k)}\}$, $r^{(k)}_{i} >\rdim$ and
$(s^{(k)}_{i,j})_j$ is a non-increasing sequence of elements in
$\{2^k,\dots,2^{k+1}-1\}$. The corresponding sets of columns in the input
matrix $\evMat$ and the output $\mat{F}$ will be written $$\evMat^{(k, >
\rdim)} = (\evMat^{(k,>\rdim)}_{i,j})_{1 \le i \le t^{(k)}, 1 \le j \le
r^{(k)}_i}$$ and $$\mat{F}^{(k, > \rdim)} = (\mat{F}^{(k,>\rdim)}_{i,j})_{1 \le
  i \le t^{(k)}, 1 \le j \le r^{(k)}_i};$$ they will be treated using a direct
  application of \Cref{lem:evaluations}. Similarly, we write
\begin{align*}
\mulmat^{(k, \le \rdim)}&=((\xi^{(k)}_{1},\szbl^{(k)}_{1,1}),\dots,(\xi^{(k)}_{1},\szbl^{(k)}_{1,\rho^{(k)}_{1}}),\dots,(\xi^{(k)}_{\tau^{(k)}},\szbl^{(k)}_{\tau^{(k)},1}),\dots,(\xi^{(k)}_{\tau^{(k)}},\szbl^{(k)}_{\tau^{(k)},\rho^{(k)}_{\tau^{(k)}}})),
\end{align*}
with $(\rho^{(k)}_i)_i = (\rep{\xi^{(k)}_{i}}{\mulmat^{(k)}})_i$
non-increasing, and where for $i$ in $\{1,\dots,\tau^{(k)}\}$, $\rho^{(k)}_{i}
\le \rdim$ and $(\szbl^{(k)}_{i,j})_j$ is a non-increasing sequence of elements
in $\{2^k,\dots,2^{k+1}-1\}$. The corresponding sets of columns in the input
matrix $\evMat$ and the output $\mat{F}$ will be written $\evMat^{(k,\le
\rdim)}$ and $\mat{F}^{(k,\le \rdim)}$; more precisely, they take the form
$$\evMat^{(k, \le \rdim)} = (\evMat^{(k,\le \rdim)}_{i,j})_{1 \le i \le
  \tau^{(k)}, 1 \le j \le \rho^{(k)}_i}$$ and $$\mat{F}^{(k, \le \rdim)} =
  (\mat{F}^{(k, \le \rdim)}_{i,j})_{1 \le i \le \tau^{(k)}, 1 \le j \le
  \rho^{(k)}_i},$$ and will be treated using a Chinese remaindering approach.

In the main loop, the index $k$ will range from $0$ to $\lfloor
\log(\order/\rdim) \rfloor$. After that stage, all entries $(x,s)$ in
$\mulmat$ that were not processed yet are such that $s > \order/\rdim$. In
particular, if we call $\mulmat^{(\infty,\le \rdim)}$ the set of these
remaining entries, we deduce that this set has cardinality at most $\rdim$;
thus $\rep{x}{\mulmat^{(\infty,\le \rdim)}} \le \rdim$ holds for all $x$ and we
process these entries using the Chinese remaindering approach.

Algorithm \algoname{ComputeResiduals} constructs all these sets $\mulmat^{(k,>
\rdim)}$, $\mulmat^{(k,\le \rdim)}$, and $\mulmat^{(\infty,\le \rdim)}$, then
extracts the corresponding columns from~$\evMat$ (this is the subroutine
\algoname{ExtractColumns}), and processes these subsets of columns, before
merging all the results.

\begin{figure}[h!]
  \centering
  \fbox{\begin{minipage}{12.5cm}
  \begin{algorithm}[\algoname{ComputeResiduals}]
  \label{algo:update-evaluations}
  ~\smallskip\\
  \algoword{Input:}
  \begin{itemize}
  \item a Jordan matrix $\mulmat$ in $\matSpace[\order]$ in standard representation,
  \item a matrix $\mat{P}\in\intSpace$,
  \item a matrix $\evMat \in \matSpace[\rdim][\order]$.
  \end{itemize}

  \smallskip
	\algoword{Output:} 
	the product $\mat{P} \mul \evMat \in \matSpace[\rdim][\order]$.

  \smallskip
	\begin{enumerate}[{\bf 1.}]
    \item \algoword{For} $k$ \algoword{from} $0$ \algoword{to} $\lfloor \log(\order/\rdim) \rfloor$
    \begin{enumerate}[{\bf a.}]
	  \item ${\mulmat}^{(k)} \leftarrow ((\evpt,s)\in\mulmat \mid  2^k \le s <  2^{k+1})$
	  \item ${\mulmat}^{(k,> \rdim)} \leftarrow ((\evpt,s)\in\mulmat^{(k)} \mid  \rep{\evpt}{\mulmat^{(k)}} > \rdim)$
          \item $\evMat^{(k,> \rdim)} \leftarrow \algoname{ExtractColumns}(\evMat, {\mulmat}^{(k,> \rdim)})$
          \item $\mat{F}^{(k,> \rdim)}\leftarrow \algoname{ComputeResidualsByShiftingP}({\mulmat}^{(k,> \rdim)},\mat{P},\evMat^{(k,> \rdim)})$
	  \item ${\mulmat}^{(k,\le \rdim)}  \leftarrow ((\evpt,s)\in\mulmat^{(k)} \mid \rep{\evpt}{\mulmat^{(k)}} \le \rdim )$
          \item $\evMat^{(k, \le \rdim)} \leftarrow \algoname{ExtractColumns}(\evMat, {\mulmat}^{(k,\le \rdim)})$
          \item $\mat{F}^{(k,\le \rdim)}\leftarrow \algoname{ComputeResidualsByCRT}({\mulmat}^{(k,\le \rdim)},\mat{P},\evMat^{(k,\le \rdim)})$
          \end{enumerate}
	\item ${\mulmat}^{(\infty,\le \rdim)} \leftarrow ((\evpt,s)\in\mulmat \mid 2^{\lfloor \log(\sigma/\rdim) \rfloor+1} \le s)$
        \item $\evMat^{(\infty,\le \rdim)} \leftarrow \algoname{ExtractColumns}(\evMat, {\mulmat}^{(\infty,\le \rdim)})$
        \item $\mat{F}^{(\infty,\le \rdim)}\leftarrow \algoname{ComputeResidualsByCRT}({\mulmat}^{(\infty,\le \rdim)},\mat{P},\evMat^{(\infty,\le \rdim)})$
  \item \algoword{Return} $\algoname{Merge}(
         ( \mat{F}^{(k,> \rdim)})_{0 \le k \le \lfloor \log(\order/\rdim) \rfloor},
         ( \mat{F}^{(k,\le \rdim)})_{ 0 \le k \le \lfloor \log(\order/\rdim) \rfloor},
          \mat{F}^{(\infty,\le \rdim)}  )$
	\end{enumerate}
      \end{algorithm}
  \end{minipage}}
\end{figure}

\subsubsection{Computing the residual by shifting $\mathbf{P}$}
\label{subsubsec:residuals-byshifting}

We start with the case of the sets $\mulmat^{(k, > \rdim)}$, for which we
follow a direct approach. Below, recall that we write
$$\mulmat^{(k,> \rdim)}=(
(\evpt^{(k)}_{1},s^{(k)}_{1,1}),\dots,(\evpt^{(k)}_{1},s^{(k)}_{1,r^{(k)}_{1}}),
\dots, (\evpt^{(k)}_{t^{(k)}},s^{(k)}_{t^{(k)},1}), \dots,
(\evpt^{(k)}_{t^{(k)}},s^{(k)}_{t^{(k)},r^{(k)}_{t^{(k)}}})),$$
with
$s^{(k)}_{i,1} \ge s^{(k)}_{i,j}$ for any $k$, $i$, and $j$. For a fixed $k$,
we compute $\mat{P}_i^{(k)}=\mat{P}(X+\evpt^{(k)}_i) \bmod X^{s^{(k)}_{i,1}}$,
for $i$ in $\{1,\dots,t^{(k)}\}$, and do the corresponding matrix products. This is
described in \Cref{algo:update-evaluations-direct}; we give below a
bound on the total time spent in this algorithm, that is, for all $k$ in
$\{0,\dots,\lfloor \log(\order/\rdim) \rfloor\}$. Before that, we give two
lemmas: the first one will allow us to control the cost of the calculations in
this case; in the second one, we explain how to efficiently compute the
polynomial matrices $\mat{P}_i^{(k)}$.

\begin{lemma}\label{lemma:sumdeg}
  The following bound holds:
  \[ \sum_{k=0}^{\lfloor\log(\order/\rdim)\rfloor}
    \sum_{i=1}^{t^{(k)}} r_i^{(k)} s^{(k)}_{i,1} \;\in \bigO{\order}. \]
\end{lemma}
\begin{proof}
By construction, we have the estimate \[
  \sum_{k=0}^{\lfloor\log(\order/\rdim)\rfloor} \sum_{i=1}^{t^{(k)}}
  \sum_{j=1}^{r^{(k)}_i}  s^{(k)}_{i,j} \;\le \sigma,
\] since this represents the total size of all blocks contained in the
sequences $\mulmat^{(k, > \rdim)}$. Now, for fixed $k$ and $i$, the
construction of $\mulmat^{(k)}$ implies that the inequality $s^{(k)}_{i,1} \le
2 s^{(k)}_{i,j}$ holds for all $j$. This shows that we have $$r_i^{(k)}
s^{(k)}_{i,1} \le 2 \sum_{j=1}^{r^{(k)}_i} s^{(k)}_{i,j},$$ and the
conclusion follows by summing over all $k$ and $i$.
\end{proof}

In the following lemma, we explain how to compute the polynomial matrices
$\mat{P}_i^{(k)}$ in an efficient manner, for $i$ in $\{1,\dots,t^{(k)}\}$ and
for all the values of $k$ we need.
\begin{lemma}\label{lemma:allsimulmod}
  Suppose that the sum of the row degrees of $\mat{P}$ is $\bigO{\order}$.
  Then one can compute the matrices $\mat{P}_i^{(k)}$ for all $k$ in
  $\{0,\dots,\lfloor \log(\order/\rdim)\rfloor\}$ and $i$ in
  $\{1,\dots,t^{(k)}\}$ using $\bigO{\rdim \, \polmultime{\order}
  \log(\order)}$ operations in $\field$.
\end{lemma}
\begin{proof}
  We use the second item in \Cref{lemma:multirem} to first compute
  $\mat{P} \bmod (X-\evpt^{(k)}_i)^{s^{(k)}_{i,1}}$, for all $k$ and $i$ as in
  the statement of the lemma. Here, the sum of the degrees is $$S =\sum_{k,i}
  s^{(k)}_{i,1},$$ so we get a total cost of $\bigO{ \polmultime{d} +
  \polmultime{S}\log(S) }$ for an entry of $\mat{P}$ of degree $d$. Summing
  over all entries, and using the fact that the sum of the row degrees of
  $\mat{P}$ is $\bigO{\order}$, we obtain a total cost of $$\bigO{
    \rdim\,\polmultime{\order} + \rdim^2\polmultime{S }\log(S) }.$$ Now,
    because we consider here $\mulmat^{(k,> \rdim)}$, we have $r^{(k)}_i >
    \rdim$ for all $k$ and $i$. Hence, using the super-linearity of
    $\polmultime{\cdot}$, the term $\rdim^2 \polmultime{S} \log(S)$ admits the
    upper bound $$ \rdim\, \polmultimePar {\left( \sum_{k,i} r^{(k)}_i
    s^{(k)}_{i,1} \right) } \log(S),$$ which is in $\bigO{\rdim\,
      \polmultime{\order}\log(\order)}$ in view of \Cref{lemma:sumdeg}.

  Then we apply a variable shift to all these polynomials to replace $X$ by
  $X+\evpt^{(k)}_i$. Using the first item in \Cref{lemma:multirem}, for
  fixed $k$ and $i$, the cost is $\bigO{\rdim^2 \polmultime{s^{(k)}_{i,1}}
  \log(s^{(k)}_{i,1})}$.  Hence, the total time is again $\bigO{\rdim^2
  \polmultime{S}\log(S)}$, so the same overall bound as above holds.
\end{proof}

\begin{figure}[h!]
  \centering
  \fbox{\begin{minipage}{12.5cm}
  \begin{algorithm}[\algoname{ComputingResidualsByShiftingP}]
	\label{algo:update-evaluations-direct}
  ~\smallskip\\
	\algoword{Input:}
    \begin{itemize}
    \item $\mulmat^{(k,> \rdim)} =((\evpt^{(k)}_{1},s^{(k)}_{1,1}),\dots,(\evpt^{(k)}_{1},s^{(k)}_{1,r^{(k)}_{1}}),\dots,(\evpt^{(k)}_{t^{(k)}},s^{(k)}_{t^{(k)},1}),\dots,(\evpt^{(k)}_{t^{(k)}},s^{(k)}_{t^{(k)},r^{(k)}_{t^{(k)}}}))$
in standard representation,
    \item a matrix $\mat{P}\in\intSpace$,
    \item a matrix $\evMat^{(k,>\rdim)}=[\evMat^{(k,>\rdim)}_{1,1} | \cdots | \evMat^{(k,>\rdim)}_{t^{(k)},r^{(k)}_{t^{(k)}}}] \in \matSpace[\rdim][\sum_{i,j}s^{(k)}_{i,j}]$
      with $\evMat^{(k,>\rdim)}_{i,j} \in \matSpace[\rdim][s^{(k)}_{i,j}]$ for all $i,j$.
    \end{itemize}

  \smallskip
	\algoword{Output:} 
	the product $\mat{P} \mul \evMat^{(k,>\rdim)} \in \matSpace[\rdim][\sum_{i,j}s^{(k)}_{i,j}]$.

  \smallskip
	\begin{enumerate}[{\bf 1.}]
	\item $(\mat{P}^{(k)}_i)_{1 \le i \le t^{(k)}} \leftarrow (\mat{P}(X+\evpt_i) \bmod X^{s^{(k)}_{i,1}})_{1 \le i \le t^{(k)}}$
  \item \algoword{For} $i$ \algoword{from} $1$ \algoword{to} $t^{(k)}$
	  \begin{enumerate}[{\bf a.}]
	  \item $(\evMat^{(k,>\rdim)}_{i,j,{\rm poly}})_{1 \le j \le r^{(k)}_i} \leftarrow (\evMat^{(k,>\rdim)}_{i,j} \trsp{[1,\var,\ldots,\var^{s^{(k)}_{i,j}-1}]})_{1 \le j \le r^{(k)}_i}$
	  \item $ [\evMatF^{(k,>\rdim)}_{i,1,{\rm poly}} | \cdots | \evMatF^{(k,>\rdim)}_{i,r^{(k)}_i,{\rm poly}}] \leftarrow {\mat{P}}^{(k)}_i [\evMat^{(k,>\rdim)}_{i,1,{\rm poly}} | \cdots | \evMat^{(k,>\rdim)}_{i,r^{(k)}_i,{\rm poly}}] $
    \item \algoword{For} $j$ \algoword{from} $1$ \algoword{to} $r^{(k)}_i$, $\evMatF^{(k,>\rdim)}_{i,j} \leftarrow ( {\rm coeff}(\evMatF^{(k,>\rdim)}_{i,j,{\rm poly}}, \ell ))_{ 0 \le \ell < s^{(k)}_{i,j}} $ 
	  \end{enumerate}
  \item \algoword{Return} $[\evMatF^{(k,>\rdim)}_{1,1} | \cdots | \evMatF^{(k,>\rdim)}_{t^{(k)},r^{(k)}_{t^{(k)}}}]$
	\end{enumerate}
  \end{algorithm}
  \end{minipage}}
\end{figure}

\begin{lemma}
  \label{lem:update-evaluations-direct}
  \Cref{algo:update-evaluations-direct} is correct. Given the
  polynomial matrices computed in \Cref{lemma:allsimulmod}, the total time
  spent in this algorithm for all $k$ in
  $\{0,\dots,\lfloor\log(\order/\rdim)\rfloor\}$ is $\bigO{
    \polmatmultime{\rdim,\order/\rdim} }$ operations in $\field$.
\end{lemma}
\begin{proof}
  Correctness of the algorithm follows from \Cref{lem:evaluations}, so we
  focus on the cost analysis. 
  
  \Cref{lemma:allsimulmod} gives the cost of computing all polynomial matrices
  needed at Step~\textbf{1}. The only other arithmetic operations are those
  done in the matrix products at Step~\textbf{2.b}: we multiply matrices of
  respective sizes $\rdim\times \rdim$ and $\rdim \times r^{(k)}_i$, with
  entries of degree less than $s^{(k)}_{i,1}$. For given $k$ and $i$, since we
  have $\rdim < r^{(k)}_i$, the cost is $\bigO{ \polmatmultime{\rdim,
    s^{(k)}_{i,1} } r^{(k)}_i / \rdim }$; using the super-linearity of $d
    \mapsto \polmatmultime{\rdim,d}$, this is in $\bigO{ \polmatmultime{\rdim,
      r^{(k)}_i s^{(k)}_{i,1} / \rdim } }$. Applying again \Cref{lemma:sumdeg},
      we deduce that the sum over all $k$ and $i$ is $\bigO{ \polmatmultime{
      \rdim, \order/\rdim } }$.
\end{proof}

\subsubsection{Computing the residual by Chinese remaindering}
\label{subsubsec:residual-bycrt}

The second case to consider is $\mulmat^{(k, \le \rdim)}$. Recall that for a
given index $k$, we write this sequence as $$\mulmat^{(k, \le \rdim)} =
((\xi^{(k)}_{1},\sigma^{(k)}_{1,1}),\dots,(\xi^{(k)}_{1},\sigma^{(k)}_{1,\rho^{(k)}_{1}}),\dots,(\xi^{(k)}_{\tau^{(k)}},\sigma^{(k)}_{\tau^{(k)},1}),\dots,(\xi^{(k)}_{\tau^{(k)}},\sigma^{(k)}_{\tau^{(k)},\rho^{(k)}_{\tau^{(k)}}})),$$
with $\rho^{(k)}_{\tau^{(k)}} \le \cdots \le \rho^{(k)}_{1} \le \rdim$
for all $i$ in $\{1,\dots,\tau^{(k)}\}$. In this case, $\tau^{(k)}$ may be
large so the previous approach may lead us to compute too many matrices
$\mat{P}^{(k)}_i$. Instead, for fixed $k$ and~$j$, we use Chinese remaindering
to transform the corresponding submatrices $\evMat^{(k,\le \rdim)}_{i,j}$ into
a polynomial matrix $\evMat^{(k,\le \rdim)}_{j}$ of small column dimension;
this allows us to efficiently perform matrix multiplication by $\mat{P}$ on the
left, and we eventually get $\mat{P}\mul\evMat^{(k,\le \rdim)}_{i,j}$ by
computing the first coefficients in a Taylor expansion of this product around
every $\xi^{(k)}_i$. 

To simplify the notation in the algorithm, we also suppose that for a fixed
$k$, the points $\xi^{(k)}_{1},\dots,\xi^{(k)}_{\tau^{(k)}}$ all appear the
same number of times in $\mulmat^{(k, \le \rdim)}$. This is done by replacing
$\rho^{(k)}_1,\dots,\rho^{(k)}_{\tau^{(k)}}$ by their maximum $\rho^{(k)}_1$
(simply written $\rho^{(k)}$ in the pseudo-code) and adding suitable blocks
$(\xi^{(k)}_i,\sigma^{(k)}_{i,j})$, with all new $\sigma^{(k)}_{i,j}$ set to
zero.

\begin{figure}[h!]
  \centering
  \fbox{\begin{minipage}{12.5cm}
    \begin{algorithm}[\algoname{ComputingResidualsByCRT}]
	\label{algo:update-evaluations-interpolation}
  ~\smallskip\\
	\algoword{Input:}
    \begin{itemize}
     \item $\mulmat^{(k, \le \rdim)}=((\xi^{(k)}_{1},\sigma^{(k)}_{1,1}),\dots,(\xi^{(k)}_{1},\sigma^{(k)}_{1,\rho^{(k)}}),\dots,(\xi^{(k)}_{\tau^{(k)}},\sigma^{(k)}_{\tau^{(k)},1}),\dots,(\xi^{(k)}_{\tau^{(k)}},\sigma^{(k)}_{\tau^{(k)},\rho^{(k)}}))$
in standard representation,
    \item a matrix $\mat{P}\in\intSpace$,
    \item a matrix $\evMat^{(k, \le \rdim)} = [\evMat^{(k, \le \rdim)}_{1,1} | \cdots | \evMat^{(k, \le \rdim)}_{\tau^{(k)},\rho^{(k)}}] \in \matSpace[\rdim][\sum_{i,j}\sigma^{(k)}_{i,j}]$
      with $\evMat^{(k,\le \rdim)}_{i,j} \in \matSpace[\rdim][\sigma^{(k)}_{i,j}]$ for all $i,j$.
    \end{itemize}

  \smallskip
	\algoword{Output:} 
	the product $\mat{P} \mul \evMat^{(k,\le \rdim)} \in \matSpace[\rdim][\sum_{i,j}\sigma^{(k)}_{i,j}]$.
	
  \smallskip
	\begin{enumerate}[{\bf 1.}]
    \item \algoword{For} $j$ \algoword{from} $1$ \algoword{to} $\rho^{(k)}$
	  \begin{enumerate}[{\bf a.}]
          \item  $({\evMat}^{(k,\le \rdim)}_{i,j,{\rm shifted}})_{1 \le i \le \tau^{(k)}} \leftarrow (\evMat^{(k,\le \rdim)}_{i,j} \trsp{[1,\var-\xi^{(k)}_i,\ldots,(\var-\xi^{(k)}_i)^{\sigma^{(k)}_{i,j}-1}]})_{1 \le i \le \tau^{(k)}}$
          \item $\evMat^{(k,\le \rdim)}_{j,{\rm shifted}} \leftarrow \algoname{CRT}((\evMat^{(k,\le \rdim)}_{i,j,{\rm shifted}})_{1 \le i \le \tau^{(k)}},((X-\xi^{(k)}_i)^{\sigma^{(k)}_{i,j}})_{1 \le i \le \tau^{(k)}})$
          \end{enumerate}
        \item $ [\evMatF^{(k,\le \rdim)}_{1,{\rm shifted}} | \cdots | \evMatF^{(k,\le \rdim)}_{\rho^{(k)},{\rm shifted}}] \leftarrow {\mat{P}} [\evMat^{(k,\le \rdim)}_{1,{\rm shifted}} | \cdots | {\evMat}^{(k,\le \rdim)}_{\rho^{(k)},{\rm shifted}}] $
        \item \algoword{For} $j$ \algoword{from} $1$ \algoword{to} $\rho^{(k)}$
	  \begin{enumerate}[{\bf a.}]
          \item $(\evMatF^{(k, \le \rdim)}_{i,j,{\rm shifted}})_{1 \le i \le \tau^{(k)}} \leftarrow (\evMatF^{(k, \le \rdim)}_{j,{\rm shifted}} \bmod (X-\xi^{(k)}_i)^{\sigma^{(k)}_{i,j}})_{1 \le i \le \tau^{(k)}}$
          \item ${\evMatF}^{(k,\le \rdim)}_{i,j} \leftarrow ( {\rm coeff}(\evMatF^{(k, \le \rdim)}_{i,j,{\rm shifted}}(X+\xi^{(k,\le m)}_i), \ell) )_{0 \le \ell < \sigma^{(k)}_{i,j}} $
          \end{enumerate}
    \item \algoword{Return} $[\evMatF^{(k,\le \rdim)}_{1,1} | \cdots | \evMatF^{(k,\le \rdim)}_{\tau^{(k)},\rho^{(k)}}]$
	\end{enumerate}
      \end{algorithm}
  \end{minipage}}
\end{figure}

\begin{lemma}
  \label{lem:update-evaluations-interpolation}
  \Cref{algo:update-evaluations-interpolation} is correct. If the sum
  of the row degrees of $\mat{P}$ is in $\bigO{\order}$, the total time spent
  in this algorithm for all $k$ in
  $\{0,\dots,\lfloor\log(\order/\rdim)\rfloor,\infty\}$ is \[ \bigO{
  \polmatmultime{ \rdim, \order/\rdim } \log(\order/\rdim) + \rdim
\polmultime{\order}\log(\order) }\] operations in $\field$.
\end{lemma}
\begin{proof}
  Proving correctness amounts to verifying that we compute the
  quantities described in \Cref{lem:evaluations}. Indeed, the
  formulas in the algorithm show that for all $k,i,j$, we have
  $\evMatF^{(k, \le \rdim)}_{i,j,{\rm shifted}}= \mat{P}\, \evMat^{(k, \le \rdim)}_{i,j,{\rm
      shifted}} \bmod (X-\xi^{(k)}_i)^{\sigma^{(k)}_{i,j}}$; the link with
  \Cref{lem:evaluations} is made by observing that
  $\evMat^{(k, \le \rdim)}_{i,j,{\rm shifted}}=\evMat^{(k,\le \rdim)}_{i,j,{\rm poly}}(X-\xi^{(k)}_i)$ and
  $\evMatF^{(k, \le \rdim)}_{i,j,{\rm shifted}}=\evMatF^{(k,\le \rdim)}_{i,j,{\rm poly}}(X-\xi^{(k)}_i)$.

  In terms of complexity, the first item in \Cref{lemma:multirem} shows
  that for a given index $k$, Step~{\bf 1.a} can be done in time $$\bigOPar{\rdim
  \sum_{i,j} \polmultimePar{\left(\sigma^{(k)}_{i,j}\right)}
  \log\left(\sigma^{(k)}_{i,j}\right)},$$ for a total cost of $\bigO{\rdim\,
  \polmultime \order \log(\order)}$. Step~{\bf 1.b} can be done in
  quasi-linear time as well: for each $k$ and $j$, we can compute each of the
  $\rdim$ entries of the polynomial vector $\evMat^{(k,\le \rdim)}_{j,{\rm
  shifted}}$ by fast Chinese remaindering (third item in
  \Cref{lemma:multirem}), using $$\bigOPar{\polmultimePar{\left (
  S^{(k)}_j \right)}\log\left ( S^{(k)}_j \right)}$$ operations in $\field$,
  with $S^{(k)}_j = \sum_i \sigma^{(k)}_{i,j}$. Taking all rows into account,
  and summing over all indices $k$ and $j$, we obtain again a total cost of
  $\bigO{\rdim\, \polmultime \order \log(\order)}$.

  The next step to examine is the polynomial matrix product at Step~{\bf 2}.
  The matrix $\mat{P}$ has size $\rdim \times \rdim$, and the sum of its row
  degrees is by assumption $\bigO{\order}$; using the partial linearization
  technique presented in \Cref{sec:cost-mult}, we can replace $\mat{P}$
  by a matrix of size $\bigO{\rdim} \times \rdim$ with entries of degree at
  most $\order/\rdim$.

  For a fixed choice of $k$, the right-hand side has size $\rdim
  \times \rho^{(k)}$, and its columns have respective degrees less
  than $S^{(k)}_1,\dots,S^{(k)}_{\rho^{(k)}}$.  We split each of its
  columns into new columns of degree at most $\order/\rdim$, so that the $j$th
  column is split into $\bigO{1 + S^{(k)}_j \rdim/\order}$ columns
  (the constant term 1 dominates when $S^{(k)}_j \le \order/\rdim$).
  Thus, the new right-hand side
  has $\bigO{\rho^{(k)} + (S^{(k)}_1 +\cdots +S^{(k)}_{\rho^{(k)}})
  \rdim/\order}$ columns and degree at most~$\order/\rdim$.
  
  Now, taking all $k$ into account, we remark that the left-hand side remains
  the same; thus, we are led to do one matrix product with degrees
  $\order/\rdim$, with left-hand side of size $\bigO{\rdim}\times \rdim$, and
  right-hand side having column dimension at most $$\sum_{k \in \{ 0,\dots,\lfloor
  \log(\order/\rdim)\rfloor\} \cup \{\infty \} } \rho^{(k)} + \frac{(S^{(k)}_1 +\cdots
  +S^{(k)}_{\rho^{(k)}}) \rdim}\order.$$ Since all $\rho^{(k)}$ are at most
  $\rdim$, the first term sums up to $\bigO{\rdim \log(\order/\rdim)}$; by
  construction, the second one adds up to $\bigO{\rdim}$. Hence, the matrix
  product we need can be done in time $\bigO{ \polmatmultime{ \rdim,
  \order/\rdim } \log(\order/\rdim)}$.

  For a given $k$, $\evMatF^{(k,\le \rdim)}_{1,{\rm shifted}}, \dots,
  \evMatF^{(k,\le \rdim)}_{\rho^{(k)},{\rm shifted}}$ are vectors of size
  $\rdim$. Furthermore, for each $j$ the entries of $\evMatF^{(k,\le
  \rdim)}_{j,{\rm shifted}}$ have degree less than $S^{(k)}_1 + d_1,
  \ldots,S^{(k)}_\rdim + d_\rdim$ respectively, where $d_1,\ldots,d_\rdim$ are
  the degrees of the rows of $\mat{P}$. In particular, for a fixed $k$, the
  reductions at Step~{\bf 3.a} can be done in time $$\bigOPar{ \rho^{(k)}
  (\polmultime{d_1 + \cdots + d_\rdim}) +\rdim \sum_{j=1}^{\rho^{(k)}}
  \polmultime{S^{(k)}_j} \log(S^{(k)}_j) }$$ using fast multiple reduction, by
  means of the second item in \Cref{lemma:multirem}.  Using our assumption on
  $\mat{P}$, and the fact that $\rho^{(k)} \le m$, we see that the first term
  is $\bigO{m \polmultime \sigma}$, which adds up to $\bigO{\rdim \polmultime
  \sigma \log(\sigma/\rdim)}$ if we sum over $k$. The second term adds up to
  $\bigO{\rdim \polmultime \sigma \log(\sigma)}$, as was the case for Step~{\bf
  1.b}.

  The same analysis is used for the shifts taking place at Step~{\bf 3.b} as
  for those in Step~\textbf{1.a}: for fixed $k$ and $j$, the cost is
  $\bigO{\rdim \polmultime{S^{(k)}_j}\log(S^{(k)}_j)}$, and we conclude as
  above.
\end{proof}

%
%
\section{Minimal interpolation basis via linearization}
\label{sec:lin-mib} 

In this section, we give an efficient algorithm based on linearization
techniques to compute interpolation bases for the case of an arbitrary matrix
$\mulmat$ and an arbitrary shift; in particular, we prove
\Cref{thm:mib-linear}.

In addition to being interesting on its own, the algorithm in this section
allows us to handle the base cases in the recursion of the divide-and-conquer
algorithm presented in \Cref{sec:algo}. For that particular case, we
have $\rdim/2 \le \order \le \rdim$; the algorithm we give here solves this
base case using $\bigO{\rdim^\expmatmul \log(\rdim)}$ operations in $\field$.

\begin{proposition}
\label{prop:lin-mib}
Let $\mulmat \in \matSpace[\order]$, $\evMat \in \matSpace[\rdim][\order]$,
$\shifts \in \shiftSpace$, and let $\maxDeg \in \NN$ be a bound on the
degree of the minimal polynomial of $\mulmat$. Then,
\Cref{algo:lin-mib} solves \Cref{pbm:mib} deterministically,
using 
\begin{align*}
  \bigO{ & \order^\expmatmul ( \lceil \rdim / \order \rceil + \log(\maxDeg)) } & \text{if } \expmatmul >2 \\
  \bigO{ & \order^2 ( \lceil \rdim / \order \rceil + \log(\maxDeg)) \log(\order) } & \text{if } \expmatmul =2
\end{align*}
operations in $\field$; it returns the unique interpolation basis
$\intBasis$ for $(\evMat,\mulmat)$ which is in $\shifts$-Popov form. Besides,
the maximal degree in $\intBasis$ is at most $\maxDeg$ and the sum of the
\emph{column} degrees of $\intBasis$ is at most $\order$.
\end{proposition}

We remark that the degree of the minimal polynomial of $\mulmat$ is at most
$\order$. In \Cref{algo:lin-mib}, we require that $\maxDeg$ be a power of $2$
and thus we may have $\maxDeg > \order$; still we can always choose $\maxDeg <
2\order$. The proof is deferred until \Cref{subsec:lin-mib-algo}, where we
also recall the definition of the shifted Popov form~\citep{BeLaVi06}.

To obtain this result, we rely on linear algebra tools via the use of the
linearization in~\citep{BecLab00}, where an interpolant is seen as a linear
relation between the rows of a striped Krylov matrix. The reader may also refer
to~\citep[\S6.3 and \S6.4]{Kailath80} for a presentation of this point of view.
In \citep{BecLab00}, it is assumed that $\jordan$ is upper triangular: this
yields recurrence relations \citepalias[Theorem 6.1]{BecLab00}, leading to an
iterative algorithm \citepalias[Algorithm FFFG]{BecLab00} to compute an
interpolation basis in shifted Popov form in a fraction-free way.

Here, to obtain efficiency and deal with a general $\jordan$, we proceed in two
steps. First, we compute the row rank profile of the striped Krylov matrix
$\mathcal{K}$ with an algorithm \emph{\`a la}~\citet{KelGeh85}, which uses at
most $\log(\minDeg)$ steps and supports different orderings of the rows in
$\mathcal{K}$ depending on the input shift.  Then, we use the resulting
independent rows of $\mathcal{K}$ to compute the specific rows in the nullspace
of $\mathcal{K}$ which correspond to the interpolation basis in shifted Popov
form.

We note that when $\order=\bigO{1}$, the cost bound in \Cref{prop:lin-mib} is
linear in $\rdim$, while the dense representation of the output $\rdim \times
\rdim$ polynomial matrix will use at least $\rdim^2$ field elements. We will
see in \Cref{subsec:lin-mib-algo} that when $\order < \rdim$, at least
$\rdim-\order$ columns of the basis in $\shifts$-Popov form have only one
nonzero coefficient which is $1$, and thus those columns can be described
without involving any arithmetic operation. Hence, the actual computation is
restricted to an $\rdim \times \order$ submatrix of the output basis.

\subsection{Linearization}
\label[subsec]{subsec:linearization}

Our goal is to explain how to transform the problem of finding interpolants
into a problem of linear algebra over $\field$. This will involve a
straightforward linearization of the polynomials in the output interpolation
basis $\intBasis$, expanding them as a list of coefficients so that $\intBasis$
is represented as a matrix over $\field$. Correspondingly, we show how from the
input $(\evMat,\mulmat)$ one can build a matrix $\mathcal{K}$ over $\field$
which is such that an interpolant for $(\evMat,\mulmat)$ corresponds to a
vector in the left nullspace of $\mathcal{K}$. Then, since we will be looking
for interpolants that have a small degree with respect to the column shifts
given by $\shifts$, we describe a way to adapt these constructions so that they
facilitate taking into account the influence of $\shifts$. This gives a first
intuition of some properties of the linearization of an interpolant that has
small shifted degree: this will then be presented in details in
\Cref{subsec:rel-int}.

Let us first describe the linearization of interpolants, which are seen as row
vectors in $\polMatSpace[1][\rdim]$. In what follows, we suppose that we know a
bound $\maxDeg \in \NNp$ on the degree of the minimal polynomial of $\mulmat$;
one can always choose $\maxDeg = \order$. In \Cref{subsec:rel-int},
we will exhibit $\shifts$-minimal interpolation bases for $(\evMat,\mulmat)$
whose entries all have degree at most $\maxDeg$ (while in general such a basis
may have degree up to $\maxDeg + \sshifts[\shifts-\min(\shifts)]$). Thus, in
this \Cref{sec:lin-mib}, we focus on solutions to \Cref{pbm:mib}
that have degree at most $\maxDeg$. Correspondingly, $\polRing_{\le\maxDeg}$
denotes the set of polynomials in $\polRing$ of degree at most $\maxDeg$.

Given $\mat{P} \in \polMatSpace[\cdim][\rdim]_{\le\maxDeg}$ for some $\cdim \ge
1$, we write it as a polynomial of matrices: $\mat{P} = \mat{P}_0 +
\mat{P}_1\var + \cdots + \mat{P}_{\maxDeg} \var^{\maxDeg}$ where each
$\mat{P}_j$ is a scalar matrix in $\matSpace[\cdim][\rdim]$; then the
\emph{expansion} of $\mat{P}$ (in degree $\maxDeg$) is the matrix
$\linPolMat{\mat{P}} = \left[ \mat{P}_0 \mid \mat{P}_1 \mid \cdots \mid
\mat{P}_\maxDeg \right] \in\matSpace[\cdim][\rdim(\maxDeg+1)]$.  The reciprocal
operation is called \emph{compression} (in degree $\maxDeg$): given a scalar
matrix $\mat{M}\in\matSpace[\cdim][\rdim(\maxDeg+1)]$, we write it with blocks
$\mat{M} = \left[\mat{M}_0 \mid \mat{M}_1 \mid \cdots \mid
\mat{M}_\maxDeg\right]$ where each $\mat{M}_j$ is in $\matSpace[\cdim][\rdim]$,
and then we define its compression as $\polFromLin{\mat{M}} = \mat{M}_0 +
\mat{M}_1\var + \cdots + \mat{M}_{\maxDeg}\var^{\maxDeg} \in
\polMatSpace[\cdim][\rdim]_{\le\maxDeg}$. These definitions of
$\linPolMat{\mat{P}}$ and $\polFromLin{\mat{M}}$ hold for any row dimension
$n$; this $n$ will always be clear from the context.

Now, given some matrices $\evMat \in \evSpace{\order}$ and $\mulmat \in
\matSpace[\order]$, our interpolation problem asks to find $\intBasis \in
\intSpace_{\le\maxDeg}$ such that $\intBasis \mul \evMat = 0$.  Writing
$\intBasis = \intBasis_0 + \intBasis_1\var + \cdots + \intBasis_{\maxDeg}
\var^{\maxDeg}$, we recall that $\intBasis \mul \evMat = \intBasis_0 \evMat +
\intBasis_1 \evMat \mulmat + \cdots + \intBasis_\maxDeg \evMat
\mulmat^\maxDeg$.  Then, in accordance to the linearization of $\intBasis$, the
input $(\evMat,\mulmat)$ is expanded as follows: \[\krylov{\evMat} =
\left[\begin{array}{c} \evMat \\ \hline \evMat\mulmat \\ \hline \vdots \\
  \hline \evMat\mulmat^{\maxDeg} \end{array}\right] \in
\matSpace[\rdim(\maxDeg+1)][\order].\] This way, we have $\mat{P}\mul\evMat =
\linPolMat{\mat{P}}\krylov{\evMat}$ for any polynomial matrix $\mat{P} \in
\polMatSpace[\cdim][\rdim]_{\le\maxDeg}$.  In particular, a row vector
$\row{p} \in \polMatSpace[1][\rdim]_{\le\maxDeg}$ is an interpolant for
$\evMat$ if and only if $\linPolMat{\row{p}} \krylov{\evMat} = 0$, that is,
$\linPolMat{\row{p}}$ is in the (left) nullspace of $\krylov{\evMat}$. Up to
some permutation of the rows and different degree constraints, this so-called
\emph{striped-Krylov} matrix $\krylov{\evMat}$ was used in~\citep{BecLab00} for
the purpose of computing interpolants.

\medskip \noindent \textbf{Notation.}
For the rest of this \Cref{sec:lin-mib}, we will use the letter $i$ for
rows of $\krylov{\evMat}$ and for columns of $\linPolMat{\mat{P}}$; the letter
$j$ for columns of $\krylov{\evMat}$; the letter $d$ for the block of columns
of $\linPolMat{\mat{P}}$ which correspond to coefficients of degree $d$ in
$\mat{P}$, as well as for the corresponding block $\evMat\mulmat^d$ of rows of
$\krylov{\evMat}$; the letter $c$ for the columns of this degree $d$ block in
$\linPolMat{\mat{P}}$ and for the rows of the block $\evMat\mulmat^d$ in
$\krylov{\evMat}$.

\begin{example}[Linearization]
\label{eg:linearization}

In this example, we have $\rdim = \order = \maxDeg = 3$ and the base field is
the finite field with $ 97 $ elements; the input matrices are
\[ \evMat =
  \begin{bmatrix}
27 & 49 & 29 \\
50 & 58 & 0 \\
77 & 10 & 29
\end{bmatrix}  \qquad \text{and} \qquad
\mulshift = \begin{bmatrix}
0 & 1 & 0 \\
0 & 0 & 1 \\
0 & 0 & 0
\end{bmatrix} .\]
Then, we have
\[\krylov{\evMat} = \begin{bmatrix}
27 & 49 & 29 \\
50 & 58 & 0 \\
77 & 10 & 29 \\
0 & 27 & 49 \\
0 & 50 & 58 \\
0 & 77 & 10 \\
0 & 0 & 27 \\
0 & 0 & 50 \\
0 & 0 & 77 \\
0 & 0 & 0 \\
0 & 0 & 0 \\
0 & 0 & 0
\end{bmatrix} .\]
It is easily checked that $\row{p}_1 = (-1, -1, 1) \in \polRing^\rdim$ is an
interpolant for $(\evMat,\mulshift)$, since $\matrow{\evMat}{3} =
\matrow{\evMat}{1} + \matrow{\evMat}{2}$. Other interpolants are for example
$\row{p}_2 =  \left(3 X + 13,\,X + 57,\,0\right) $ which has row degree $1$,
$\row{p}_3 =  \left(X^{2} + 36 X,\,31 X,\,0\right) $ which has row degree $2$,
and $\row{p}_4 =  \left(X^{3},\,0,\,0\right) $ which has row degree $3$. We
have 
\[
\begin{array}{ccccccccccccccccccl}
  \linPolMat{\row{p}_1} & = & [ & 96 & 96 & 1 & | & 0  & 0 &  0 & | & 0 & 0 & 0 & | & 0 & 0 & 0 & ]  \\
  \linPolMat{\row{p}_2} & = & [ & 13 & 57 & 0 & | & 3  & 1 &  0 & | & 0 & 0 & 0 & | & 0 & 0 & 0 & ]  \\
  \linPolMat{\row{p}_3} & = & [ & 0  & 0  & 0 & | & 36 & 31 & 0 & | & 1 & 0 & 0 & | & 0 & 0 & 0 & ]  \\
  \linPolMat{\row{p}_4} & = & [ & 0  & 0  & 0 & | & 0  & 0 &  0 & | & 0 & 0 & 0 & | & 1 & 0 & 0 & ].
\end{array}
\]
Besides, one can check that the matrix
\[ \intBasis =  \begin{bmatrix}
X^{2} + 36 X & 31 X & 0 \\
3 X + 13 & X + 57 & 0 \\
96 & 96 & 1
\end{bmatrix} , \] whose rows are $(\row{p}_3,\row{p}_2,\row{p}_1)$ is a
reduced basis for the module $\intMod[\evMat,\mulshift]$ of Hermite-Pad\'e
approximants of order $3$ for $\row{f} = \left(29 X^{2} + 49 X + 27,\,58 X +
50,\,29 X^{2} + 10 X + 77\right)$.  \qed
\end{example}

Now,
we need tools to interpret the $\shifts$-minimality of an interpolation basis.
In \Cref{eg:linearization}, we see that $\row{p}_1$ has $\unishift$-row
degree $0$ and therefore appears in $\intBasis$; however $\row{p}_4$ has
$\unishift$-row degree $3$ and does not appear in $\intBasis$.  On the other
hand, considering $\shifts=(0,3,6)$, the $\shifts$-row degree of $\row{p}_4$
is $3$, while the one of $\row{p}_1$ is $6$: when forming rows of a
$(0,3,6)$-minimal interpolation basis, $\row{p}_4$ is a better candidate than
$\row{p}_1$. We see through this example that the uniform shift
$\shifts=\unishift$ leads to look in priority for relations involving the
first rows of the matrix $\krylov{\evMat}$; on the other hand, the shift
$\shifts=(0,3,6)$ leads to look for relations involving in priority the rows
$\matrow{\evMat}{1}$, $\matrow{\evMat}{1} \mulshift$, $\matrow{\evMat}{1}
\mulshift^2$, and $\matrow{\evMat}{1} \mulshift^3$ in $\krylov{\evMat}$
before considering the rows $\matrow{\evMat}{2}$ and $\matrow{\evMat}{3}$.

Going back to the general case, we define a notion of \emph{priority} of the
row $c$ of $\evMat \mulmat^d$ in $\krylov{\evMat}$. Let $\row{v} \in
\matSpace[1][\rdim(\maxDeg+1)]$ be any relation between the rows of
$\krylov{\evMat}$ involving this row, meaning that, writing $\row{p} =
\row{p}_0 + \row{p}_1 \var + \cdots + \row{p}_\maxDeg \var^\maxDeg=
\polFromLin{\row{v}}$ for the corresponding interpolant, the coefficient in
column $c$ of $\row{p}_d$ is nonzero. This implies that the $\shifts$-row
degree of $\row{p}$ is at least $\shift{c} + d$.  Since the $\shifts$-row
degree is precisely what we want to minimize in order to obtain an
$\shifts$-minimal interpolation basis, the priority of the rows of
$\krylov{\evMat}$ can be measured by the function $\prioEval$ defined by
$\prioEval(c,d) = \shift{c}+d$. Then, when computing relations between rows of
$\krylov{\evMat}$, we should use in priority the rows with low $\prioEval(d,r)$
in order to get interpolants with small $\shifts$-row degree.

To take this into account, we extend the linearization framework by using a
permutation of the rows of $\krylov{\evMat}$ so that they appear in
non-increasing order of their priority given by $\shifts$.  This way, an
interpolant with \emph{small} $\shifts$-row degree is always one whose
expansion forms a relation between the \emph{first} rows of the permuted
$\krylov{\evMat}$. To preserve properties such as $\intBasis \mul \evMat =
\linPolMat{\intBasis} \krylov{\evMat}$, we naturally permute the columns of
$\linPolMat{\mat{P}}$ accordingly. If $\ell = [\prioEval(1,0), \ldots,
\prioEval(\rdim,0), \prioEval(1,1), \ldots, \prioEval(\rdim,1), \ldots,
\prioEval(1,\maxDeg), \ldots, \prioEval(\rdim,\maxDeg)]$ in $\ZZ^{1 \times
\rdim(\maxDeg+1)}$ denotes the row vector indicating the priorities of the rows
of $\krylov{\evMat}$, then we choose an $\rdim(\maxDeg+1) \times
\rdim(\maxDeg+1)$ permutation matrix $\prioPerm$ such that the list $\ell
\prioPerm$ is non-decreasing. Then, the matrix $\prioPerm^{-1} \krylov{\evMat}$
is the matrix $\krylov{\evMat}$ with rows permuted so that they are arranged by
non-increasing priority, that is, by non-decreasing values of $\prioEval$.
Furthermore, the permutation $\prioPerm$ induces a bijection $\prioIndex$ which
keeps track of the position changes when applying the permutation: it
associates to $(c,d)$ the index $\prioIndex(c,d)$ of the element
$\prioEval(c,d)$ in the sorted list $\ell \prioPerm$. We now give precise
definitions.

\begin{definition}[Priority]
\label{dfn:priority}
Let $\evMat \in \matSpace[\rdim][\order]$ and $\mulmat \in \matSpace[\order]$,
and let $\shifts \in \shiftSpace$. The priority function $\prioEval:
\{1,\ldots,\rdim\} \times \{0,\ldots,\maxDeg\} \rightarrow \ZZ$ is defined by
$\prioEval(c,d) = \shift{c}+d$. Let $\ell = [\prioEval(1,0), \ldots,
\prioEval(\rdim,0), \prioEval(1,1), \ldots, \prioEval(\rdim,1), \ldots,
\prioEval(1,\maxDeg), \ldots, \prioEval(\rdim,\maxDeg)]$ be the sequence of
priorities in $\ZZ^{1 \times \rdim(\maxDeg+1)}$. Then, we define $\prioPerm$ as
the unique permutation matrix in $\matSpace[\rdim(\maxDeg+1)]$ along with the
corresponding indexing function
\begin{align*}
  & \prioIndex: \{1,\ldots,\rdim\} \times
\{0,\ldots,\maxDeg\} \rightarrow \{1,\ldots,\rdim(\maxDeg+1)\} \\
  & [\prioIndex(1,0), \ldots, \prioIndex(\rdim,0),
\ldots, \prioIndex(1,\maxDeg), \ldots, \prioIndex(\rdim,\maxDeg)] =
[1,2,\ldots,\rdim(\maxDeg+1)] \; \prioPerm
\end{align*}
which are such that
\begin{enumerate}[(i)]
  \item $\ell \prioPerm$ is non-decreasing;
  \item whenever $(c,d) \neq (c',d')$ are such that $\prioEval(c,d) =
    \prioEval(c',d')$, we have $c\neq c'$ and assuming without loss of
    generality that $c < c'$, then $\prioIndex(c,d) < \prioIndex(c',d')$.
\end{enumerate}
Besides, we define $\krylov[\shifts]{\evMat} = \prioPerm^{-1} \krylov{\evMat}$
as well as the shifted expansion $\linPolMat[\shifts]{\mat{P}} =
\linPolMat{\mat{P}} \prioPerm$ and the shifted compression
$\polFromLin[\shifts]{\mat{M}} = \polFromLin{\mat{M} \prioPerm}$.
\end{definition}

\noindent In other words, $\prioPerm$ is the unique permutation which
lexicographically sorts the sequence \[ [(\prioEval(1,0),1), \ldots,
(\prioEval(\rdim,0),\rdim), \ldots, (\prioEval(1,\maxDeg),1), \ldots,
(\prioEval(\rdim,\maxDeg),\rdim)]. \]
A representation of $\prioPerm$ can be computed using $\bigO{\rdim\maxDeg
\log(\rdim\maxDeg)}$ integer comparisons, and a representation of $\prioIndex$
can be computed using its definition in time linear in $\rdim(\maxDeg+1)$.
In the specific case of the uniform shift $\shifts=\unishift$, we have
$\prioEval(c,d) = d$, $\prioPerm$ is the identity matrix, and $\prioIndex(c,d)
= c + \rdim d$, and we have the identities $\krylov[\shifts]{\evMat} =
\krylov{\evMat}$, $\polFromLin[\shifts]{\mat{M}} = \polFromLin{\mat{M}}$,
$\linPolMat[\shifts]{\mat{P}} = \linPolMat{\mat{P}}$. The main ideas of the
rest of \Cref{sec:lin-mib} can be understood focusing on this particular
case.

\begin{example}[Input linearization, continued]
\label{eg:linearization-bis}


In the context of \Cref{eg:linearization}, if we consider the shifts $\shifts[s] =  \left(0,\,3,\,6\right) $ and $\shifts[t] =  \left(3,\,0,\,2\right) $, then we have \[ \krylov[{\shifts[s]}]{\evMat} =  \begin{bmatrix}
27 & 49 & 29 \\
0 & 27 & 49 \\
0 & 0 & 27 \\
0 & 0 & 0 \\
50 & 58 & 0 \\
0 & 50 & 58 \\
0 & 0 & 50 \\
0 & 0 & 0 \\
77 & 10 & 29 \\
0 & 77 & 10 \\
0 & 0 & 77 \\
0 & 0 & 0
\end{bmatrix} \qquad \text{and} \qquad \krylov[{\shifts[t]}]{\evMat} =
\begin{bmatrix}
50 & 58 & 0 \\
0 & 50 & 58 \\
0 & 0 & 50 \\
77 & 10 & 29 \\
27 & 49 & 29 \\
0 & 0 & 0 \\
0 & 77 & 10 \\
0 & 27 & 49 \\
0 & 0 & 77 \\
0 & 0 & 27 \\
0 & 0 & 0 \\
0 & 0 & 0
\end{bmatrix} . \] 
Besides, one can check that the shifted expansions of the interpolants  $\row{p}_1$, $\row{p}_2$, $\row{p}_3$, and $\row{p}_4$  with respect to $\shifts$ and $\shifts[t]$ are 
\[
\begin{array}{cccccccccccccccl}
\linPolMat[{\shifts}]{\row{p}_1}    & = & [ & 96 & 0 & 0 & 0 & 96 & 0 & 0 & 0 & 1 & 0 & 0 & 0 & ] \\ 
\linPolMat[{\shifts[t]}]{\row{p}_1} & = & [ & 96 & 0 & 0 & 1 & 96 & 0 & 0 & 0 & 0 & 0 & 0 & 0 & ]  \\ 
\linPolMat[{\shifts}]{\row{p}_2}    & = & [ & 13 & 3 & 0 & 0 & 57 & 1 & 0 & 0 & 0 & 0 & 0 & 0 & ] \\ 
\linPolMat[{\shifts[t]}]{\row{p}_2} & = & [ & 57 & 1 & 0 & 0 & 13 & 0 & 0 & 3 & 0 & 0 & 0 & 0 & ] \\ 
\linPolMat[{\shifts}]{\row{p}_3}    & = & [ & 0 & 36 & 1 & 0 & 0 & 31 & 0 & 0 & 0 & 0 & 0 & 0 & ] \\ 
\linPolMat[{\shifts[t]}]{\row{p}_3} & = & [ & 0 & 31 & 0 & 0 & 0 & 0 & 0 & 36 & 0 & 1 & 0 & 0  & ]\\ 
\linPolMat[{\shifts}]{\row{p}_4}    & = & [ & 0 & 0 & 0 & 1 & 0 & 0 & 0 & 0 & 0 & 0 & 0 & 0  & ]  \\ 
\linPolMat[{\shifts[t]}]{\row{p}_4} & = & [ & 0 & 0 & 0 & 0 & 0 & 0 & 0 & 0 & 0 & 0 & 0 & 1 & ].
\text{\qed}\end{array}
\]

\end{example}

\subsection{Minimal linear relations and minimal interpolation bases}
\label[subsec]{subsec:rel-int}

From the previous subsection, the intuition is that the minimality of
interpolants can be read on the corresponding linear relations between the rows
of $\krylov[\shifts]{\evMat}$, as the fact that they involve in priority the
first rows. Here, we support this intuition with rigorous statements,
presenting a notion of minimality for linear relations between the rows of
$\krylov[\shifts]{\evMat}$, and showing that an $\shifts$-minimal interpolation
basis for $(\evMat,\mulmat)$ corresponds to a specific set of $\rdim$ such
minimal relations.

First we show that, given a polynomial row vector and a degree shift, one can
directly read the pivot index~\citep[Section 2]{MulSto03} of the vector from its
expansion. Extending the definitions in~\citep{MulSto03} to the shifted case, we
define the $\shifts$-pivot index, $\shifts$-pivot entry, and $\shifts$-pivot
degree of a nonzero row vector as follows.

\begin{definition}[Pivot]
\label{dfn:pivot}
Let $\row{p} = [p_c]_c \in \polMatSpace[1][\rdim]$ be a nonzero row vector and
let $\shifts \in \shiftSpace$ be a degree shift. The \emph{$\shifts$-pivot
index} of $\row{p}$ is the largest column index $c \in \{1,\ldots,\rdim\}$ such
that $\deg(p_c) + \shift{c}$ is equal to the $\shifts$-row degree
$\rdeg[\shifts]{\row{p}}$ of this row; then, $p_c$ and $\deg(p_c)$ are called
the \emph{$\shifts$-pivot entry} and the \emph{$\shifts$-pivot degree} of
$\row{p}$, respectively.
\end{definition}

\noindent The following result will be useful for our purpose, since pivot
indices can be used to easily identify some specific forms of reduced
polynomial matrices.

\begin{lemma}
\label{lem:pivot}
Let $\row{p} = [p_c]_c \in \polMatSpace[1][\rdim]_{\le\maxDeg}$ be a nonzero
row vector.  Then, $i = \prioIndex(c,d)$ is the column index of the rightmost
nonzero coefficient in $\linPolMat[\shifts]{\row{p}}$ if and only if
$\row{p}$ has $\shifts$-pivot index~$c$ and $\shifts$-pivot degree $\deg(p_c)
= d$.
\end{lemma}
\begin{proof}
We distinguish three sets of entries of $\linPolMat[\shifts]{\row{p}}$ with
column index $\ge i$: the one at index $i$, the ones that have a higher
$\prioEval$, and the ones that have the same $\prioEval$:
\begin{itemize}
  \item if the coefficient at index $i=\prioIndex(c,d)$ in
    $\linPolMat[\shifts]{\row{p}}$ is nonzero then
    $\deg(p_c) \ge d$, and if $\deg(p_c) = d$ then
    the coefficient at index $i=\prioIndex(c,d)$ in
    $\linPolMat[\shifts]{\row{p}}$ is nonzero;
  \item the coefficient at index $\prioIndex(c',d')$ in
    $\linPolMat[\shifts]{\row{p}}$ is zero for all $(c',d')$ such that
    $\prioEval(c',d') > \prioEval(c,d)$ if and only if $ \shift{c'} +
    \deg(p_{c'}) \le \shift{c} + d$ for all $1 \le c' \le \rdim$;
  \item assuming $ \shift{c'} + \deg(p_{c'}) \le \shift{c} + d$ for all $1
    \le c' \le \rdim$, the coefficient at index $\prioIndex(c',d')$ in
    $\linPolMat[\shifts]{\row{p}}$ is zero for all $(c',d')$ such that
    $\prioEval(c',d') = \prioEval(c,d)$ with $\prioIndex(c',d') > i =
    \prioIndex(c,d)$ (by definition of $\prioPerm$, this implies $c' > c$) if
    and only if we have $\shift{c'} + \deg(p_{c'}) < \shift{c} + d$ for all
    $c'>c$;
\end{itemize}
these three points prove the equivalence.
\end{proof}

\noindent We have seen that an interpolant for $(\evMat,\mulmat)$ corresponds
to a linear relation between the rows of $\krylov[\shifts]{\evMat}$.  From
this perspective, the preceding result implies that an interpolant with
$\shifts$-pivot index $c$ and $\shifts$-pivot degree $d$ corresponds to a
linear relation which expresses the row at index $\prioIndex(c,d)$ in
$\krylov[\shifts]{\evMat}$ as a linear combination of the rows at indices
smaller than $\prioIndex(c,d)$. Now, we give a precise correspondence between
minimal interpolation bases and sets of linear relations which involve in
priority the first rows of $\krylov[\shifts]{\evMat}$.

\begin{example}[Minimal relations]
\label{eg:minimal-relations}
Let us consider the context of \Cref{eg:linearization} with the uniform
shift. As mentioned above, the matrix $\intBasis$ whose rows are
$(\row{p}_3,\row{p}_2,\row{p}_1)$ is a minimal interpolation basis.  The pivot
indices of $\row{p}_3,\row{p}_2,\row{p}_1$ are $1$, $2$, $3$, and their pivot
degrees are $2$, $1$, $0$. Besides, we remark that 
\begin{itemize}
\item the relation $\linPolMat{\row{p}_1}$ involves the row $c=3$ of
  $\evMat$ and the rows above this one in $\krylov{\evMat}$;
\item $\linPolMat{\row{p}_2}$ involves the row $c=2$ of $\evMat \mulshift$
  and the rows above this one in $\krylov{\evMat}$, and there is no linear
  relation involving the row $c=2$ of $\evMat$ and the rows above it
  in $\krylov{\evMat}$;
\item $\linPolMat{\row{p}_3}$ involves the row $c=1$ of $\evMat \mulshift^2$
  and the rows above it in $\krylov{\evMat}$: one can check that there is
  no linear relation between the row $c=1$ of $\evMat \mulshift$ and the rows
  above it in $\krylov{\evMat}$. \qed
\end{itemize}
\end{example}

This example suggests that we can give a link between the minimal row degree of
a minimal interpolation basis and some minimal exponent $\minDeg_c$ such that
the row $c$ of the block $\evMat \mulmat^{\minDeg_c}$ is a linear combination
on the rows above it in $\krylov{\evMat}$. Extending this to the case of any
shift $\shifts$ leads us to the following definition, which is reminiscent of
the so-called \emph{minimal indices} (or \emph{Kronecker indices}) for
$\unishift$-minimal nullspace bases~\citep[Section 6.5.4]{Kailath80}.

\begin{definition}[Minimal degree]
\label{dfn:min-deg}
Let $\mulmat \in \matSpace[\order]$ and $\evMat \in \matSpace[\rdim][\order]$,
and let $\shifts \in \shiftSpace$.  The \emph{$\shifts$-minimal degree} of
$(\evMat,\mulmat)$ is the tuple $(\minDeg_1,\ldots,\minDeg_\rdim)$ where for
each $c \in \{1,\ldots,\rdim\}$, $\minDeg_c \in \NN$ is the smallest exponent
such that the row $\matrow{ {\evMat \mulmat^{\minDeg_c}} }{c} = \matrow{
\krylov[\shifts]{\evMat} }{\prioIndex(c,\minDeg_c)}$ is a linear combination
of the rows in $\{\matrow{\krylov[\shifts]{\evMat}}{i},
i<\prioIndex(c,\minDeg_c)\}$.
\end{definition}

\noindent 
We note that we have $\minDeg_c \le \maxDeg$ for every $c$, since the minimal
polynomial of the matrix $\mulmat \in \matSpace[\order]$ is of degree at most
$\maxDeg$. We now state in \Cref{lem:lin-piv-deg} and
\Cref{cor:weak-popov-mib} that the minimal degree of
$(\evMat,\mulmat)$ indeed corresponds to a notion of minimality of interpolants
and interpolation bases. Until the end of this \cref{subsec:rel-int},
we fix a matrix $\evMat \in \evSpace{\order}$ and a matrix $\mulmat \in
\matSpace[\order]$.

\begin{lemma}
\label{lem:lin-piv-deg}
Let $\int = [p_1,\ldots,p_\rdim] \in\polMatSpace[1][\rdim]_{\le\maxDeg}$,
$\shifts \in \shiftSpace$, and $c$ in $\{1,\ldots,\rdim\}$. If $\int$ is an
interpolant for $(\evMat,\mulmat)$ with $\shifts$-pivot index $c$, then $\int$
has $\shifts$-pivot degree $\deg(p_c) \ge \minDeg_c$. Besides, there is an
interpolant $\int \in \polMatSpace[1][\rdim]_{\le\maxDeg}$ for
$(\evMat,\mulmat)$ which has $\shifts$-pivot index $c$ and $\shifts$-pivot
degree $\deg(p_c) = \minDeg_c$.
\end{lemma}
\begin{proof}
First, assume $\int$ is an interpolant with $\shifts$-pivot index $c$, and let
$d = \deg(p_c)$ be the degree of the $\shifts$-pivot entry of $\int$.
According to \Cref{lem:pivot}, the rightmost nonzero element of
$\linPolMat[\shifts]{\int}$ is at index $\prioIndex(c,d)$, and since $\int$ is
an interpolant for $(\evMat,\mulmat)$ we have $\linPolMat[\shifts]{\mat{P}}
\krylov[\shifts]{\evMat} = 0$.  This implies that the row $\prioIndex(c,d)$ of
$\krylov[\shifts]{\evMat}$ is a linear combination of the rows in $\{\matrow{
\krylov[\shifts]{\evMat} }{i}, i<\prioIndex(c,d) \}$, which in turn implies
$d\ge \minDeg_c$ by definition of $\minDeg_c$.  Now, the definition of
$\minDeg_c$ also ensures that the row $\prioIndex(c,\minDeg_c)$ of
$\krylov[\shifts]{\evMat}$ is a linear combination of the rows $\{\matrow{
\krylov[\shifts]{\evMat} }{i}, i<\prioIndex(c,d_c) \}$.  This linear
combination forms a vector $\rowvec$ in the nullspace of
$\krylov[\shifts]{\evMat}$ with its rightmost nonzero element at index
$\prioIndex(c,\minDeg_c)$; then by \Cref{lem:pivot}, $\int =
\polFromLin[\shifts]{\rowvec}$ is an interpolant with $\shifts$-pivot index
$c$ and $\shifts$-pivot degree $\minDeg_c$. Besides, $\int$ has degree at most
$\maxDeg$ by construction.
\end{proof}

Now, we want to extend these considerations on row vectors and
interpolants to matrices and interpolation bases. In connection with the notion
of pivot of a row, there is a specific form of reduced matrices called the weak
Popov form~\citep{MulSto03}, for which we extend the definition to any shift
$\shifts$ as follows.

\begin{definition}[weak Popov form, pivot degree]
\label{dfn:weak-popov}
Let $\mat{P}$ in $\polMatSpace[\rdim]$ have full rank, and let $\shifts$ in
$\shiftSpace$. Then, $\mat{P}$ is said to be in \emph{$\shifts$-weak Popov
form} if the $\shifts$-pivot indices of its rows are pairwise distinct.
Furthermore, the \emph{$\shifts$-pivot degree} of $\mat{P}$ is the tuple
$(d_1,\ldots,d_\rdim) \in \NN^\rdim$ where for $c \in \{1,\ldots,\rdim\}$,
$d_c$ is the $\shifts$-pivot degree of the row of $\mat{P}$ which has
$\shifts$-pivot index $c$.
\end{definition}

\noindent
A matrix $\mat{P}$ in $\shifts$-weak Popov form is in particular
$\shifts$-reduced. Then, \Cref{lem:lin-piv-deg} leads to the following
result; we remark that even though the matrix $\mat{P}$ in this corollary is
$\shifts$-reduced and each of its rows is an interpolant, we do not yet claim
that it is an interpolation basis. 

\begin{corollary}
\label{cor:weak-popov-mib}
There is a matrix $\mat{P} \in \polMatSpace[\rdim]_{\le \maxDeg}$ in
$\shifts$-weak Popov form, with $\shifts$-pivot entries on the diagonal
and $\shifts$-pivot degree $(\minDeg_1,\ldots,\minDeg_\rdim)$, such that
every row of $\mat{P}$ is an interpolant for $(\evMat,\mulmat)$.
\end{corollary}
\begin{proof}
For every $c$ in $\{1,\ldots,\rdim\}$, \Cref{lem:lin-piv-deg} shows
that there is an interpolant $\row{p}_c$ for $(\evMat,\mulmat)$ which has
degree at most $\maxDeg$, has $\shifts$-pivot index $c$, and has
$\shifts$-pivot degree $\minDeg_c$.  Then, considering the matrix $\mat{P}$
in $\polMatSpace[\rdim]$ whose row $c$ is $\row{p}_c$ gives the conclusion.
\end{proof}

We conclude this section by proving that the $\shifts$-minimal degree of
$(\evMat,\mulmat)$ is directly linked to the $\shifts$-row degree of a minimal
interpolation basis, which proves in particular that the matrix $\intBasis$ in
\Cref{cor:weak-popov-mib} is an $\shifts$-minimal interpolation basis
for $(\evMat,\mulmat)$.

\begin{lemma}
\label{lem:red-to-mib}
Let $\spolmat[1] \in \intSpace$ be $\shifts$-reduced such that each row of
$\spolmat[1]$ is an interpolant for $(\evMat,\mulmat)$.  Then, $\spolmat[1]$
is an interpolation basis for $(\evMat,\mulmat)$ if and only if the
$\shifts$-row degree $\rdeg[\shifts]{\spolmat[1]}$ of $\spolmat[1]$ is
$(\shift{1}+\minDeg_1, \ldots, \shift{\rdim}+\minDeg_\rdim)$ up to a
permutation of the rows of $\spolmat[1]$. In particular, if $\intBasis$ is a
matrix as in \Cref{cor:weak-popov-mib}, then $\intBasis$ is an
$\shifts$-minimal interpolation basis for $(\evMat,\mulmat)$.
\end{lemma}
\begin{proof}
We denote $\intBasis \in \intSpace_{\le \maxDeg}$ a matrix as in
\Cref{cor:weak-popov-mib}; $\intBasis$ is in particular
$\shifts$-reduced and has $\shifts$-row degree exactly $(\shift{1}+\minDeg_1,
\ldots, \shift{\rdim}+\minDeg_\rdim)$.

First, we assume that $\spolmat[1]$ is an interpolation basis for
$(\evMat,\mulmat)$.  Remarking that a matrix $\spolmat[2]$ is in $\shifts$-weak
Popov form if and only if $\spolmat[2]\shiftMat{\shifts}$ is in weak Popov
form, we know from~\citep[Section 2]{MulSto03} that $\spolmat[1]$ is
left-unimodularly equivalent to a matrix $\spolmat[2]$ in $\shifts$-weak Popov
form. Besides, up to a permutation of the rows of $\spolmat[2]$, we assume
without loss of generality that the pivot entries $\spolmat[2]$ are on the
diagonal. Then, denoting its $\shifts$-pivot degree by $(d_1,\ldots,d_\rdim)$,
the $\shifts$-row degree of $\spolmat[2]$ is $(\shift{1}+d_1, \ldots,
\shift{\rdim}+d_\rdim)$. Since $\spolmat[1]$ and $\spolmat[2]$ are
$\shifts$-reduced and unimodularly equivalent, they have the same $\shifts$-row
degree up to a permutation \citep[Lemma 6.3-14]{Kailath80}: thus it is enough
to prove that $(d_1, \ldots, d_\rdim) = (\minDeg_1, \ldots, \minDeg_\rdim)$.
By \Cref{lem:lin-piv-deg} applied to each row of $\spolmat[2]$,
$d_1,\ldots,d_\rdim$ are at least $\minDeg_1,\ldots,\minDeg_\rdim$,
respectively. On the other hand, since $\spolmat[2]$ is an interpolation basis,
there is a nonsingular matrix $\mat{U}\in\intSpace$ such that $\intBasis =
\mat{U}\spolmat[2]$. Since $\intBasis$ is $\shifts$-reduced with the
$\shifts$-row degree $(\shift{1}+\minDeg_1, \ldots,
\shift{\rdim}+\minDeg_\rdim)$, we have $\deg(\det(\intBasis)) =
\sshifts[{\rdeg[\shifts]{\intBasis}}] - \sshifts = \minDeg_1 + \cdots +
\minDeg_\rdim$ \citep[Lemma~2.10]{Zhou12}. Similarly, we have
$\deg(\det(\spolmat[2])) = \sshifts[{\rdeg[\shifts]{\spolmat[2]}}] - \sshifts =
d_1 + \cdots + d_\rdim$.  Considering the determinantal degree in the identity
$\intBasis = \mat{U}\spolmat[2]$ yields $\minDeg_1+\cdots+\minDeg_\rdim \ge
d_1+\cdots+d_\rdim$, from which we conclude $d_c = \minDeg_c$ for all $1\le c
\le \rdim$.

Now, we note that this also implies that the determinant of $\mat{U}$ is
constant, thus $\mat{U}$ is unimodular and consequently $\intBasis$ is an
interpolation basis: since $\intBasis$ is $\shifts$-reduced by construction,
it is an $\shifts$-minimal interpolation basis.

Finally, we assume that $\spolmat[1]$ has $\shifts$-row degree
$(\shift{1}+\minDeg_1, \ldots, \shift{\rdim}+\minDeg_\rdim)$ up to a
permutation. Since $\intBasis$ is an interpolation basis, there is a
nonsingular matrix $\mat{U}\in\intSpace$ such that
$\spolmat[1]=\mat{U}\intBasis$.  Since $\spolmat[1]$ is $\shifts$-reduced, we
have $\deg(\det(\spolmat[1])) = \sshifts[{\rdeg[\shifts]{\spolmat[1]}}] -
\sshifts = \minDeg_1 + \cdots + \minDeg_\rdim$.  Then, considering the
determinantal degree in the identity $\spolmat[1]=\mat{U}\intBasis$ shows that
the determinant of $\mat{U}$ is a nonzero constant, that is, $\mat{U}$ is
unimodular.  Thus, $\spolmat[1]$ is an interpolation basis.
\end{proof}

\begin{remark}
  \label{rmk:elimination}
As can be observed in the definition of the $\shifts$-minimal degree and in the
proof of \Cref{lem:lin-piv-deg}, one can use Gaussian elimination on the
rows of $\krylov[\shifts]{\evMat}$ to build each row of the $\shifts$-minimal
interpolation basis~$\intBasis$. This gives a method for solving
\Cref{pbm:mib} using linear algebra. Then, the main goal in the rest of
this section is to show how to perform the computation of $\intBasis$
efficiently. \qed
\end{remark}

\subsection{Row rank profile and minimal degree}
\label[subsec]{subsec:rank-profile}

The reader may have noted that $\krylov[\shifts]{\evMat}$ has
$\rdim\order(\maxDeg+1)$ coefficients in $\field$, and in general $\rdim \order
\maxDeg$ may be beyond our target cost bound given in
\Cref{prop:lin-mib}. Here, we show that one can focus on a small
subset of the rows of $\krylov[\shifts]{\evMat}$ which contains enough
information to compute linear relations leading to a matrix $\mat{P}$ as in
\Cref{cor:weak-popov-mib}. Then, we present a fast algorithm to
compute this subset of rows. We also use these results to bound the average
$\shifts$-row degree of any $\shifts$-minimal interpolation basis.

To begin with, we give some helpful structure properties of
$\krylov[\shifts]{\evMat}$, which will be central in choosing the subset of
rows and in the designing a fast algorithm which computes independent rows in
$\krylov[\shifts]{\evMat}$ without having to consider the whole matrix. 

\begin{lemma}[{Structure of $\krylov[\shifts]{\evMat}$}]
\label{lem:structure}
Let $\mulmat \in \matSpace[\order]$ and $\evMat \in
\matSpace[\rdim][\order]$, and let $\shifts \in \shiftSpace$.
Let $\prioIndex,\prioEval$ be as in \Cref{dfn:priority}.
\begin{itemize}
  \item For each $c$, $d \mapsto \prioIndex(c,d)$ is strictly increasing.
  \item If $(c,d)$ and $(c',d')$ are such that $\prioIndex(c,d) <
    \prioIndex(c',d')$, then for any $k \le \min(\maxDeg- d, \maxDeg-d')$ we
    have $\prioIndex(c,d+k) < \prioIndex(c',d'+k)$.
  \item Suppose that for some $i \in \{1,\ldots,\rdim(\maxDeg+1)\}$, the row at
    index $i$ in $\krylov[\shifts]{\evMat}$ is a linear combination of the
    rows of $\krylov[\shifts]{\evMat}$ with indices in $\{1,\ldots,i-1\}$.
    Then, writing $i= \prioIndex(c,d)$, for every $d' \in
    \{0,\ldots,\maxDeg-d\}$ the row at index $i' = \prioIndex(c,d+d')$ in
    $\krylov[\shifts]{\evMat}$ is a linear combination of the rows of
    $\krylov[\shifts]{\evMat}$ with indices in $\{1,\ldots,i'-1\}$.
\end{itemize}
\end{lemma}
\begin{proof}
  The first item is clear since $\prioEval(c,d) = \shift{c}+d$.  For the
  second item, we consider two cases. First, if $\prioEval(c,d) <
  \prioEval(c',d')$, this means $\shift{c} + d < \shift{c'} + d'$ from which
  we obviously obtain $\prioEval(c,d+k) < \prioEval(c',d'+k)$, and in
  particular $\prioIndex(c,d+k) < \prioIndex(c',d'+k)$. Second, if
  $\prioEval(c,d) = \prioEval(c',d')$, we must have $c<c'$ by choice of
  $\prioIndex$, and then we also have $\prioEval(c,d+k) = \prioEval(c',d'+k)$
  with $c<c'$ which implies that $\prioIndex(c,d+k) < \prioIndex(c',d'+k)$.

  The third item is a direct rewriting in the linearization framework of the
  following property. Let $\row{p} \in \polMatSpace[1][\rdim]_{\le \maxDeg}$ be
  an interpolant for $(\evMat,\mulmat)$ with $\shifts$-pivot index $c$ and
  $\shifts$-row degree $d$, and consider $d' \le \maxDeg - d$; then the row
  vector $\var^{d'} \row{p}$, with entries of degree more than $\delta$ taken
  modulo the minimal polynomial of $\mulmat$, is an interpolant for
  $(\evMat,\mulmat)$ with $\shifts$-pivot index $c$ and $\shifts$-row degree
  $d+d'$.
\end{proof}

We remark that, when choosing a subset of $r =
\mathrm{rank}(\krylov[\shifts]{\evMat})$ linearly independent rows in
$\krylov[\shifts]{\evMat}$, all other rows in the matrix are linear
combinations of those. Because our goal is to find relations which involve in
priority the first rows of $\krylov[\shifts]{\evMat}$, we are specifically
interested in the \emph{first} $r$ independent rows in
$\krylov[\shifts]{\evMat}$. More precisely, we focus on the \emph{row rank
profile} $(i_1,\ldots,i_r)$ of $\krylov[\shifts]{\evMat}$, that is, the
lexicographically smallest tuple with entries in
$\{1,\ldots,\rdim(\maxDeg+1)\}$ such that the submatrix of
$\krylov[\shifts]{\evMat}$ formed by the rows with indices in
$\{i_1,\ldots,i_r\}$ has rank $r = \mathrm{rank}(\krylov[\shifts]{\evMat})$.
Then, for each $c$ the row $\matrow{ {\evMat \mulmat^{\minDeg_c}} }{c} =
\matrow{\krylov[\shifts]{\evMat}}{\prioIndex(c,\minDeg_c)}$ is a linear
combination of the rows in $\{ \matrow{\krylov[\shifts]{\evMat}}{i_k}, 1\le k
\le r \}$. We now show that this row rank profile is directly related to the
$\shifts$-minimal degree of $(\evMat,\mulmat)$.

\begin{lemma}
\label{lem:rank-profile}
Let $\mulmat \in \matSpace[\order]$ and $\evMat \in
\matSpace[\rdim][\order]$, let $\shifts \in \shiftSpace$, and
let $(i_1,\ldots,i_r)$ be the row rank profile of $\krylov[\shifts]{\evMat}$.
For $k \in \{1,\ldots,r\}$, we write $i_k = \prioIndex(c_k,d_k)$ for some
(unique) $1 \le c_k \le \rdim$ and $0 \le d_k \le \maxDeg$. Given $c \in
\{1,\ldots,\rdim\}$,
\begin{itemize}
  \item for all $0 \le d < \minDeg_c$ we have $\prioIndex(c,d) \in
    \{i_1,\ldots,i_r\}$,
  \item $\minDeg_c = 0$ if and only if $c_k \neq c$ for all $k \in
    \{1,\ldots,r\}$,
  \item if $\minDeg_c > 0$ we have $\minDeg_c = 1 + \max \{d_k \mid 1\le k \le
    r \text{ and } c_k = c \}$.
\end{itemize}
\end{lemma}
\begin{proof}
Let us fix $c \in \{1,\ldots,\rdim\}$. We recall that $\minDeg_c$ is the
smallest exponent such that the row $\matrow{{\evMat \mulmat^{\minDeg_c}}}{c} =
\matrow{ \krylov[\shifts]{\evMat} }{\prioIndex(c,\minDeg_c)}$ is a linear
combination of the rows in $\{\matrow{\krylov[\shifts]{\evMat}}{i},
i<\prioIndex(c,\minDeg_c)\}$.

First, we assume that $\minDeg_c > 0$ and we let $d < \minDeg_c$.  By
definition of $\minDeg_c$, the row at index $\prioIndex(c,d)$ in
$\krylov[\shifts]{\evMat}$ is linearly independent from the rows with smaller
indices. Thus, by minimality of the row rank profile, $\prioIndex(c,d) \in
\{i_1,\ldots,i_r\}$.

In particular, choosing $d=0$, we obtain that $\prioIndex(c,0) \in
\{i_1,\ldots,i_r\}$, or in other words, there is some $k \in \{1,\ldots,r\}$
such that $(c,0) = (c_k,d_k)$. This proves that if $\minDeg_c > 0$, then $c_k =
c$ for some $k \in \{1,\ldots,r\}$.

Now, we assume that $\minDeg_c = 0$ and we show that $c \neq c_k$ for all $k
\in \{1,\ldots,r\}$. The definition of $\minDeg_c = 0$ and the third item in
\Cref{lem:structure} together prove that for every $d \in \{0,\ldots,
\maxDeg\}$, the row $\matrow{\evMat}{c} \mulmat^d$ which is at index
$\prioIndex(c,d)$ in $\krylov[\shifts]{\evMat}$ is a linear combination of the
rows with smaller indices. Then, by minimality of the row rank profile,
$\prioIndex(c,d) \not\in \{i_1,\ldots,i_r\}$ for all $d \in \{0,\ldots,
\maxDeg\}$, and thus in particular $(c,d_k) \neq (c_k,d_k)$ for all $k \in
\{1,\ldots,r\}$.

Finally, we assume that $\minDeg_c > 0$ and we show that $\minDeg_c = 1 + \max
\{ d_k \mid 1 \le k \le r \text{ and } c_k = c \}$. Using the first item with
$d = \minDeg_c-1$, there exists $\bar{k} \in \{1,\ldots,r\}$ such that
$(c_{\bar{k}},d_{\bar{k}}) = (c,\minDeg_c-1)$. As in the previous paragraph,
the definition of $\minDeg_c$, the third item in \Cref{lem:structure}, and
the minimality of the row rank profile imply that $\prioIndex(c,d) \not\in
\{i_1,\ldots,i_r\}$ for all $d \in \{\minDeg_c,\ldots,\maxDeg\}$; in
particular, $(c,d_k) \neq (c_k,d_k)$ for all $k \in \{1,\ldots,r\}$ such that
$d_k > d_{\bar{k}}$. Thus, we have $d_{\bar{k}} = \max \{d_k \mid 1\le k \le r
\text{ and } c_k = c \}$.
\end{proof}

\begin{example}[Minimal degree and row rank profile]
  \label{eg:mindeg}
In the context of \Cref{eg:linearization,eg:linearization-bis},

\begin{itemize}
  \item for the uniform shift, the row rank profile of
    $\krylov[\unishift]{\evMat}$ is $(0,1,3)$ with $0 =
    \prioIndex[\unishift](0,0)$, $1 = \prioIndex[\unishift](1,0)$, and $3 =
    \prioIndex[\unishift](0,1)$: then, the $\unishift$-minimal degree of
    $(\evMat,\mulshift)$ is $(2,1,0)$;
  \item for the shift $\shifts = (0,3,6)$, the row rank profile of $\krylov{\evMat}$
    is $(0,1,2)$ with $0 = \prioIndex(0,0)$, $1 = \prioIndex(0,1)$, and $2 =
    \prioIndex(0,2)$: the $\shifts$-minimal degree of 
    $(\evMat,\mulshift)$ is $(3,0,0)$;
  \item for the shift $\shifts[t] = (3,0,2)$, the row rank profile of
    $\krylov[{\shifts[t]}]{\evMat}$ is $(0,1,2)$ with $0 =
    \prioIndex[{\shifts[t]}](1,0)$, $1 = \prioIndex[{\shifts[t]}](1,1)$, and $2
    = \prioIndex[{\shifts[t]}](1,2)$: the $\shifts[t]$-minimal degree of
    $(\evMat,\mulshift)$ is $(0,3,0)$. \qed
\end{itemize}
\end{example}

In particular, the previous lemma implies a bound on the $\shifts$-minimal
degree of $(\evMat,\mulmat)$. Since the minimal polynomial of the matrix
$\mulmat \in \matSpace[\order]$ is of degree at most $\maxDeg$, we have
$\minDeg_c \le \maxDeg \le \order$ for $1 \le c \le \rdim$: we actually have
the following stronger identity, which shows that the sum of
$\minDeg_1,\ldots,\minDeg_\rdim$ is at most $\order$.

\begin{lemma}
\label{lem:sum-mindeg}
Let $\mulmat \in \matSpace[\order]$ and $\evMat \in \matSpace[\rdim][\order]$,
and let $\shifts \in \shiftSpace$. Let $(\minDeg_1,\ldots,\minDeg_\rdim)$ be
the $\shifts$-minimal degree of $(\evMat,\mulmat)$. Then,
$\minDeg_1+\cdots+\minDeg_\rdim = \mathrm{rank}(\krylov[\shifts]{\evMat})$.
\end{lemma}
\begin{proof}
  From \Cref{lem:rank-profile}, we see that one can partition the set
  $\{i_1,\ldots,i_r\}$ as the disjoint union of the sets $\{ \prioIndex(c,d) ,
  0 \le d < \minDeg_c \}$ for each $c$ with $\minDeg_c > 0$. This union has
  cardinality $\minDeg_1 + \cdots + \minDeg_\rdim$, and the set
  $\{i_1,\ldots,i_r\}$ has cardinality $r =
  \mathrm{rank}(\krylov[\shifts]{\evMat})$.
\end{proof}
\noindent

\begin{remark}
  \label{rmk:sum_rdeg}
Combining this with \Cref{lem:red-to-mib}, one can directly deduce the
following bound on the average row degree of minimal interpolation bases:

\noindent\emph{Let $\mulmat \in \matSpace[\order]$ and $\evMat \in
  \matSpace[\rdim][\order]$, and let $\shifts \in \shiftSpace$. Then, for any
  $\shifts$-minimal interpolation basis $\intBasis$ for $(\evMat,\mulmat)$, the
  sum of the $\shifts$-row degrees of $\intBasis$ satisfies
  $\sshifts[{\rdeg[\shifts]{\intBasis}}] \le \order + \sshifts$.}

This bound has already been given before in \citep[Theorem 7.3.(b)]{BecLab00},
and also in the context of M-Pad\'e approximation \citep[Theorem
4.1]{BarBul92}, which includes order basis computation. This result is central
regarding the cost of algorithms which compute shifted minimal interpolation
bases since it gives a bound on the size of the output matrix. In particular it
is a keystone for the efficiency of our divide-and-conquer algorithm in
\Cref{sec:algo}, where it gives a bound on the average row degree of all
intermediate bases and thus allows fast computation of the product of bases
(\Cref{sec:cost-mult}), of the change of shift (\Cref{sec:change-shift}), and
of the residuals (\Cref{sec:multiplication}).  \qed
\end{remark}

Now, we show how to compute the row rank profile of $\krylov[\shifts]{\evMat}$
efficiently. In the style of the algorithm of~\citet[p.~313]{KelGeh85}, our
algorithm processes submatrices of $\krylov[\shifts]{\evMat}$ containing all
rows up to some degree, doubling this degree at each iteration. The structure
property in \Cref{lem:structure} allows us to always consider at most $2r$
rows of $\krylov[\shifts]{\evMat}$, discarding most rows with indices not in
$\{i_1,\ldots,i_r\}$ without computing them. (There is one exception at the
beginning, where the $\rdim$ rows of $\evMat$ are considered, with possibly
$\rdim$ much larger than $r$.) This algorithm also returns the submatrix formed
by the rows corresponding to the row rank profile, as well as the column rank
profile of this submatrix, since they will both be useful later in
\Cref{subsec:lin-mib-algo}.

\begin{proposition}
\label{prop:algo-rank-profile}
\Cref{algo:rank-profile} is correct and uses $\bigO{ \order^\expmatmul
( \lceil \rdim/\order \rceil + \log(\maxDeg))}$ operations in $\field$ if
$\expmatmul>2$, and $\bigO{ \order^2 ( \lceil \rdim/\order \rceil +
\log(\maxDeg))\log(\order)}$ operations if $\expmatmul=2$.
\end{proposition}
\begin{proof}
The algorithm takes as input $\maxDeg$ a power of $2$: one can always ensure
this by taking the next power of $2$ without impacting the cost bound. After
Steps~\textbf{2},~\textbf{3}, and~\textbf{4}, $(i_1,\ldots,i_r)$ correspond to
the indices in $\krylov[\shifts]{\evMat}$ of the row rank profile of $\evMat$,
and $\mat{M}$ is the submatrix of $\krylov[\shifts]{\evMat}$ formed by the rows
with indices in $\{i_1,\ldots,i_r\}$.  Relying on the algorithm
in~\citep[Section 2.2]{Storjohann00}, Step~\textbf{2} can be performed using
$\bigO{ \rdim \order (\min(\rdim,\order)^{\expmatmul-2} +
\log(\min(\rdim,\order)) }$ operations, and Step~\textbf{6} can be computed
using $\bigO{r \order (r^{\expmatmul-2} + \log(r))}$ operations (where the
logarithmic terms account for the possibility that $\expmatmul=2$). The loop
performs $\log(\maxDeg)$ iterations. In each iteration $\ell$, since the
matrix $\mat{M}$ has $\order$ columns and has at most $2r$ rows with $r\le
\order$, one can compute the square $\mulmat^{2^{\ell}}$ and the product
$\mat{M} \mulmat^{2^\ell}$ using $\bigO{\order^\expmatmul}$ operations, and the
row rank profile of $\mat{M}$ using $\bigO{r \order (r^{\expmatmul-2} +
\log(r))}$ operations~\citep[Section 2.2]{Storjohann00}. Thus, overall, the for
loop uses $\bigO{\order^\expmatmul \log(\maxDeg)}$ operations if
$\expmatmul>2$, and $\bigO{\order^2 \log(\maxDeg) \log(\order)}$
operations if $\expmatmul=2$. Adding these costs leads to the announced bound.

Let us now prove the correctness of the algorithm. For each $\ell \in
\{0,\ldots,\log(\maxDeg)\}$ let $\mathcal{I}_\ell = \{ \prioIndex(c,d), 1\le
c \le \rdim, 0\le d < 2^\ell \}$ denote the set of indices of rows of
$\krylov[\shifts]{\evMat}$ which correspond to degrees less than $2^\ell$, and
let $\mat{K}_\ell$ be the submatrix of $\krylov[\shifts]{\evMat}$ formed by the
rows with indices in $\mathcal{I}_\ell$, that is, the submatrix $\mat{K}_\ell$
of $\krylov[\shifts]{\evMat}$ which is a row permutation of the matrix
\[\left[\begin{array}{c} \evMat \\ \hline \evMat\mulmat \\ \hline \vdots \\
\hline \evMat\mulmat^{2^{\ell}-1 }\end{array}\right] \in
\matSpace[(2^{\ell}-1)\rdim][\order];\] for ease of presentation, we continue to
index the rows of $\mat{K}_\ell$ with $\mathcal{I}_\ell$. Now, suppose that at
the beginning of the iteration $\ell$ of the loop, $(i_1,\ldots,i_r)$ is the
row rank profile of $\mat{K}_\ell$, and $\mat{M}$ is the submatrix of
$\krylov[\shifts]{\evMat}$ formed by the rows with indices in
$\{i_1,\ldots,i_r\}$.  Then, we claim that at the end of this iteration,
$(k_{m_1},\ldots,k_{m_r})$ is the row rank profile of $\mat{K}_{\ell+1}$; it is
then obvious that the updated matrix $\mat{M}$, after Step~\textbf{5.g}, is the
corresponding submatrix of $\krylov[\shifts]{\evMat}$. 

First, the indices $(i_{r+1},\ldots,i_{2r})$ computed at Step~\textbf{b} are in
$\mathcal{I}_{\ell+1} - \mathcal{I}_\ell$, which is the set of indices of the
rows of $\mat{K}_\ell \mulmat^{2^\ell}$ in the matrix
$\krylov[\shifts]{\evMat}$ (or in the matrix $\mat{K}_{\ell+1}$ since we chose
to keep the same indexing). From \Cref{lem:structure}, we know that if two
indices $i = \prioIndex(c,d) < i' = \prioIndex(c',d')$ are in
$\mathcal{I}_\ell$, then we also have $\prioIndex(c,d+2^\ell) <
\prioIndex(c',d'+2^\ell)$ in $\mathcal{I}_{\ell+1} - \mathcal{I}_\ell$.  This
means that $\mat{K}_{\ell} \mulmat^{2^\ell}$ is not only formed by the rows of
$\mat{K}_{\ell+1}$ with indices in $\mathcal{I}_{\ell+1} - \mathcal{I}_\ell$:
it is actually the submatrix of $\mat{K}_{\ell+1}$ formed by these rows,
keeping the same row order.

In particular, for a given $k\in\mathcal{I}_\ell$, if the row $k$ of
$\mat{K}_\ell$ is a linear combination of the rows of this matrix with smaller
indices, then the same property holds in the matrix $\mat{K}_{\ell+1}$; and
similarly if the row $k\in \mathcal{I}_{\ell+1} - \mathcal{I}_\ell$ of
$\mat{K}_\ell \mulmat^{2^\ell}$ is a linear combination of the rows of this
matrix with smaller indices, then the same holds in $\mat{K}_{\ell+1}$.
Another consequence is that the sequence $(i_{r+1},\ldots,i_{2r})$ defined in
Step~\textbf{5.b} is strictly increasing, as stated in Step~\textbf{5.c};
besides, it does not share any common element with $(i_1,\ldots,i_r)$, so that
their merge $(k_1,\ldots,k_{2r})$ in Step~\textbf{5.c} is unique and strictly
increasing).

Now, since the row rank profile of $\mat{K}_\ell \mulmat^{2^\ell}$ is a
subsequence of the row rank profile of $\mat{K}_\ell$,
the row rank profile of the submatrix of $\mat{K}_{\ell+1}$ formed by the rows
in $\mathcal{I}_{\ell+1} - \mathcal{I}_\ell$ is a subsequence of
$(i_{r+1},\ldots,i_{2r})$. Thus, if $k$ is an index in $\mathcal{I}_{\ell+1} -
\{k_1,\ldots,k_{2r}\}$, then the row $k$ of $\mat{K}_{\ell+1}$ is a linear
combination of the rows with smaller indices, and thus $k$ will not appear in
the row rank profile of $\mat{K}_{\ell+1}$. Thus, the row rank profile of
$\mat{K}_{\ell+1}$, that is, of the submatrix of $\krylov[\shifts]{\evMat}$
formed by the rows in $\mathcal{I}_{\ell+1}$, is a subsequence of
$(k_1,\ldots,k_{2r})$. This justifies that in Steps~\textbf{5.e}
and~\textbf{5.f} one may only focus on the rows with indices in
$\{k_1,\ldots,k_{2r}\}$. The conclusion follows.
\end{proof}

\begin{figure}[!h]
\centering
\fbox{\begin{minipage}{12.5cm}
\begin{algorithm} [\algoname{KrylovRankProfile}]
\label{algo:rank-profile}

~\smallskip\\
\algoword{Input:}
\begin{itemize}
  \item matrix $\evMat \in \evSpace{\order}$,
  \item matrix $\mulmat \in \matSpace[\order]$, 
  \item shift $\shifts \in \shiftSpace$,
  \item a bound $\maxDeg$ on the degree of
    the minimal polynomial of $\evMat$, with $\maxDeg$ a power of $2$ in $\{1,\ldots,2\order-1\}$.
\end{itemize}

\smallskip
\algoword{Output:}
\begin{itemize}
  \item the row rank profile $(i_1,\ldots,i_r)$
    of $\krylov[\shifts]{\evMat}$, 
  \item the submatrix $\pivmat[\shifts]{\evMat}$ of
    $\krylov[\shifts]{\evMat}$ formed by the rows with indices in
    $\{i_1,\ldots,i_r\}$, 
  \item the column rank profile $(j_1,\ldots,j_r)$ of
    $\pivmat[\shifts]{\evMat}$.
\end{itemize}

\smallskip
\begin{enumerate}[{\bf 1.}]
  \item compute $\prioIndex$ as in \Cref{dfn:priority}
  \item $r,(c_1,\ldots,c_r) \leftarrow \algoname{RowRankProfile}(\evMat)$
  \item $(i_1,\ldots,i_r) \leftarrow (\prioIndex(c_1,0), \ldots, \prioIndex(c_r,0))$
  \item $\mat{M} \leftarrow$ submatrix of $\krylov[\shifts]{\evMat}$ with rows of indices in $\{i_1,\ldots,i_r\}$
  \item \algoword{For} $\ell$ \algoword{from} $0$ \algoword{to} $\log(\maxDeg)$,
    \begin{enumerate}[{\bf a.}]
      \item $(c_1,d_1) \leftarrow \prioIndex^{-1}(i_1), \; \ldots, \; (c_r,d_r) \leftarrow \prioIndex^{-1}(i_r)$
      \item $i_{r+1} \leftarrow \prioIndex(c_1,d_1+2^\ell), \; \ldots, \; i_{2r} \leftarrow \prioIndex(c_r,d_r+2^\ell)$
      \item $(k_1,\ldots,k_{2r}) \leftarrow$ merge the 
        increasing sequences $(i_1,\ldots,i_r)$ and $(i_{r+1},\ldots,i_{2r})$
      \item compute $\mat{M}\mulmat^{2^\ell} \in \matSpace[r][\order]$, the rows at indices $i_{r+1}, \ldots, i_{2r}$ in $\krylov[\shifts]{\evMat}$
      \item $\mat{M} \leftarrow$ the submatrix of $\krylov[\shifts]{\evMat}$
        formed by the rows in $\mat{M}$ and $\mat{M}\mulmat^{2^\ell}$; that
        is, the rows of $\krylov[\shifts]{\evMat}$ with indices in $\{k_1,
        \ldots, k_{2r}\}$
      \item $r', (m_1,\ldots,m_{r'}) \leftarrow \algoname{RowRankProfile}(\mat{M})$
      \item $\mat{M} \leftarrow$ the submatrix of $\mat{M}$ formed by rows with indices in $\{m_1,\ldots,m_{r'}\}$
      \item $r, (i_1,\ldots,i_r) \leftarrow r', (k_{m_1}, \ldots, k_{m_{r'}})$
    \end{enumerate}
  \item $(j_1,\ldots,j_r) \leftarrow \algoname{ColRankProfile}(\mat{M})$ 
  \item \algoword{Return} $(i_1,\ldots,i_r)$, $\mat{M}$, and $(j_1,\ldots,j_r)$.
\end{enumerate}
\end{algorithm}
\end{minipage}}
\end{figure}

\subsection{Computing minimal interpolation bases via linearization}
\label[subsec]{subsec:lin-mib-algo}

As noted in \Cref{rmk:elimination}, an $\shifts$-minimal interpolation
basis for $(\evMat,\mulmat)$ can be retrieved from linear relations which
express the rows of $\krylov[\shifts]{\evMat}$ of indices
$\{\prioIndex(1,\minDeg_1), \ldots, \prioIndex(\rdim,\minDeg_\rdim)\}$ as
combinations of the rows with smaller indices. Concerning the latter rows, one
can for example restrict to those given by the row rank profile
$(i_1,\ldots,i_r)$: thus, one can build an interpolation basis by considering
only $r+\rdim \le \order+\rdim$ rows in $\krylov[\shifts]{\evMat}$. In many
useful cases, $\order+\rdim$ is significantly smaller than the total number of
rows $\rdim (\maxDeg+1)$ in the matrix. 

\begin{definition}
\label{dfn:pivot-target}
Let $\mulmat \in \matSpace[\order]$ and $\evMat \in \matSpace[\rdim][\order]$,
and let $\shifts \in \shiftSpace$. Then $\pivmat[\shifts]{\evMat} \in
\matSpace[r][\order]$ is the submatrix of $\krylov[\shifts]{\evMat}$ formed by
its rows with indices in $\{i_1, \ldots, i_r\}$, and $\tgtmat[\shifts]{\evMat}
\in \matSpace[\rdim][\order]$ is the submatrix of $\krylov[\shifts]{\evMat}$
formed by the rows with indices in $\{\prioIndex(1,\minDeg_1), \ldots,
\prioIndex(\rdim,\minDeg_\rdim)\}$.
\end{definition}

\noindent
The matrix $\pivmat[\shifts]{\evMat} \in \matSpace[r][\order]$ can be thought
of as a \emph{pivot} matrix, since its rows are used as pivots to find
relations through the elimination of the rows in $\tgtmat[\shifts]{\evMat} \in
\matSpace[\rdim][\order]$, which we therefore think of as the \emph{target}
matrix. From \Cref{subsec:rel-int}, we know that these relations
correspond to an interpolation basis $\intBasis$ in $\shifts$-weak Popov form.
It turns out that restricting our view of $\krylov[\shifts]{\evMat}$ to the
submatrix $\pivmat[\shifts]{\evMat}$ leads to find such relations with a
minimal number of coefficients, which corresponds to a stronger type of
minimality: $\intBasis$ is in $\shifts$-Popov form~\citep{Kailath80,BeLaVi06}.

\begin{definition}[shifted Popov form]
\label{dfn:popov}
Let $\mat{P} \in \polMatSpace[\rdim]$ have full rank, and let $\shifts$ in
$\shiftSpace$.  Then, $\mat{P}$ is said to be in \emph{$\shifts$-Popov form}
if the $\shifts$-pivot entries are monic and on the diagonal of $\mat{P}$,
and in each column of $\mat{P}$ the nonpivot entries have degree less than
the pivot entry.
\end{definition}

\noindent
A matrix in $\shifts$-Popov form is in particular in $\shifts$-weak Popov form
and $\shifts$-reduced; besides, this is a normal form in the sense that, for a
given $\mat{A}\in \polMatSpace[\rdim]$ with full rank and a given shift
$\shifts$, there is a unique matrix $\mat{P}$ in $\shifts$-Popov form which is
unimodularly equivalent to $\mat{A}$. In particular, given $(\evMat,\mulmat)$,
for each shift $\shifts$ there is a unique ($\shifts$-minimal) interpolation
basis for $(\evMat,\mulmat)$ which is in $\shifts$-Popov form.

Since all rows in $\krylov[\shifts]{\evMat}$ are linear combinations of those
in the submatrix $\pivmat[\shifts]{\evMat}$, there is an $\rdim \times r$ matrix
$\relmat[\shifts]{\evMat}$ such that $\tgtmat[\shifts]{\evMat} =
\relmat[\shifts]{\evMat} \pivmat[\shifts]{\evMat}$, which we think of as the
\emph{relation} matrix; besides, since the pivot matrix has full rank, this
defines $\relmat[\shifts]{\evMat}$ uniquely. Then, the linear relations that we
are looking for are $[ -\relmat[\shifts]{\evMat} | \idMat ]$, and they can be
computed for example using Gaussian elimination on the rows of $\trsp{ [
\trsp{\pivmat[\shifts]{\evMat}} | \trsp{\tgtmat[\shifts]{\evMat}} ] }$.  More
precisely, $[ -\relmat[\shifts]{\evMat} | \idMat ]$ is the set of columns
with indices in $(i_1,\ldots,i_r, \prioIndex(1,\minDeg_1), \ldots,
\prioIndex(\rdim,\minDeg_\rdim))$ of the relations that we are looking for
between the rows of $\krylov[\shifts]{\evMat}$, and the interpolation basis
in $\shifts$-Popov form is the compression of these relations. Generally,
given a matrix $\smat[1]$ in $\matSpace[n][(r+\rdim)]$ for some $n$, we see it
as formed by the columns with indices $(i_1,\ldots,i_r,
\prioIndex(1,\minDeg_1), \ldots, \prioIndex(\rdim,\minDeg_\rdim))$ (in this
order) of a matrix $\smat[2]$ in $\matSpace[n][\rdim (\maxDeg+1)]$ which has
other columns zero. Then, the compression of $\smat[1]$ is the compression
$\polFromLin[\shifts]{\smat[2]}$ of $\smat[2]$ as defined in
\Cref{subsec:linearization}; we abusively denote it
$\polFromLin[\shifts]{\smat[1]}$ since there will be no ambiguity.

\begin{lemma}
\label{lem:lin-mib}
Let $\mulmat \in \matSpace[\order]$ and $\evMat \in \matSpace[\rdim][\order]$,
and let $\shifts \in \shiftSpace$.  Let $\pivmat[\shifts]{\evMat} \in
\matSpace[r][\order]$ and $\tgtmat[\shifts]{\evMat} \in
\matSpace[\rdim][\order]$ be as in \Cref{dfn:pivot-target}, and let
$\relmat[\shifts]{\evMat}$ be the unique matrix in $\matSpace[\rdim][r]$ such
that $\tgtmat[\shifts]{\evMat} = \relmat[\shifts]{\evMat}
\pivmat[\shifts]{\evMat}$. Then,
\[\intBasis = \polFromLin[\shifts]{ [
  -\relmat[\shifts]{\evMat} \,|\, \idMat ] }\]
is an interpolation basis for
$(\evMat,\mulmat)$ in $\shifts$-Popov form.

Besides, if $(j_1,\ldots,j_r)$ denotes the column rank profile of
$\pivmat[\shifts]{\evMat}$, and $\mat{C} \in \matSpace[r]$ and $\mat{D} \in
\matSpace[\rdim][r]$ are the submatrices of $\pivmat[\shifts]{\evMat}$ and
$\tgtmat[\shifts]{\evMat}$, respectively, formed by the columns with indices in
$\{j_1,\ldots,j_r\}$, then we have \[\relmat[\shifts]{\evMat} = \mat{D}
\mat{C}^{-1}.\]
\end{lemma}
\begin{proof}
First, restricting the identity $\tgtmat[\shifts]{\evMat} =
\relmat[\shifts]{\evMat} \pivmat[\shifts]{\evMat}$ to the submatrices with
column indices in $\{j_1,\ldots,j_r\}$ we have in particular $\mat{D} =
\relmat[\shifts]{\evMat} \mat{C}$. By construction, $\mat{C}$ is invertible
and thus $\relmat[\shifts]{\evMat} = \mat{D} \mat{C}^{-1}$.

Let $\mat{R} \in \matSpace[\rdim][\rdim(\maxDeg+1)]$ be the matrix whose
columns at indices $i_1$,$\ldots$,$i_r$, $\prioIndex(1,\minDeg_1)$, $\ldots$,
$\prioIndex(\rdim,\minDeg_\rdim)$ are the columns $1,\ldots,r+\rdim$ of $[
  -\relmat[\shifts]{\evMat} | \idMat ]$, respectively, and other columns are
  zero; let also $\intBasis = \polFromLin[\shifts]{[ -\relmat[\shifts]{\evMat}
| \idMat ] } = \polFromLin[\shifts]{\mat{R}}$. By construction, every row $c$
of $\intBasis$ is the compression $\polFromLin[\shifts]{\matrow{\mat{R}}{c}}$
of a linear relation between the rows of $\krylov[\shifts]{\evMat}$ and is thus
an interpolant for $(\evMat,\mulmat)$. We will further prove that $\intBasis$
is in $\shifts$-Popov form with $\shifts$-pivot degree
$(\minDeg_1,\ldots,\minDeg_\rdim)$; in particular, this implies that
$\intBasis$ is $\shifts$-reduced and has $\shifts$-row degree
$(\shift{1}+\minDeg_1,\ldots,\shift{\rdim}+\minDeg_\rdim)$, so that
\Cref{lem:red-to-mib} shows that $\intBasis$ is an interpolation basis for
$(\evMat,\mulmat)$. For $k \in \{1,\ldots,r\}$, we write $i_k =
\prioIndex(c_k,d_k)$ for some unique $(c_k, d_k)$. We fix
$c\in\{1,\ldots,\rdim\}$.

First, we consider the column $c$ of $\intBasis$ and we show that all its
entries have degree less than $\minDeg_c$ except the entry on the diagonal,
which is monic and has degree exactly $\minDeg_c$. Indeed, for any $k$ such
that $c_k=c$, by definition of $\polFromLin[\shifts]{\mat{R}}$ the column $i_k$
of $\mat{R}$ is compressed into the coefficient of degree $d_k$ in the column
$c$ of $\intBasis$, and by \Cref{lem:rank-profile} we know that $d_k <
\minDeg_c$. Besides, the column of $\mat{R}$ at index
$\prioIndex(c,\minDeg_c)$, which has all its entries $0$ except the entry on
row $c$ which is $1$, only brings a coefficient $1$ of degree $\minDeg_c$ in
the diagonal entry $(c,c)$ of $\intBasis$.

Second, we consider the row $c$ of $\intBasis$ and we show that it has
$\shifts$-pivot index $c$ and $\shifts$-pivot degree $\minDeg_c$. Thanks to
\Cref{lem:pivot}, it is enough to show that the rightmost nonzero entry in
the row $c$ of $\mat{R}$ is the entry $1$ at column index
$\prioIndex(c,\minDeg_c)$. All entries in the row $c$ of $\mat{R}$ with indices
greater than $\prioIndex(c,\minDeg_c)$ and not in $\{i_1,\ldots,i_r\}$ are
obviously zero. Now, by definition of $\minDeg_c$, we know that the row of
$\krylov[\shifts]{\evMat}$ at index $\prioIndex(c,\minDeg_c)$ is a linear
combination of the rows at indices smaller than $\prioIndex(c,\minDeg_c)$; in
particular, because the rows of $\krylov[\shifts]{\evMat}$ at indices
$i_1,\ldots,i_r$ are linearly independent, the linear combination given by the
row $c$ of $\mat{R}$ has entries $0$ on the columns at indices $i_k$ for $k$
such that $i_k > \prioIndex(c,\minDeg_c)$.
\end{proof}

Now, we turn to the fast computation of an interpolation basis $\intBasis$ for
$(\evMat,\mulmat)$ in $\shifts$-Popov form. In view of what precedes, this
boils down to two steps, detailed in \Cref{algo:lin-mib}: first, we
compute the row rank profile $(i_1,\ldots,i_r)$ of $\krylov[\shifts]{\evMat}$
from which we also deduce the $\shifts$-minimal degree
$(\minDeg_1,\ldots,\minDeg_\rdim)$, and second, we compute the linear relations
$\mat{D} \mat{C}^{-1}$. We now prove \Cref{prop:lin-mib}, by
showing that \Cref{algo:lin-mib} is correct and uses $\bigO{
\order^\expmatmul ( \lceil \rdim / \order \rceil + \log(\maxDeg) ) }$
operations in $\field$ if $\expmatmul>2$, and $\bigO{
\order^\expmatmul ( \lceil \rdim / \order \rceil + \log(\maxDeg) ) \log(\order) }$
operations if $\expmatmul=2$.
\begin{proof}[Proof of \Cref{prop:lin-mib}]
The correctness follows from \Cref{lem:rank-profile,lem:lin-mib} and from the
correctness of the algorithm \algoname{KrylovRankProfile}. Since $r \le
\order$, the computation of $\mat{C}^{-1}$ at Step~\textbf{7} uses
$\bigO{r^\expmatmul} \subset \bigO{\order^\expmatmul}$ operations, and the
computation of $\mat{D}\mat{C}^{-1}$ uses $\bigO{\order^\expmatmul}$ operations
when $\rdim \le \order$, and $\bigO{\rdim \order^{\expmatmul-1}}$ operations
when $\order \le \rdim$. Then, the announced cost bound follows from
\Cref{prop:algo-rank-profile}.
\end{proof}

\begin{figure}[h!]
\centering
\fbox{\begin{minipage}{12.5cm}
\begin{algorithm} [\algoname{LinearizationInterpolationBasis}]
\label{algo:lin-mib}

~\smallskip\\
\algoword{Input:} 
\begin{itemize}
  \item matrix $\evMat \in \evSpace{\order}$,
  \item matrix $\mulmat \in \matSpace[\order]$, 
  \item shift $\shifts \in \shiftSpace$,
  \item a bound $\maxDeg$ on the degree of the minimal polynomial
of $\mulmat$, with $\maxDeg$ a power of $2$ in $\{1,\ldots,2\order-1\}$.
\end{itemize}

\smallskip
\algoword{Output:} 
the interpolation basis $\intBasis \in \intSpace_{\le\maxDeg}$
for $(\evMat,\mulmat)$ in $\shifts$-Popov form.

\smallskip
\begin{enumerate}[{\bf 1.}]
  \item compute $\prioEval$ and $\prioIndex$ as in \Cref{dfn:priority}
  \item $(i_1,\ldots,i_r),\pivmat[\shifts]{\evMat},(j_1,\ldots,j_r) \leftarrow \algoname{KrylovRankProfile}(\evMat,\mulmat,\shifts,\maxDeg)$
  \item \algoword{For} $1\le k \le r$, compute $(c_k,d_k) \leftarrow \prioIndex^{-1}(i_k)$
  \item \algoword{For} $1\le c \le \rdim$, compute $\minDeg_c \leftarrow 1 + \max \{d_k \mid 1 \le k \le r \text{ and } c_k=c\}$
    if the set is nonempty, and $\minDeg_c \leftarrow 0$ if it is empty
  \item $\tgtmat[\shifts]{\evMat} \leftarrow$ submatrix 
    of $\krylov[\shifts]{\evMat}$ formed by the rows with indices in
    $\{\prioIndex(1,\minDeg_1), \ldots, \prioIndex(\rdim,\minDeg_\rdim)\}$
  \item $\mat{C},\mat{D} \leftarrow$ submatrices of $\pivmat[\shifts]{\evMat}$
    and $\tgtmat[\shifts]{\evMat}$, respectively, formed by the columns with
    indices in $\{j_1,\ldots,j_r\}$
  \item compute $\relmat[\shifts]{\evMat} \leftarrow \mat{D} \mat{C}^{-1}$
  \item $\intBasis \leftarrow \polFromLin[\shifts]{ [ -\relmat[\shifts]{\evMat} \,|\, \idMat[\rdim] ] }$ 
  \item \algoword{Return} $\intBasis$
\end{enumerate}
\end{algorithm}
\end{minipage}}
\end{figure}

\begin{ack}
We thank B. Beckermann, E. Kaltofen, G. Labahn, and E.  Tsigaridas for useful
discussions. We also thank the reviewers for their thorough reading and helpful
comments. C.-P.  Jeannerod and G. Villard were partly supported by the ANR HPAC
project (ANR 11 BS02 013). V. Neiger was supported by the international
mobility grants from \emph{Projet Avenir Lyon Saint-\'Etienne}, \emph{Mitacs
Globalink - Inria}, and Explo'ra Doc from \emph{R\'egion Rh\^one-Alpes}. \'E.
Schost was supported by NSERC and by the Canada Research Chairs program.
\end{ack}

\appendix

\section{Bounds for polynomial matrix multiplication functions}
\label{app:cost-mult}

In this appendix, we give upper bounds for the quantities in
\Cref{dfn:polmatmul}.

\begin{lemma}
  \label{lem:polmatmul-bound}
  We have the upper bounds
  \begin{align*}
    \polmatmultimePrime{\rdim,d}  \;\in\; \bigO{ & \rdim^{\expmatmul-1}
    \polmultime{\rdim d} } & \text{if } \expmatmul>2, \\
    \polmatmultimePrime{\rdim,d} \;\in\; \bigO{ & \rdim \polmultime{\rdim d}
    \log(\rdim) } & \text{if } \expmatmul=2, \\
    \polmatmultimeBis{\rdim,d} \;\in\; \bigO{ & \rdim^\expmatmul \polmultime{d}
  \log(d)}.
  \end{align*}
\end{lemma}
\begin{proof}
  It is enough to show these bounds for $\rdim$ and $d$ powers of $2$. The
  bound on $\polmatmultimeBis{\rdim,d}$ follows from the super-linearity
  property $2^j \polmultime{2^{-j} d} \le \polmultime{d}$. 

  Using the super-linearity property $\polmultime{ 2^{-i} \rdim d } \le 2^{-i}
  \polmultime{ \rdim d }$, we obtain $\polmatmultimePrime{\rdim,d} = \sum_{0\le
  i \le \log(\rdim)} 2^{-i} \rdim \polmatmultime{2^i, 2^{-i} \rdim d} \in
  \bigO{ \sum_{0\le i \le \log(\rdim)} 2^{i(\expmatmul-2)} \rdim
  \polmultime{\rdim d} }$. This concludes the proof since we have $\sum_{0\le i
  \le \log(\rdim)} 2^{i(\expmatmul-2)} \in \Theta(\rdim^{\expmatmul-2} +
  \log(\rdim) )$, where the logarithmic term accounts for the possibility that
  $\expmatmul=2$.
\end{proof}

\begin{lemma}
  \label{lem:polmatmulDnc-bound}
  We have the upper bounds
  \begin{align*}
    \polmatmultimePrimeDnc{\rdim,d} \;\in\; \bigO{ & \rdim^{\expmatmul-1}
    \polmultime{\rdim d} } & \text{if } \expmatmul>2, \\
    \polmatmultimePrimeDnc{\rdim,d} \;\in\; \bigO{ & \rdim \polmultime{\rdim d}
    \log(\rdim)^2 } & \text{if } \expmatmul=2, \\
    \polmatmultimeBisDnc{\rdim,d} \;\in\; \bigO{ & \rdim^{\expmatmul-1}
    \polmultime{\rdim d} + \rdim^{\expmatmul} \polmultime{d} \log(d) } &
    \text{if } \expmatmul>2, \\
    \polmatmultimeBisDnc{\rdim,d} \;\in\; \bigO{ & \rdim \polmultime{\rdim d}
    \log(\rdim)^2 + \rdim^2 \polmultime{d} \log(d) \log(\rdim) } & \text{if }
    \expmatmul=2.
  \end{align*}
\end{lemma}
\begin{proof}
  It is enough to show these bounds for $\rdim$ and $d$ powers of $2$.
  The first two bounds are obtained from \Cref{lem:polmatmul-bound}, which
  implies that
  \[
    \sum_{i=0}^{\log(\rdim)} 2^i \polmatmultimePrime{2^{-i}\rdim,2^i d} \in
    \bigOPar{ \left( \sum_{i=0}^{\log(\rdim)} 2^{i(2-\expmatmul)} \rdim^{\expmatmul-1}
    \polmultime{\rdim d} \right)  + \rdim \polmultime{\rdim d} \log(\rdim)^2 },
  \]
  and from the fact that $\sum_{0\le i\le \log(\rdim)} 2^{i(2-\expmatmul)}$
  is upper bounded by a constant if $\expmatmul>2$.
    
  Now, we focus on the last two bounds. By definition,
  \[ \polmatmultimeBisDnc{ \rdim, d } \in \bigOPar{ \sum_{i=0}^{\log(\rdim)}
    2^i \sum_{j=0}^{\log(2^i d)} 2^j (2^{-i}\rdim)^\expmatmul
  \polmultime{2^{i-j} d} }.\] In the inner sum, $j$ goes from $0$ to
  $\log(2^i d) = i + \log(d)$: we will separately study the first terms
  with $j \le i$ and the remaining terms with $j > i$.

  First, using the super-linearity property $\polmultime{2^{i-j} d} \le 2^{i-j}
  \polmultime{\rdim d} / \rdim$, we obtain
  \begin{align*}
  \sum_{i=0}^{\log(\rdim)} 2^i \sum_{j=0}^{i} 2^j (2^{-i}\rdim)^\expmatmul
  \polmultime{2^{i-j} d} & \in \bigOPar{ \sum_{i=0}^{\log(\rdim)} (i+1)
  2^{i(2-\expmatmul)} \rdim^{\expmatmul-1} \polmultime{\rdim d} } \\
   & \in \bigO{ \rdim^{\expmatmul-1}
   \polmultime{\rdim d} + \rdim \polmultime{\rdim d} \log(\rdim)^2 },
  \end{align*}
  since when $\expmatmul>2$, the sum $\sum_{0 \le i \le \log(\rdim)} (i+1)
  2^{i(2-\expmatmul)}$ is known to be less than its limit
  $(1-2^{2-\expmatmul})^{-2}$ when $\rdim \to \infty$. We note that the second
  term $\rdim \polmultime{\rdim d} \log(\rdim)^2$ accounts for the
  possibility that $\expmatmul=2$.

  Then, using the super-linearity property $\polmultime{2^{i-j} d} \le 2^{i-j}
  \polmultime{d}$ when $j > i$, we obtain
  \begin{align*}
    \sum_{i=0}^{\log(\rdim)} 2^i \sum_{j=i+1}^{i+\log(d)} 2^j
    (2^{-i}\rdim)^\expmatmul \polmultime{2^{i-j} d} & \le
    \sum_{i=0}^{\log(\rdim)} 2^{i(2-\expmatmul)} \rdim^\expmatmul
    \polmultime{d}
    \log(d) \\
      & \in \bigO{ \rdim^\expmatmul \polmultime{d} \log(d) + \rdim^2
      \polmultime{d} \log(d)\log(\rdim) },
  \end{align*}
  which concludes the proof.
\end{proof}

\begin{lemma}
  \label{lem:polmatmulDoubleDnc-bound}
  Let $\bar{d}$ denote the power of $2$ such that $d \le \bar{d} < 2d$; we have
  the upper bounds
  \begin{align*}
    \textstyle{\sum_{0\le i\le \log(\bar{d})}}\; 2^i \polmatmultimePrimeDnc{\rdim,2^{-i}\bar{d}} \;\in\; \bigO{ & \rdim^{\expmatmul-1}
    \polmultime{\rdim d} + \rdim^\expmatmul \polmultime{d} \log(d) } & \text{if } \expmatmul>2, \\
    \textstyle{\sum_{0\le i\le \log(\bar{d})}}\; 2^i \polmatmultimePrimeDnc{\rdim,2^{-i}\bar{d}} \;\in\; \bigO{ & \rdim \polmultime{\rdim d}
    \log(\rdim)^3 + \rdim^2 \polmultime{d}\log(d)) \log(\rdim)^2 } & \text{if } \expmatmul=2, \\
    \textstyle{\sum_{0\le i\le \log(\bar{d})}}\; 2^i\polmatmultimeBisDnc{\rdim,2^{-i}\bar{d}} \;\in\; \bigO{ & \rdim^{\expmatmul-1}
    \polmultime{\rdim d} + \rdim^{\expmatmul} \polmultime{d} \log(d)^2 } &
    \text{if } \expmatmul>2, \\
    \textstyle{\sum_{0\le i\le \log(\bar{d})}}\; 2^i\polmatmultimeBisDnc{\rdim,2^{-i}\bar{d}} \;\in\; \bigO{ & \rdim \polmultime{\rdim d}
    \log(\rdim)^3 + \rdim^2 \polmultime{d} \log(d)^2 \log(\rdim) } & \text{if }
    \expmatmul=2.
  \end{align*}
\end{lemma}
\begin{proof}
  It is enough to show these bounds for $\rdim$ and $d$ powers of $2$; in
  particular, $d = \bar{d}$.

  Let us study the first two bounds. By definition,
  \begin{align*}
    \sum_{i=0}^{\log(d)} 2^i \polmatmultimePrimeDnc{\rdim,2^{-i}d} 
    & = \sum_{i=0}^{\log(d)} 2^i \sum_{j=0}^{\log(\rdim)} 2^j \sum_{k=0}^{\log(\rdim)-j}
    2^{k} \polmatmultime{2^{-j-k}\rdim,2^{j+k-i}d} \\ 
    & = \sum_{j=0}^{\log(\rdim)} \sum_{k=0}^{\log(\rdim)-j} \sum_{i=0}^{\log(d)} 
    2^{i+j+k} \polmatmultime{2^{-j-k}\rdim,2^{j+k-i}d} .
  \end{align*}
  Considering the terms with $i\le j+k$, we use $\polmultime{2^{j+k-i}d} \le
  2^{j+k-i} \polmultime{\rdim d} / \rdim$ to obtain
  \begin{align*}
    \sum_{j=0}^{\log(\rdim)} \sum_{k=0}^{\log(\rdim)-j} \sum_{i=0}^{j+k}
    & 2^{i+j+k} \polmatmultime{2^{-j-k}\rdim,2^{j+k-i}d} \\
    & \in \bigOPar{
    \sum_{j=0}^{\log(\rdim)} \sum_{k=0}^{\log(\rdim)-j} (j+k+1)
    2^{(j+k)(2-\expmatmul)} \rdim^{\expmatmul-1} \polmultime{\rdim d} },
  \end{align*}
  from which we conclude since the sum $\sum_{0\le j\le \log(\rdim)} \sum_{0
  \le k \le \log(\rdim)-j} (j+k+1) 2^{(j+k)(2-\expmatmul)}$ is $\bigO{1}$ if
  $\expmatmul>2$, and $\bigO{\log(\rdim)^3}$ if $\expmatmul=2$.

  Now, considering the terms with $i> j+k$, we use $\polmultime{2^{j+k-i}d} \le
  2^{j+k-i} \polmultime{d}$ to obtain
  \begin{align*}
    \sum_{j=0}^{\log(\rdim)} \sum_{k=0}^{\log(\rdim)-j} \sum_{i=j+k+1}^{\log(d)}
    & 2^{i+j+k} \polmatmultime{2^{-j-k}\rdim,2^{j+k-i}d} \\
    & \in \bigOPar{
    \sum_{j=0}^{\log(\rdim)} \sum_{k=0}^{\log(\rdim)-j}
    2^{(j+k)(2-\expmatmul)} \rdim^{\expmatmul} \polmultime{d} \log(d) },
  \end{align*}
  where again the sum on $j$ and $k$ is $\bigO{1}$ if $\expmatmul>2$, and
  $\bigO{\log(\rdim)^2}$ if $\expmatmul=2$.

  Then, we study the last two bounds. By definition,
  \begin{align*}
    \sum_{i=0}^{\log(d)} 2^i \polmatmultimeBisDnc{\rdim,2^{-i}d} 
    & = \sum_{i=0}^{\log(d)} 2^i \sum_{j=0}^{\log(\rdim)} 2^j \sum_{k=0}^{\log(d)+j-i}
    2^k \polmatmultime{2^{-j}\rdim,2^{j-i-k}d} \\
    & = \sum_{j=0}^{\log(\rdim)} \sum_{i=0}^{\log(d)} \sum_{k=0}^{\log(d)+j-i}
    2^{i+j+k} \polmatmultime{2^{-j}\rdim,2^{j-i-k}d}.
  \end{align*}
  Considering the terms with $k > j-i$, we use $\polmultime{2^{j-i-k} d} \le 2^{j-i-k} \polmultime{d}$
  to obtain
   \[ \sum_{j=0}^{\log(\rdim)} \sum_{i=0}^{\log(d)} \sum_{k=\max(0,1+j-i)}^{\log(d)+j-i}
    2^{i+j+k} \polmatmultime{2^{-j}\rdim,2^{j-i-k}d} \in
  \bigOPar{ \sum_{j=0}^{\log(\rdim)} 2^{j(2-\expmatmul)} \rdim^\expmatmul \polmultime{d} \log(d)^2 }; \]
  this is $\bigO{\rdim^\expmatmul \polmultime{d} \log(d)^2}$ if $\expmatmul>2$
  and $\bigO{\rdim \polmultime{d} \log(d)^2 \log(\rdim)}$ if $\expmatmul=2$.

  Now, considering the terms with $0 \le k\le j-i$, and thus also $i \le j$, we use
  $\polmultime{2^{j-i-k} d} \le 2^{j-i-k} \polmultime{\rdim d} / \rdim$ to obtain
  \[
    \sum_{j=0}^{\log(\rdim)} \sum_{i=0}^{j} \sum_{k=0}^{j-i}
    2^{i+j+k} \polmatmultime{2^{-j}\rdim,2^{j-i-k}d}
    \in \bigOPar{ \sum_{j=0}^{\log(\rdim)} (j+1)^2 2^{j(2-\expmatmul)}
    \rdim^{\expmatmul-1} \polmultime{\rdim d} }.
  \]
  This gives the conclusion, since $\sum_{j=0}^{\log(\rdim)} (j+1)^2
  2^{j(2-\expmatmul)}$ is $\bigO{1}$ if $\expmatmul>2$ and
  $\bigO{\log(\rdim)^3}$ if $\expmatmul=2$.
\end{proof}

\section{Cost analysis for the computation of minimal nullspace bases}
\label{app:cost-mnb}

Here, we give a detailed cost analysis for the minimal nullspace basis
algorithm of \citet{ZhLaSt12}, which we rewrite in \Cref{algo:mnb}
using our convention here that basis vectors are rows of the basis matrix
(whereas in the above reference they are its columns).  Furthermore, we assume
that the input matrix has full rank, which allows us to better control the
dimensions of the matrices encountered in the computations: in the recursive
calls, we always have input matrices with more rows than columns.

Here, the quantity $\polmatmultimeBis{\rdim,d} = \sum_{0\le j\le \log(d)} 2^j
\polmatmultime{\rdim,2^{-j}d}$ arises in the cost analysis of fast algorithms
for the computation of Hermite-Pad\'e approximants~\citep{BecLab94,GiJeVi03},
which use a divide-and-conquer approach on the degree $d$. The minimal
nullspace basis algorithm in~\citep{ZhLaSt12} follows a divide-and-conquer
approach on the dimension of the input matrix, and computes at each node of the
recursion some products of matrices with unbalanced row degrees as well as a
minimal basis of Hermite-Pad\'e approximants. In particular, its cost will be
expressed using the quantities $\polmatmultimePrimeDnc{\rdim,d}$ and
$\polmatmultimeBisDnc{\rdim,d}$ introduced in \Cref{dfn:polmatmul}.

The following result refines the cost analysis in~\citep[Theorem~4.1]{ZhLaSt12},
counting the logarithmic factors. 

\begin{figure}[h!]
\centering
\fbox{\begin{minipage}{12.5cm}
  \begin{algorithm} [\algoname{MinimalNullspaceBasis} \citep{ZhLaSt12}]
  \label{algo:mnb}

~\smallskip\\
\algoword{Input:} 
\begin{itemize}
  \item matrix $\mat{F}\in\polMatSpace[\rdim][\cdim]$ with full rank and $\rdim
    \ge \cdim$, 
  \item a shift $\shifts \in \shiftSpace$ with entries in non-decreasing order
    and bounding the row degree of $\mat{F}$ componentwise.
\end{itemize}

\smallskip
\algoword{Output:} 
\begin{itemize}
  \item an $\shifts$-minimal nullspace basis $\mat{N}$ of $\mat{F}$,
  \item the $\shifts$-row degree of $\mat{N}$.
\end{itemize}

\smallskip
\begin{enumerate}[{\bf 1.}]
  \item $\rho \leftarrow \sum_{i=\rdim-\cdim+1}^{\rdim} \shift{i}$ and $\lambda \leftarrow \lceil \rho / \cdim \rceil$
  \item $\mat{P} \leftarrow$ a solution to \Cref{pbm:h-pade} on input
    $((3\lambda,\ldots,3\lambda),\mat{F},\shifts)$, obtained using the
    algorithm \algoname{PM-Basis}~\citep{GiJeVi03}, and with the rows of
    $\mat{P}$ arranged so that $\rdeg[\shifts]{\mat{P}}$ is non-decreasing
  \item Write $\mat{P} = \trsp{[\trsp{\mat{P}_1} | \trsp{\mat{P}_2}]}$ where $\mat{P}_1$ consists of
    all rows $\row{p}$ of $\mat{P}$ satisfying $\row{p} \mat{F} = 0$
  \item \algoword{If} $\cdim = 1$, \algoword{Return} $(\mat{P}_1,\rdeg[\shifts]{\mat{P}_1})$
  \item \algoword{Else}
    \begin{enumerate}[{\bf a.}]
      \item $\shifts[t] \leftarrow \rdeg[\shifts]{\mat{P}_2} - (3\lambda,\ldots,3\lambda)$
      \item $\mat{G} \leftarrow X^{-3\lambda} \mat{P}_2 \mat{F}$
      \item Write $\mat{G} = [ \mat{G}_1 | \mat{G}_2 ]$ where $\mat{G}_1$ has
        $\lfloor \cdim/2 \rfloor$ columns and $\mat{G}_2$ has $\lceil \cdim/2
        \rceil$ columns
      \item $(\mat{N}_1,\shifts[u]) \leftarrow \algoname{MinimalNullspaceBasis}(\mat{G}_1,\shifts[t])$
      \item $(\mat{N}_2,\shifts[v]) \leftarrow \algoname{MinimalNullspaceBasis}(\mat{N}_1 \mat{G}_2,\shifts[u])$
      \item $\mat{N} \leftarrow \trsp{[ \trsp{\mat{P}_1} | \trsp{(\mat{N}_2 \mat{N}_1 \mat{P}_2)} ]}$
      \item \algoword{Return} $(\mat{N},(\rdeg[\shifts]{\mat{P}_1},\shifts[v]))$
    \end{enumerate}
\end{enumerate}
\end{algorithm}

\end{minipage}}
\end{figure}

\begin{proposition}
  \label{prop:cost-mnb}
  Let $\mat{F}$ in $\polMatSpace[\rdim][\cdim]$ have full rank with $\rdim \ge
  \cdim$, and let $\shifts$ in $\shiftSpace$ which bounds the row degree of
  $\mat{F}$ componentwise. Let $\xi \ge \rdim$ be an integer such that
  $\sshifts \le \xi$. Assuming that $\rdim \in \bigO{\cdim}$, \Cref{algo:mnb}
  computes an $\shifts$-minimal nullspace basis of $\mat{F}$ using
  \begin{align*}
    \bigO{ & \polmatmultimePrimeDnc{\rdim,\xi/\rdim} +
    \polmatmultimeBisDnc{\rdim,\xi/\rdim} }   \\
    & \subseteq\; \bigO{ \rdim^{\expmatmul-1} \polmultime{\xi} +
    \rdim^{\expmatmul} \polmultime{\xi/\rdim} \log(\xi/\rdim) } & \text{if } \expmatmul>2, \\
    & \subseteq\; \bigO{ \rdim \polmultime{\xi} \log(\rdim)^2 + \rdim^2
    \polmultime{\xi/\rdim} \log(\xi/\rdim) \log(\rdim) } & \text{if } \expmatmul=2
  \end{align*}
  operations in $\field$.
\end{proposition}
\begin{proof}
  The proof of correctness can be found in~\citep{ZhLaSt12}. We prove the cost
  bound following \Cref{algo:mnb} step by step.

  Step~\textbf{1}: since $\rho \le \sshifts \le \xi$, we have $\lambda \le \lceil
  \xi/\cdim \rceil$.

  Step~\textbf{2}: using the algorithm \texttt{PM-Basis} in~\citep{GiJeVi03},
  $\mat{P}$ can be computed using $\bigO{ \polmatmultimeBis{\rdim, \lambda} }$
  operations in $\field$; see \citep[Theorem~2.4]{GiJeVi03}. Since $\lambda \le
  \lceil \xi/\cdim \rceil$ and $\rdim \in \bigO{\cdim}$, this step uses
  $\bigO{\polmatmultimeBis{ \cdim, \xi/\cdim }}$ operations. Besides, from
  \Cref{rmk:sum_rdeg} on the sum of the $\shifts$-row degrees of an
  $\shifts$-minimal interpolation basis, we have
  $\sshifts[{\rdeg[\shifts]{\mat{P}}}] \le 3\cdim \lambda + \xi \le 3(\rho +
  \rdim) + \xi \le 7\xi$.

  Step~\textbf{3}: finding $\mat{P}_1$ and $\mat{P}_2$ can be done by computing
  $\mat{P} \mat{F}$. The matrix $\mat{F}$ is $\rdim \times \cdim$ with row
  degree $\shifts[w] = \rdeg{\mat{F}} \le \shifts$ (componentwise); in
  particular, $\sshifts[{\shifts[w]}] \le \xi$. Besides, $\mat{P}$ is an $\rdim
  \times \rdim$ matrix and $\sshifts[{\rdeg[{\shifts[w]}]{\mat{P}}}] \le
  \sshifts[{\rdeg[\shifts]{\mat{P}}}] \le 7\xi$. Then, one can augment
  $\mat{F}$ with $\rdim - \cdim$ zero columns and use
  \Cref{algo:unbalanced-polmatmul} to compute $\mat{P} \mat{F}$;
  according to \Cref{prop:unbalanced-polmatmul}, this uses $\bigO{
    \polmatmultimePrime{\rdim,\xi/\rdim} } \subseteq \bigO{
      \polmatmultimePrime{\cdim,\xi/\cdim} }$ operations.

  Steps~\textbf{5.a} and \textbf{5.b}: Computing $\mat{G}$ involves no
  arithmetic operation since the product $\mat{P} \mat{F}$ has already been
  computed in Step~\textbf{3}; $\mat{G}$ has row degree bounded by $\shifts[t]$
  (componentwise). Let us denote $\hat{\rdim}$ the number of rows of
  $\mat{P}_2$.  Because both $\mat{P}$ and $\mat{F}$ have full rank and
  $\mat{P}_1 \mat{F} = \mat{0}$, $\mat{G}$ has full rank and at least $\cdim$
  rows in $\mat{P}$ are not in the nullspace of $\mat{F}$, which means $\cdim
  \le \hat{\rdim}$. Furthermore, according to~\citep[Theorem~3.6]{ZhLaSt12}, we
  have $\hat{\rdim} \le 3\cdim/2$. Then, $\mat{G}$ is an $\hat{\rdim} \times
  \cdim$ matrix with $\cdim \le \hat{\rdim} \le 3\cdim/2$ and with row degree
  bounded by $\shifts[t]$. In addition, we have $\shifts[t] \le \shifts$
  \citep[Lemma~3.12]{ZhLaSt12}, and thus in particular $\sshifts[{\shifts[t]}]
  \le \xi$.

  Step~\textbf{5.c}: for the recursive calls of Steps~\textbf{5.d}
  and~\textbf{5.e}, we will need to check that our assumptions on the
  dimensions, the degrees, and the rank of the input are maintained. Here, we
  first remark that $\mat{G}_1$ and $\mat{G}_2$ have full rank and respective
  dimensions $\hat{\rdim} \times \lfloor \cdim/2 \rfloor$ and $\hat{\rdim}
  \times \lceil \cdim/2 \rceil$, with $\hat{\rdim} \ge \lceil \cdim/2 \rceil
  \ge \lfloor \cdim/2 \rfloor$. Their row degrees are bounded by $\shifts[t]$,
  which is in non-decreasing order and satisfies $\sshifts[{\shifts[t]}] \le
  \xi$.

  Step~\textbf{5.d}: $\mat{N}_1$ is a $\shifts[t]$-minimal nullspace basis
  of~$\mat{G}_1$ and therefore it has $\hat{\rdim} - \lfloor \cdim/2 \rfloor$
  rows and $\hat{\rdim}$ columns. Besides, $\shifts[u] =
  \rdeg[{\shifts[t]}]{\mat{N}_1}$ and by~\citep[Theorem~3.4]{ZhLaSt12}, we have
  $\sshifts[{\shifts[u]}] \le \sshifts[{\shifts[t]}] \le \xi$.

  Step~\textbf{5.e}: we remark that $\mat{N}_1 \mat{G}_2$ has $\lceil \cdim/2
  \rceil$ columns and $\hat{\rdim} - \lfloor \cdim/2 \rfloor \ge \lceil \cdim/2
  \rceil$ rows. We now show that it has full rank. Let us consider
  $\hat{\mat{N}}_2$ any $\shifts[u]$-minimal nullspace basis of $\mat{N}_1
  \mat{G}_2$.  Then $\hat{\mat{N}}_2$ has $\hat{\rdim} - \lfloor \cdim/2
  \rfloor - r$ rows, where $r$ is the rank of $\mat{N}_1 \mat{G}_2$. Our goal
  is to prove that $r = \lceil \cdim/2 \rceil$. The matrix $\hat{\mat{N}} =
  \trsp{[ \trsp{\mat{P}_1} | \trsp{(\hat{\mat{N}}_2 \mat{N}_1 \mat{P}_2)} ]}$
  is an $\shifts$-minimal nullspace basis of $\mat{F}$ \citep[Theorems 3.9 and
  3.15]{ZhLaSt12}. In particular, since $\mat{F}$ has full rank,
  $\hat{\mat{N}}$ has $\rdim - \cdim$ rows.  Since $\mat{P}_1$ has $\rdim -
  \hat{\rdim}$ rows, this gives $\rdim-\cdim = \rdim - \hat{\rdim} +
  \hat{\rdim} - \lfloor \cdim/2 \rfloor - r = \rdim - \lfloor \cdim/2 \rfloor -
  r$. Thus $\cdim = \lfloor \cdim/2 \rfloor + r$, and $r = \lceil \cdim/2
  \rceil$.

  Furthermore, $\mat{G}_2$ has row degree bounded by $\shifts[t]$ and
  $\mat{N}_1$ has $\shifts[t]$-row degree exactly $\shifts[u]$, so that
  $\rdeg{\mat{N}_1 \mat{G}_2} \le \rdeg[\rdeg{\mat{G}_2}]{\mat{N}_1} \le
  \rdeg[{\shifts[t]}]{\mat{N}_1} = \shifts[u]$. We have $\sshifts[{\shifts[t]}]
  \le \xi$ and $\sshifts[{\shifts[u]}] \le \xi$. Augmenting $\mat{N}_1$ and
  $\mat{G}_2$ so that they are $\hat{\rdim} \times \hat{\rdim}$, by
  \Cref{prop:unbalanced-polmatmul}, $\mat{N}_1 \mat{G}_2$ can be computed using
  $\bigO{ \polmatmultimePrime{\hat{\rdim},\xi/\hat{\rdim}} } \subseteq \bigO{
    \polmatmultimePrime{\cdim,\xi/\cdim} }$ operations. Then, $\mat{N}_2$ is a
    $\shifts[t]$-minimal nullspace basis of~$\mat{N}_1 \mat{G}_2$; it has
    $\hat{\rdim} - \cdim$ rows and $\hat{\rdim} - \lceil \cdim/2 \rceil$
    columns, its $\shifts[u]$-row degree is $\shifts[v] =
    \rdeg[{\shifts[u]}]{\mat{N}_2}$, and we have $\sshifts[{\shifts[v]}] \le
    \sshifts[{\shifts[u]}] \le \xi$~\citep[Theorem~3.4]{ZhLaSt12}.

  Step~\textbf{5.f}: using the previously given dimensions and degree bounds
  for $\mat{N}_1$ and $\mat{N}_2$, one can easily check that the product
  $\mat{N}_2 \mat{N}_1$ can be computed by \Cref{algo:unbalanced-polmatmul}
  using $\bigO{ \polmatmultimePrime{ \hat{\rdim}, \xi/\hat{\rdim} } } \subseteq
  \bigO{ \polmatmultimePrime{ \cdim, \xi/\cdim } }$ operations. Now,
  $\mat{P}_2$ is $\hat{\rdim} \times \rdim$ with $\rdim \ge \hat{\rdim}$, and
  denoting $\shifts[w'] = \shifts[t] + (3\lambda,\ldots,3\lambda)$, $\mat{P}_2$
  has its row degree bounded by $\rdeg[\shifts]{\mat{P}_2} = \shifts[w']$, with
  $\sshifts[{\shifts[w']}] = \sshifts[{\rdeg[\shifts]{\mat{P}_2}}] \le
  \sshifts[{\rdeg[\shifts]{\mat{P}}}] \le 7\xi$. Besides,
  $\sshifts[{\rdeg[{\shifts[w']}]{\mat{N}_2 \mat{N}_1}}] \le
  \sshifts[{\rdeg[{\shifts[t]}]{\mat{N}_2 \mat{N}_1}}] +
  3(\hat{\rdim}-\cdim)\lambda \le \sshifts[{\shifts[v]}] + 3 \cdim \lambda/2
  \le 4\xi$.  Then, the product $\mat{N}_2\mat{N}_1\mat{P}_2$ can be computed
  with \Cref{algo:unbalanced-polmatmul} using $\bigO{ \rdim / \hat{\rdim}
  \polmatmultimePrime{\hat{\rdim},\xi/\hat{\rdim}} } \subseteq \bigO{
    \polmatmultimePrime{ \cdim, \xi/\cdim } }$ operations, since $\rdim \in
    \bigO{\cdim}$ and $\cdim \le \hat{\rdim}$.

  Thus, we have two recursive calls with half the column dimension and the same
  bound $\xi$, and additional $\bigO{ \polmatmultimePrime{\cdim,\xi/\cdim} +
  \polmatmultimeBis{\cdim,\xi/\cdim} }$ operations for the matrix products and
  the computation of a minimal basis of Hermite-Pad\'e approximants. Overall
  \Cref{algo:mnb} uses $\bigO{ \polmatmultimePrimeDnc{\cdim,\xi/\cdim}
  + \polmatmultimeBisDnc{\cdim,\xi/\cdim} }$ operations: since $\cdim \in
  \Theta(\rdim)$, we obtain the announced cost estimate; the upper bound is
  a direct consequence of \Cref{lem:polmatmulDnc-bound}.
\end{proof}


\bibliographystyle{elsart-harv}

\begin{thebibliography}{60}
\expandafter\ifx\csname natexlab\endcsname\relax\def\natexlab#1{#1}\fi
\expandafter\ifx\csname url\endcsname\relax
  \def\url#1{\texttt{#1}}\fi
\expandafter\ifx\csname urlprefix\endcsname\relax\def\urlprefix{URL }\fi

\bibitem[{Alekhnovich(2002)}]{Alekhnovich02}
Alekhnovich, M., 2002. Linear {D}iophantine equations over polynomials and soft
  decoding of {R}eed-{S}olomon codes. In: FOCS'02. IEEE, pp. 439--448.
\newline\urlprefix\url{http://dx.doi.org/10.1109/SFCS.2002.1181968}

\bibitem[{Alekhnovich(2005)}]{Alekhnovich05}
Alekhnovich, M., Jul. 2005. Linear {D}iophantine equations over polynomials and
  soft decoding of {R}eed-{S}olomon codes. IEEE Trans. Inf. Theory 51~(7),
  2257--2265.
\newline\urlprefix\url{http://dx.doi.org/10.1109/TIT.2005.850097}

\bibitem[{Beckermann(1990)}]{Beckermann90}
Beckermann, B., 1990. {Zur Interpolation mit polynomialen Linearkombinationen
  beliebiger Funktionen}. Ph.D. thesis, Department of Applied Mathematics,
  University of Hannover, Germany.

\bibitem[{Beckermann(1992)}]{Beckermann92}
Beckermann, B., 1992. A reliable method for computing {M}-{P}ad{\'e}
  approximants on arbitrary staircases. J. Comput. Appl. Math. 40~(1), 19--42.
\newline\urlprefix\url{http://dx.doi.org/10.1016/0377-0427(92)90039-Z}

\bibitem[{Beckermann and Labahn(1994)}]{BecLab94}
Beckermann, B., Labahn, G., Jul. 1994. A uniform approach for the fast
  computation of matrix-type {P}ad\'e approximants. SIAM J. Matrix Anal. Appl.
  15~(3), 804--823.
\newline\urlprefix\url{http://dx.doi.org/10.1137/S0895479892230031}

\bibitem[{Beckermann and Labahn(2000)}]{BecLab00}
Beckermann, B., Labahn, G., 2000. Fraction-free computation of matrix rational
  interpolants and matrix gcds. SIAM J. Matrix Anal. Appl. 22~(1), 114--144.
\newline\urlprefix\url{http://dx.doi.org/10.1137/S0895479897326912}

\bibitem[{Beckermann et~al.(2006)Beckermann, Labahn, and Villard}]{BeLaVi06}
Beckermann, B., Labahn, G., Villard, G., 2006. Normal forms for general
  polynomial matrices. J. Symbolic Comput. 41~(6), 708--737.
\newline\urlprefix\url{http://dx.doi.org/10.1016/j.jsc.2006.02.001}

\bibitem[{Beelen and Brander(2010)}]{BeeBra10}
Beelen, P., Brander, K., 2010. Key equations for list decoding of
  {R}eed-{S}olomon codes and how to solve them. J. Symbolic Comput. 45~(7),
  773--786.
\newline\urlprefix\url{http://dx.doi.org/10.1016/j.jsc.2010.03.010}

\bibitem[{Beelen et~al.(1988)Beelen, van~den Hurk, and Praagman}]{BeHuPr88}
Beelen, T. G.~J., van~den Hurk, G.~J., Praagman, C., 1988. A new method for
  computing a column reduced polynomial matrix. Systems and Control Letters
  10~(4), 217--224.
\newline\urlprefix\url{http://dx.doi.org/10.1016/0167-6911(88)90010-2}

\bibitem[{Bernstein(2011)}]{Bernstein11}
Bernstein, D.~J., 2011. Simplified high-speed high-distance list decoding for
  alternant codes. In: PQCrypto'11. Vol. 7071 of LNCS. Springer, pp. 200--216.
\newline\urlprefix\url{http://dx.doi.org/10.1007/978-3-642-25405-5_13}

\bibitem[{Bostan et~al.(2008)Bostan, Jeannerod, and Schost}]{BoJeSc08}
Bostan, A., Jeannerod, C.-P., Schost, E., 2008. Solving structured linear
  systems with large displacement rank. Theor. Comput. Sci. 407~(1-3),
  155--181.
\newline\urlprefix\url{http://dx.doi.org/10.1016/j.tcs.2008.05.014}

\bibitem[{Brander(2010)}]{KB10}
Brander, K., 2010. Interpolation and list decoding of algebraic codes. Ph.D.
  thesis, Technical University of Denmark.

\bibitem[{Brent et~al.(1980)Brent, Gustavson, and Yun}]{BrGuYu80}
Brent, R.~P., Gustavson, F.~G., Yun, D. Y.~Y., 1980. Fast solution of
  {Toeplitz} systems of equations and computation of {Pad\'e} approximants. J.
  of Algorithms 1~(3), 259--295.
\newline\urlprefix\url{http://dx.doi.org/10.1016/0196-6774(80)90013-9}

\bibitem[{Busse(2008)}]{PB08}
Busse, P., 2008. Multivariate list decoding of evaluation codes with a
  {G}r\"obner basis perspective. Ph.D. thesis, University of Kentucky.

\bibitem[{Cantor and Kaltofen(1991)}]{CanKal91}
Cantor, D.~G., Kaltofen, E., 1991. On fast multiplication of polynomials over
  arbitrary algebras. Acta Inform. 28~(7), 693--701.
\newline\urlprefix\url{http://dx.doi.org/10.1007/BF01178683}

\bibitem[{Chowdhury et~al.(2015)Chowdhury, Jeannerod, Neiger, Schost, and
  Villard}]{CJNSV15}
Chowdhury, M., Jeannerod, C.-P., Neiger, V., Schost, E., Villard, G., 2015.
  Faster algorithms for multivariate interpolation with multiplicities and
  simultaneous polynomial approximations. IEEE Trans. Inf. Theory 61~(5),
  2370--2387.
\newline\urlprefix\url{http://dx.doi.org/10.1109/TIT.2015.2416068}

\bibitem[{Cohn and Heninger(2012-2013)}]{CohHen12}
Cohn, H., Heninger, N., 2012-2013. Approximate common divisors via lattices.
  In: Tenth Algorithmic Number Theory Symposium. Mathematical Sciences
  Publishers (MSP), pp. 271--293.
\newline\urlprefix\url{http://dx.doi.org/10.2140/obs.2013.1.271}

\bibitem[{Cohn and Heninger(2015)}]{CohHen15}
Cohn, H., Heninger, N., 2015. Ideal forms of {Coppersmith}'s theorem and
  {Guruswami}-{Sudan} list decoding. Advances in Mathematics of Communications
  9~(3), 311--339.
\newline\urlprefix\url{http://dx.doi.org/10.3934/amc.2015.9.311}

\bibitem[{Coppersmith and Winograd(1990)}]{CopWin90}
Coppersmith, D., Winograd, S., 1990. Matrix multiplication via arithmetic
  progressions. J. Symbolic Comput. 9~(3), 251--280.
\newline\urlprefix\url{http://dx.doi.org/10.1016/S0747-7171(08)80013-2}

\bibitem[{Devet et~al.(2012)Devet, Goldberg, and Heninger}]{DeGoHe12}
Devet, C., Goldberg, I., Heninger, N., 2012. Optimally robust private
  information retrieval. In: USENIX Security 12. USENIX, pp. 269--283.

\bibitem[{Dummit and Foote(2004)}]{DumFoo04}
Dummit, D.~S., Foote, R.~M., 2004. Abstract Algebra. John Wiley \& Sons.

\bibitem[{\gathen{von zur} Gathen and Gerhard(2013)}]{vzGathen13}
\gathen{von zur} Gathen, J., Gerhard, J., 2013. Modern Computer Algebra (third
  edition). Cambridge University Press.

\bibitem[{Giorgi et~al.(2003)Giorgi, Jeannerod, and Villard}]{GiJeVi03}
Giorgi, P., Jeannerod, C.-P., Villard, G., 2003. On the complexity of
  polynomial matrix computations. In: ISSAC'03. ACM, pp. 135--142.
\newline\urlprefix\url{http://dx.doi.org/10.1145/860854.860889}

\bibitem[{Gupta et~al.(2012)Gupta, Sarkar, Storjohann, and
  Valeriote}]{GuSaStVa12}
Gupta, S., Sarkar, S., Storjohann, A., Valeriote, J., 2012. Triangular
  $x$-basis decompositions and derandomization of linear algebra algorithms
  over ${K}[x]$. J. Symbolic Comput. 47~(4), 422--453.
\newline\urlprefix\url{http://dx.doi.org/10.1016/j.jsc.2011.09.006}

\bibitem[{Guruswami and Rudra(2008)}]{GurRud08}
Guruswami, V., Rudra, A., 2008. Explicit codes achieving list decoding
  capacity: Error-correction with optimal redundancy. IEEE Trans. Inf. Theory
  54~(1), 135--150.
\newline\urlprefix\url{http://dx.doi.org/10.1109/TIT.2007.911222}

\bibitem[{Guruswami and Sudan(1998)}]{GurSud98}
Guruswami, V., Sudan, M., Nov. 1998. Improved decoding of {{R}eed-{S}olomon}
  and algebraic-geometric codes. In: FOCS'98. pp. 28--39.
\newline\urlprefix\url{http://dx.doi.org/10.1109/SFCS.1998.743426}

\bibitem[{Guruswami and Sudan(1999)}]{GurSud99}
Guruswami, V., Sudan, M., 1999. Improved decoding of {{R}eed-{S}olomon} and
  algebraic-geometry codes. IEEE Trans. Inf. Theory 45~(6), 1757--1767.
\newline\urlprefix\url{http://dx.doi.org/10.1109/18.782097}

\bibitem[{Kailath(1980)}]{Kailath80}
Kailath, T., 1980. {Linear Systems}. Prentice-Hall.

\bibitem[{Keller-Gehrig(1985)}]{KelGeh85}
Keller-Gehrig, W., 1985. Fast algorithms for the characteristic polynomial.
  Theoretical Computer Science 36, 309--317.
\newline\urlprefix\url{http://dx.doi.org/10.1016/0304-3975(85)90049-0}

\bibitem[{Knuth(1970)}]{Knu70}
Knuth, D.~E., 1970. The analysis of algorithms. In: {Congr\`es int. Math.,
  Nice, France}. Vol.~3. pp. 269--274.

\bibitem[{K{\"o}tter(1996)}]{Koetter96}
K{\"o}tter, R., 1996. Fast generalized minimum-distance decoding of
  algebraic-geometry and {R}eed-{S}olomon codes. IEEE Trans. Inf. Theory
  42~(3), 721--737.
\newline\urlprefix\url{http://dx.doi.org/10.1109/18.490540}

\bibitem[{K{\"o}tter and Vardy(2003{\natexlab{a}})}]{KoeVar03a}
K{\"o}tter, R., Vardy, A., 2003{\natexlab{a}}. Algebraic soft-decision decoding
  of {{R}eed-{S}olomon} codes. IEEE Trans. Inf. Theory 49~(11), 2809--2825.
\newline\urlprefix\url{http://dx.doi.org/10.1109/TIT.2003.819332}

\bibitem[{K{\"o}tter and Vardy(2003{\natexlab{b}})}]{KoeVar03b}
K{\"o}tter, R., Vardy, A., 2003{\natexlab{b}}. A complexity reducing
  transformation in algebraic list decoding of {R}eed-{S}olomon codes. In:
  ITW2003. IEEE, pp. 10--13.
\newline\urlprefix\url{http://dx.doi.org/10.1109/ITW.2003.1216682}

\bibitem[{Le~Gall(2014)}]{LeGall14}
Le~Gall, F., 2014. Powers of tensors and fast matrix multiplication. In:
  ISSAC'14. ACM, pp. 296--303.
\newline\urlprefix\url{http://dx.doi.org/10.1145/2608628.2608664}

\bibitem[{Lecerf(2001)}]{Lecerf01}
Lecerf, G., 2001. Private communication.

\bibitem[{Lee and O'Sullivan(2006)}]{LeeOSul06}
Lee, K., O'Sullivan, M., July 2006. An interpolation algorithm using
  {G}r\"obner bases for soft-decision decoding of {R}eed-{S}olomon codes. In:
  2006 IEEE International Symposium on Information Theory. pp. 2032--2036.
\newline\urlprefix\url{http://dx.doi.org/10.1109/ISIT.2006.261906}

\bibitem[{Lee and O'Sullivan(2008)}]{LeeOSul08}
Lee, K., O'Sullivan, M.~E., 2008. List decoding of {R}eed-{S}olomon codes from
  a {G}r\"obner basis perspective. J. Symbolic Comput. 43~(9), 645--658.
\newline\urlprefix\url{http://dx.doi.org/10.1016/j.jsc.2008.01.002}

\bibitem[{L\"ubbe(1983)}]{Lubbe83}
L\"ubbe, W., 1983. {\"U}ber ein allgemeines {I}nterpolationsproblem --- lineare
  {I}dentit\"aten zwischen benachbarten {L}\"osungssystemen. Ph.D. thesis,
  Department of Applied Mathematics, University of Hannover, Germany.

\bibitem[{Mahler(1968)}]{Mahler68}
Mahler, K., 1968. Perfect systems. Composit. Math. 19~(2), 95--166.

\bibitem[{McEliece(2003)}]{McEliece03}
McEliece, R.~J., 2003. The {G}uruswami-{S}udan decoding algorithm for
  {R}eed-{S}olomon codes. IPN Progress Report 42-153.

\bibitem[{Mulders and Storjohann(2003)}]{MulSto03}
Mulders, T., Storjohann, A., 2003. On lattice reduction for polynomial
  matrices. J. Symb. Comput. 35, 377--401.
\newline\urlprefix\url{http://dx.doi.org/10.1016/S0747-7171(02)00139-6}

\bibitem[{Nielsen(2014)}]{Nielsen14}
Nielsen, J. S.~R., 2014. Fast {K}\"otter-{N}ielsen-{H}{\o}holdt interpolation
  in the {G}uruswami-{S}udan algorithm. In: ACCT'14.
\newline\urlprefix\url{http://arxiv.org/abs/1406.0053}

\bibitem[{Nielsen and {H{\o}holdt}(2000)}]{HohNie00}
Nielsen, R.~R., {H{\o}holdt}, T., 2000. Decoding {R}eed-{S}olomon codes beyond
  half the minimum distance. In: Coding Theory, Cryptography and Related Areas.
  Springer, pp. 221--236.
\newline\urlprefix\url{http://dx.doi.org/10.1007/978-3-642-57189-3_20}

\bibitem[{Olshevsky and Shokrollahi(1999)}]{OlsSho99}
Olshevsky, V., Shokrollahi, M.~A., 1999. A displacement approach to efficient
  decoding of algebraic-geometric codes. In: STOC'99. ACM, pp. 235--244.
\newline\urlprefix\url{http://doi.acm.org/10.1145/301250.301311}

\bibitem[{Parvaresh and Vardy(2005)}]{ParVar05}
Parvaresh, F., Vardy, A., 2005. Correcting errors beyond the
  {G}uruswami-{S}udan radius in polynomial time. In: FOCS'05. IEEE, pp.
  285--294.
\newline\urlprefix\url{http://dx.doi.org/10.1109/SFCS.2005.29}

\bibitem[{Paszkowski(1987)}]{Paszkowski87}
Paszkowski, S., Jul. 1987. Recurrence relations in {P}ad{\'e}-{H}ermite
  approximation. J. Comput. Appl. Math. 19~(1), 99--107.
\newline\urlprefix\url{http://dx.doi.org/10.1016/0377-0427(87)90177-4}

\bibitem[{Roth and Ruckenstein(2000)}]{RotRuc00}
Roth, R.~M., Ruckenstein, G., 2000. Efficient decoding of {R}eed-{S}olomon
  codes beyond half the minimum distance. IEEE Trans. Inf. Theory 46~(1),
  246--257.
\newline\urlprefix\url{http://dx.doi.org/10.1109/18.817522}

\bibitem[{Sch{\"o}nhage(1971)}]{Sch71}
Sch{\"o}nhage, A., 1971. Schnelle {Berechnung} {von}
  {Kettenbruchentwicklungen}. Acta Inform. 1, 139--144, in German.
\newline\urlprefix\url{http://dx.doi.org/10.1007/BF00289520}

\bibitem[{Sergeyev(1987)}]{Sergeyev87}
Sergeyev, A.~V., Jul. 1987. A recursive algorithm for {Pad{\'e}}-{Hermite}
  approximations. USSR Comput. Math. Math. Phys. 26~(2), 17--22.
\newline\urlprefix\url{http://dx.doi.org/10.1016/0041-5553(86)90003-0}

\bibitem[{Storjohann(2000)}]{Storjohann00}
Storjohann, A., 2000. Algorithms for matrix canonical forms. Ph.D. thesis,
  Swiss Federal Institute of Technology -- ETH.

\bibitem[{Storjohann(2006)}]{Storjohann06}
Storjohann, A., 2006. Notes on computing minimal approximant bases. In:
  Challenges in Symbolic Computation Software. Dagstuhl Seminar Proceedings.
\newline\urlprefix\url{http://drops.dagstuhl.de/opus/volltexte/2006/776}

\bibitem[{Sudan(1997)}]{Sudan97}
Sudan, M., 1997. Decoding of {R}eed-{S}olomon codes beyond the error-correction
  bound. J. Complexity 13~(1), 180--193.
\newline\urlprefix\url{http://dx.doi.org/10.1006/jcom.1997.0439}

\bibitem[{Van~Barel and Bultheel(1991)}]{BarBul91}
Van~Barel, M., Bultheel, A., 1991. The computation of non-perfect
  {P}ad{\'e}-{H}ermite approximants. Numer. Algorithms 1~(3), 285--304.
\newline\urlprefix\url{http://dx.doi.org/10.1007/BF02142327}

\bibitem[{Van~Barel and Bultheel(1992)}]{BarBul92}
Van~Barel, M., Bultheel, A., 1992. A general module theoretic framework for
  vector {M-Pad\'e} and matrix rational interpolation. Numer. Algorithms 3,
  451--462.
\newline\urlprefix\url{http://dx.doi.org/10.1007/BF02141952}

\bibitem[{Welch and Berlekamp(1986)}]{WelBer86}
Welch, L.~R., Berlekamp, E.~R., Dec.~30 1986. Error correction for algebraic
  block codes. US Patent 4,633,470.

\bibitem[{Zeh(2013)}]{Zeh13}
Zeh, A., 2013. {Algebraic Soft- and Hard-Decision Decoding of Generalized
  Reed--Solomon and Cyclic Codes}. Ph.D. thesis, {\'Ecole Polytechnique}.

\bibitem[{Zeh et~al.(2011)Zeh, Gentner, and Augot}]{ZeGeAu11}
Zeh, A., Gentner, C., Augot, D., 2011. An interpolation procedure for list
  decoding {R}eed-{S}olomon codes based on generalized key equations. IEEE
  Trans. Inf. Theory 57~(9), 5946--5959.
\newline\urlprefix\url{http://dx.doi.org/10.1109/TIT.2011.2162160}

\bibitem[{Zhou(2012)}]{Zhou12}
Zhou, W., 2012. Fast order basis and kernel basis computation and related
  problems. Ph.D. thesis, University of Waterloo.

\bibitem[{Zhou and Labahn(2012)}]{ZhoLab12}
Zhou, W., Labahn, G., 2012. Efficient algorithms for order basis computation.
  J. Symbolic Comput. 47~(7), 793--819.
\newline\urlprefix\url{http://dx.doi.org/10.1016/j.jsc.2011.12.009}

\bibitem[{Zhou et~al.(2012)Zhou, Labahn, and Storjohann}]{ZhLaSt12}
Zhou, W., Labahn, G., Storjohann, A., 2012. Computing minimal nullspace bases.
  In: ISSAC'12. ACM, pp. 366--373.
\newline\urlprefix\url{http://dx.doi.org/10.1145/2442829.2442881}

\end{thebibliography}

\end{document}